\documentclass[aps,prl,superscriptaddress,reprint,amsmath,amssymb,longbibliography]{revtex4-2}

\usepackage{graphicx}
\usepackage{tabularx}
\usepackage[T1]{fontenc}
\usepackage[utf8]{inputenc}
\usepackage{mathrsfs}
\usepackage{bm}
\usepackage{amsthm}
\usepackage{float}
\usepackage{colortbl}
\usepackage[dvipsnames]{xcolor}
\usepackage{hyperref}
\hypersetup{colorlinks=true,
            citecolor=blue,
            linkcolor=blue,
            filecolor=blue,
            urlcolor=blue,
            breaklinks=true}

\newcommand{\nc}{\newcommand}
\nc{\rnc}{\renewcommand}

\nc{\bra}[1]{\langle#1|}
\nc{\ket}[1]{|#1\rangle}
\nc{\<}{\langle}
\rnc{\>}{\rangle}
\nc{\ketbra}[2]{|#1\rangle\!\langle#2|}
\nc{\braket}[2]{\langle#1|#2\rangle}
\nc{\braandket}[3]{\langle #1|#2|#3\rangle}
\nc{\proj}[1]{| #1\rangle\!\langle #1 |}
\nc{\avg}[1]{\langle#1\rangle}

\nc{\tr}{\operatorname{Tr}}
\nc{\ox}{\otimes}
\nc{\id}{I}
\def\ve{\varepsilon}

\nc{\plus}{{\scalebox{0.7}{+}}}
\nc{\HERM}{\mathscr{H}}
\nc{\PSD}{\HERM_{\plus}}
\nc{\PD}{\HERM_{\plus\plus}}
\nc{\density}{\mathscr{D}}
\nc{\conv}{\operatorname{conv}}
\nc{\aff}{\operatorname{aff}}
\nc{\cone}{\operatorname{cone}}
\nc{\cl}{\operatorname{cl}}
\nc{\diam}{\operatorname{diam}}
\nc{\CPTP}{\text{\rm CPTP}}
\nc{\supp}{{\operatorname{supp}}}
\nc{\polarPSD}[1]{{#1}_{\plus}^{\circ}}

\definecolor{shadecolor}{rgb}{0.9,0.9,0.9}
\definecolor{mylightgray}{RGB}{100,100,100}
\definecolor{myblue}{RGB}{0, 68, 116}
\definecolor{mycyan}{RGB}{0, 97, 91}
\definecolor{mygreen}{RGB}{2, 102, 1}
\definecolor{myorange}{RGB}{240, 102, 0}
\definecolor{myred}{RGB}{172, 23, 0}
\definecolor{mymagenta}{RGB}{140,16,73}

\theoremstyle{plain}
\newtheorem{theorem}{Theorem}
\newtheorem{proposition}[theorem]{Proposition}
\newtheorem{lemma}[theorem]{Lemma}
\newtheorem{corollary}[theorem]{Corollary}

\theoremstyle{definition}
\newtheorem{definition}{Definition}
\newtheorem*{remark}{Remark}
\newtheorem*{example}{Example}
\AtEndEnvironment{example}{\qed}

\nc{\cE}{{\cal E}}
\nc{\cF}{{\cal F}}
\nc{\cH}{{\cal H}}
\nc{\cK}{{\cal K}}
\nc{\cL}{{\cal L}}
\nc{\cP}{{\cal P}}

\nc{\BB}{\mathbb{B}}
\nc{\CC}{\mathbb{C}}
\nc{\NN}{\mathbb{N}}
\nc{\RR}{\mathbb{R}}
\nc{\WW}{\mathbb{W}}

\nc{\sB}{{{\mathscr{B}}}}
\nc{\sC}{{{\mathscr{C}}}}
\nc{\sE}{{{\mathscr{E}}}}
\nc{\sK}{{{\mathscr{K}}}}
\nc{\sL}{{{\mathscr{L}}}}
\nc{\sP}{{{\mathscr{P}}}}

\nc{\fF}{{\mathfrak{F}}}
\nc{\fL}{{\mathfrak{L}}}

\nc{\GPO}{\text{\rm GPO}}
\nc{\GPL}{\text{\rm GPL}}
\nc{\TO}{\text{\rm TO}}
\nc{\GPC}{\text{\rm GPC}}

\makeatletter
\newcommand*\rel@kern[1]{\kern#1\dimexpr\macc@kerna}
\newcommand*\widebar[1]{%
  \begingroup
  \def\mathaccent##1##2{%
    \rel@kern{0.8}%
    \overline{\rel@kern{-0.8}\macc@nucleus\rel@kern{0.2}}%
    \rel@kern{-0.2}%
  }%
  \macc@depth\@ne
  \let\math@bgroup\@empty \let\math@egroup\macc@set@skewchar
  \mathsurround\z@ \frozen@everymath{\mathgroup\macc@group\relax}%
  \macc@set@skewchar\relax
  \let\mathaccentV\macc@nested@a
  \macc@nested@a\relax111{#1}%
  \endgroup
}
\makeatother

\usepackage[pass]{geometry}

\nc{\restateheading}{}
\newtheorem*{restateinner}{\restateheading}

\newcounter{mainresult} 
\nc{\mainresultnumber}{}
\rnc{\themainresult}{\mainresultnumber}

\NewDocumentEnvironment{restate}{m m o}{%
  \rnc{\mainresultnumber}{#2}%
  \refstepcounter{mainresult}%
  \IfNoValueTF{#3}{%
    \rnc{\restateheading}{Restatement of #1~#2}%
  }{%
    \rnc{\restateheading}{Restatement of #1~#2 {\normalfont(#3)}}%
  }%
  \restateinner
}{%
  \endrestateinner
}

\nc{\cuscol}[2]{\parbox[t]{#1}{#2}}

\nc{\KF}[1]{\textcolor{magenta}{[KF: #1]}}
\nc{\MN}[1]{\textcolor{violet}{[MNZ: #1]}}

\begin{document}
\title{Quantum thermodynamics with uncertain equilibrium}

\author{Munan Zhang}
\affiliation{School of Data Science, The Chinese University of Hong Kong, Shenzhen, Guangdong, 518172, China}

\author{Kun Fang}
\email{kunfang@cuhk.edu.cn}
\affiliation{School of Data Science, The Chinese University of Hong Kong, Shenzhen, Guangdong, 518172, China}

\date{\today}

\begin{abstract}
The resource-theoretic approach to quantum thermodynamics typically assumes perfect knowledge of the thermal equilibrium state, an idealization incompatible with finite experimental precision. We develop a framework for equilibrium uncertainty by representing the equilibrium reference as a set of candidate states. Under a generic geometric condition, we prove a no-go theorem that sharply limits athermality ``purification'': converting an uncertain athermal state into a definite target is either trivial or impossible. We then derive exact one-shot entropic characterizations of work extraction and formation for two work-storage models, a clean battery with known equilibrium and a dirty battery with uncertain equilibrium. Both models exhibit strong asymptotic irreversibility even under arbitrarily small uncertainty. An explicit example reveals two distinct extremes: clean batteries display a bound-entanglement-like phenomenon, with positive formation cost but zero extractable work, whereas dirty batteries allow positive work extraction but require infinite formation cost. These phenomena show that equilibrium uncertainty is not a minor perturbation of the standard theory, but a structural ingredient that fundamentally reshapes the limits of quantum thermodynamics.
\end{abstract}

\maketitle

\paragraph{Introduction.---} The resource-theoretic approach of quantum thermodynamics has provided a powerful viewpoint for understanding nonequilibrium phenomena, turning questions about heat, work and irreversibility into a unified framework~\cite{brandao2013resource,horodecki2013fundamental,lostaglio2019introductory}. Central to this framework is the notion of \emph{athermality}: any quantum state deviating from its thermal equilibrium (a.k.a.~Gibbs state) constitutes a thermodynamic resource~\cite{chitambar2019quantum}. This perspective has driven broad advances, including precise characterizations of state interconversion~\cite{horodecki2013fundamental,skrzypczyk2014work,gour2018quantum}, generalized second laws~\cite{cwiklinski2015limitations,brandao2015second}, and a deeper understanding of how genuinely quantum features, such as coherence and correlations~\cite{lostaglio2015description,lostaglio2015quantum,sapienza2019correlations,marvian2020coherence, marvian2022operational,gour2022role,tajima2025gibbspreserving}, shape thermodynamic behaviour. Recently, this framework has been extended to black-box settings, capturing the practical reality that nonequilibrium states are often only partially specified~\cite{safranek2023work,watanabe2024black,watanabe2026universal}. 

However, one basic assumption has remained unchallenged throughout: the equilibrium Gibbs state, which serves as the reference for thermodynamic resourcefulness, is taken to be exactly known. This idealization is difficult to justify in practice with finite experimental precision. The Gibbs state depends on both the system Hamiltonian and the bath temperature, neither of which can be determined perfectly~\cite{giovannetti2011advances,montenegro2025review}. The Hamiltonian must be inferred through calibration~\cite{gerster2022experimental,berritta2025efficient} or learning~\cite{huang2023learning,hangleiter2024robustly}, and may drift over time~\cite{proctor2020detecting,capannelli2025tracking,acharya2025quantum,bluvstein2026faulttolerant}. The bath temperature can be equally challenging to measure, especially in the low-temperature regime~\cite{mok2021optimal}. These considerations suggest that the relevant equilibrium reference should be represented not by a single state, but by a \emph{set of candidate Gibbs states} reflecting this uncertainty.

In this work, we develop the first framework of quantum thermodynamics with uncertain equilibrium that better reflects this experimental practice. We show that, perhaps  surprisingly, equilibrium uncertainty is not a minor perturbation of the standard theory, but a new structural ingredient that fundamentally reshapes the limits of quantum thermodynamics. In particular, we prove that, under a generic geometric condition, the linearity of quantum mechanics and the minimal thermodynamic requirement of Gibbs preservation alone preclude converting an uncertain athermality resource into any definite nontrivial target. Notably, such a conversion is either trivial or impossible, with no room for tradeoff.

This no-go theorem motivates a more flexible treatment of work transformations. We consider two work-storage models: a \emph{clean battery}, whose equilibrium state is known, and a \emph{dirty battery}, whose equilibrium state is uncertain. For both models, we derive exact one-shot entropic characterizations of work extraction, in which an athermality resource is converted into an excited battery, and work of formation, in which an excited battery is consumed to prepare an athermality resource. A notable feature is that the conventional correspondence between operational tasks and entropic quantities~\cite{liu2019oneshot} breaks down: standard relative entropies are insufficient for complete characterizations, necessitating new \emph{subspace-constrained relative entropies}. Specifically, in the clean-battery model, formation is characterized by the standard max-relative entropy, whereas extraction is governed by a subspace-constrained min-relative entropy. In the dirty-battery model, the roles of the standard and subspace-constrained quantities are reversed.

In the asymptotic regime, both battery models exhibit a strict separation between the rates of work extraction and work of formation, breaking the reversibility familiar from standard resource theory of athermality~\cite{brandao2013resource,gour2022role}. In particular, we give an explicit example in which the two models display qualitatively distinct forms of irreversibility, even under arbitrarily small equilibrium uncertainty.  For a clean battery, the work extraction rate vanishes while the formation rate remains positive, yielding a thermodynamic analogue of bound entanglement~\cite{horodecki1998mixedstate}. For a dirty battery, work can still be extracted at an nonzero rate, but formation requires infinite cost. These strong forms of irreversiblity show that equilibrium uncertainty fundamentally alters the structure of quantum thermodynamics, imposing stringent limitations on resource manipulation with no counterpart in the standard theory. We discuss our framework and results below, and defer all technical proofs to the Supplemental Material~\cite{zhang2026sm}.

\paragraph{Thermodynamics with uncertain equilibrium.---} The resource-theoretic framework formulates quantum thermodynamics as a theory of athermal state conversion under restricted, or free, operations. Throughout, we consider a finite-dimensional system $P$ with Hamiltonian $H^P$ in contact with a heat bath at inverse temperature $\beta$. Let $\density(P)$ denote the set of all density matrices on this system. Conventionally, an athermal state is specified by a pair $(\rho^P, \tau^P)$, where $\rho^P$ is the actual state of the system and $\tau^P := e^{-\beta H^P}/\tr[e^{-\beta H^P}]$ is the corresponding equilibrium Gibbs state~\cite{janzing2000thermodynamic,gour2022role}. In realistic settings, however, neither state need be exactly known. The nonequilibrium state may be accessible only as a black box, known to lie in a set of states $\sP\subseteq\density(P)$~\cite{watanabe2024black,watanabe2026universal}; similarly, imperfect knowledge of the system Hamiltonian or bath temperature may specify the equilibrium reference only through a set of candidate Gibbs states $\sE\subseteq\density(P)$. We therefore model an uncertain athermal state as a pair of sets $(\sP,\sE)$, representing the possible nonequilibrium and equilibrium states, respectively.

The most physically grounded model of thermodynamic free operations is the class of thermal operations (TO)~\cite{brandao2013resource,horodecki2013fundamental,janzing2000thermodynamic}. Let $\CPTP$ denote the set of all completely positive and trace-preserving maps. A map $\cE\in\CPTP(P\to P^\prime)$ is a thermal operation if it admits the form $\cE(\cdot) = \tr_{(PE)\setminus P^\prime}[U(\cdot \ox \tau^E)U^\dagger]$, where $E$ is an ancillary system initialized in its Gibbs state $\tau^E$, and $U$ is a joint energy-conserving unitary. By construction, every thermal operation preserves the Gibbs state, i.e., $\cE(\tau^P) = \tau^{P^\prime}$. This motivates the broader and technically useful class of Gibbs-preserving operations (GPO), consisting of all CPTP maps that preserve Gibbs states~\cite{faist2015gibbspreserving,faist2018fundamental,faist2019thermodynamic,shiraishi2021quantum,gour2022role,watanabe2024black,tajima2025gibbspreserving}.

To establish no-go limitations in maximal generality, we impose only the minimal structural requirements on the allowed transformations, so that any impossibility result automatically applies to more restrictive operational classes. We therefore introduce \emph{Gibbs-preserving linear maps} (GPL), defined as linear maps $\cL$ satisfying $\cL(\tau^P) = \tau^{P^\prime}$. This is the largest class of transformations compatible with the linear structure of quantum mechanics and the minimal thermodynamic requirement of Gibbs preservation, with strict inclusions $\TO \subsetneq \GPO \subsetneq \GPL$. This relaxation also encompasses generalized scenarios, including post-selected probabilistic protocols~\cite{alhambra2016fluctuating} and virtual operations combining quantum operations with classical statistical post-processing~\cite{takagi2024virtual,yuan2024virtual,zhu2024reversing}. Accordingly, all no-go results in this work are established for GPL and hence hold a fortiori for all smaller operational classes, whereas achievability results are constructed within conventional classes such as $\TO$ or $\GPO$.

We now formalize transformations between uncertain athermal states. Since the precise input $(\rho,\tau)\in\sP\times\sE$ is not known in advance, a valid transformation must be implemented by a single operation that works for all admissible candidates. Specifically, let $\sP,\sE\subseteq\density(P)$ and $\sP^\prime,\sE^\prime\subseteq\density(P^\prime)$. For an error tolerance $\ve\in[0,1)$ and a free operation class $\fF(P \to P^\prime)$, we write $(\sP,\sE) \xrightarrow[]{\cF,\; \ve} (\sP^\prime, \sE^\prime)$ if there exists an operation $\cF\in\fF$ such that, for every $(\rho,\tau) \in \sP\times\sE$, one can find a target pair $(\rho^\prime, \tau^\prime) \in \sP^\prime\times \sE^\prime$ satisfying $T(\cF(\rho), \rho^\prime) \leq \ve$ and $\cF(\tau) = \tau^\prime$. Here $T(\rho, \rho^\prime) := \tfrac{1}{2}\|\rho-\rho^\prime\|_1$ denotes the trace distance. When the uncertainty sets are singletons, this definition reduces to the standard notion of athermal state conversion.

\paragraph{No-go theorem for ``purification''---} Within this framework, a natural question is whether the athermality uncertainty can be purified: can an uncertain athermal state $(\sP,\sE)$ be converted into a definite target $(\rho^\prime, \tau^\prime)$ by a free operation? If this were possible for nontrivial targets, the uncertainty framework considered here would effectively collapse to the standard setting~\cite{horodecki2013fundamental,brandao2013resource,gour2022role}. The following theorem shows that, under a mild geometric condition, such a conversion is generically impossible. 

\begin{theorem}
    \label{thm:no-go theorem for purification}
    Let $\ve \in [0, 1)$, and let $(\rho^\prime, \tau^\prime) \in \density \times \density$. Suppose that $\sP, \sE \subseteq \density$ satisfy $\conv(\sP) \cap \aff(\sE) \neq \emptyset$, where $\conv(\sP)$ and $\aff(\sE)$ denote convex hull of $\sP$ and affine hull of $\sE$, respectively. Then $(\sP, \sE) \xrightarrow[]{\cL,\;\ve} (\rho^\prime, \tau^\prime)$ is achievable by some $\cL \in \GPL$ if and only if $\ve \geq T(\rho^\prime, \tau^\prime)$. 
\end{theorem}

This places a sharp limitation on athermality ``purification'': under the geometric condition, it is either trivial or impossible. If the target athermal state satisfies $T(\rho^\prime, \tau^\prime) \leq \ve$, it is already operationally $\ve$-equivalent to the free state $(\tau^\prime,\tau^\prime)$~\cite{zhang2026sm}. In this regime, the conversion is achievable by the trivial thermal operation $\cL(\cdot) = \tr[\cdot] \tau^\prime$, which simply discards the input and prepares $\tau^\prime$. Conversely, any target with $T(\rho^\prime,\tau^\prime) > \ve$ cannot be obtained from an uncertain input satisfying the geometric condition. In particular, since $\ve \geq T(\rho^\prime,\tau^\prime) > 0$ holds for any non-free target, exact purification is strictly impossible.

Remarkably, this no-go already emerges at the level of GPL, which imposes only linearity and Gibbs preservation. It is therefore not an artifact of any particular framework, but an unavoidable consequence of the quantum-mechanical structure of thermodynamics under equilibrium uncertainty. It is also conceptually distinct from existing no-go purifications, where impossibility is driven by free components in the resource~\cite{fang2020nogo,fang2022nogo} or by the purity of the target~\cite{liu2025no}, and relies on additional assumptions such as positivity of the map.

The geometric condition $\conv(\sP) \cap \aff(\sE) \neq \emptyset$ may appear technical at first sight, but it can be very generic in practice: it is satisfied whenever arbitrarily small perturbations of the Hamiltonian are present. Even if the nominal Hamiltonian is known, unavoidable fluctuations or incomplete knowledge of control parameters can yield a family of Gibbs states whose affine hull spans the full state space, $\aff(\sE)=\aff(\density)$, making the condition automatic. An illustrative example is given in~\cite{zhang2026sm}. Moreover, for exact conversion ($\ve = 0$), the condition can be relaxed to $\aff(\sP) \cap \aff(\sE) \neq \emptyset$.

\paragraph{Batteries and work transformations.---} A central objective in quantum thermodynamics is to characterize the efficiency of work transformations~\cite{horodecki2013fundamental,brandao2013resource,gour2022role,watanabe2024black,watanabe2026universal}. Following Refs.~\cite{gour2022role,watanabe2024black}, we quantify work using a two-level battery system $B$ with Hamiltonian $H^B_M = E_{M,0} \ket{0}\bra{0} + E_{M,1} \ket{1}\bra{1}$, where the parameter $M > 1$ determines the energy gap through $E_{M,1} - E_{M,0} = \frac{1}{\beta}\log(M-1)$. The corresponding Gibbs state is $\pi_M = \left(1-\frac{1}{M}\right) \ket{0}\bra{0} + \frac{1}{M}\ket{1}\bra{1}$. We call $B_M := (\ket{1}\bra{1}, \pi_M)$ a \emph{clean battery}: an excited battery whose equilibrium state is specified exactly. In practice, however, the battery system may be imperfectly calibrated or subject to parameter drift. This motivates the notion of a \emph{dirty battery}, whose equilibrium state is uncertain and described by a set of candidate Gibbs states. In particular, we define $\widebar{B}_M:=(\ket{1}\bra{1}, \Pi_M)$, where $\Pi_M := \{\pi_{M^\prime}: M^\prime \in [M, \infty)\}$. This one-sided uncertainty reflects a worst-case certification of the battery's work capacity. We consider two canonical work-transformation tasks: work extraction, where an athermal state is converted into an excited battery, and work of formation, where an excited battery is consumed to prepare a target athermal state.

Notably, the no-go result in Theorem~\ref{thm:no-go theorem for purification} has two immediate consequences for clean-battery work extraction.

\begin{corollary}\label{cor:no-go for work extraction to clean battery}
    Let $\ve \in [0,1)$. Suppose that $\sP, \sE \subseteq \density$ satisfy $\conv(\sP) \cap \aff(\sE) \neq \emptyset$. Then $(\sP, \sE) \xrightarrow[]{\cL,\; \ve} B_M$ is achievable by some $\cL \in \GPL$ if and only if $\ve \geq 1 -1/M$. 
    Moreover, whenever the transformation is achievable, the optimal work extraction is trivially realized by the thermal operation $\cL(\cdot) = \tr[\cdot] \pi_{1/(1-\ve)}$.
\end{corollary}

This result shows that the usual work-error tradeoff collapses to the single threshold $\ve = 1 - 1/M$: the work extraction is trivial above this threshold and impossible below it. Since the achievable regime is already realized by a thermal operation, every free-operation class between $\TO$ and $\GPL$ yields the same trivial extractable work under the geometric condition. Equilibrium uncertainty therefore eliminates the operational distinctions among these classes in the standard theory~\cite{horodecki2013fundamental,gour2022role,watanabe2024black}.

Conventionally, a battery with larger capacity is always at least as useful as one with smaller capacity, since surplus energy can simply be discarded via a free truncation operation~\cite{zhang2026sm}. The next corollary shows that equilibrium uncertainty can destroy this monotonicity completely.

\begin{corollary}\label{cor:no universal energy truncation}
    Let $\ve\in[0,1)$ and $M_2 > M_1 \geq N > 1$. Then $(\ket{1}\bra{1}, \{\pi_{M_1}, \pi_{M_2}\}) \xrightarrow[]{\cL,\;\ve} B_N$ is achievable by some $\cL \in \GPL$ if and only if $\ve \geq 1-1/N$.
\end{corollary}

This result identifies a structural obstruction caused specifically by equilibrium uncertainty. Even when the nonequilibrium state is known exactly, an \emph{arbitrarily small} uncertainty in the equilibrium reference can preclude energy truncation, despite every candidate input being individually more resourceful than the target. Moreover, regardless of how large the candidate battery capacities $M_1, M_2$ are, or how modest the target capacity $N$ is, the conversion is either impossible or trivially achievable. This sharply contrasts with the black-box setting in~\cite{watanabe2024black}, where uncertainty confined to the nonequilibrium state alone still permits nontrivial work extraction. The comparison reveals a fundamental asymmetry between the two sources of uncertainty in quantum thermodynamics: equilibrium uncertainty can be far more detrimental than nonequilibrium uncertainty. This phenomenon also differs from entanglement theory, where universal entanglement truncation is possible~\cite{matsumoto2007universal}, highlighting a structural distinction between the resource theories of athermality and entanglement.

Besides the no-go results above, we derive exact entropic characterizations of work extraction and formation for general uncertainty sets $\sP$ and $\sE$, under both the clean- and dirty-battery models. For $\ve\in[0,1)$, the one-shot extractable work and work cost under a free-operation class $\fF$ are defined as
\begin{align*}
    \beta\WW_{\fF,\ve}(\sP,\sE) & := \log\sup_{\cF\in\fF}\left\{M:(\sP,\sE)\xrightarrow[]{\cF,\,\ve} \BB_M\right\}, \\
    \beta\CC_{\fF,\ve}(\sP,\sE) & := \log\inf_{\cF\in\fF}\left\{M:\BB_M\xrightarrow[]{\cF,\,\ve} (\sP,\sE)\right\}, 
\end{align*}
respectively, where $(\WW, \CC, \BB_M)=(W, C, B_M)$ for clean-battery model and $(\widebar{W}, \widebar{C}, \widebar{B}_M)$ for dirty-battery model. 

\paragraph{Clean-battery work transformations.---} To characterize the extractable work in this model, we introduce the \emph{subspace-constrained min-relative entropy} between $\sP$ and $\sE$ with respect to a set of linear operators $\sK$,
\begin{align*}
    D_{\min,\ve}^{\sK}(\sP\|\sE) := -\log \min_{\substack{0\leq E \leq I \\ E \perp V(\sK) \\ \sup_{\rho \in \sP} \tr[(I-E)\rho] \leq \ve}} \sup_{\tau \in \sE}\tr[E \tau],
\end{align*}
where $V(\sK) := \operatorname{span}\{\tau - \tau^\prime : \tau, \tau^\prime \in \sK\}$ is the subspace spanned by the differences of elements in $\sK$~\footnote{An equivalent description is $V(\sK) = \aff(\sK) - \tau_0$ for any fixed $\tau_0 \in \sK$.} and $E \perp V(\sK)$ means $\tr[E X] = 0$ for all $X \in V(\sK)$. 

When $\sK$ is a singleton, the orthogonality constraint is vacuous, and one recovers the standard smoothed min-relative entropy $D_{\min,\ve}(\sP\|\sE)$~\cite{watanabe2024black,fang2025generalized}. The most relevant case here is $\sK = \sE$, for which the quantity admits a natural interpretation as a \emph{constrained} hypothesis testing. In standard hypothesis testing between $\sP$ and $\sE$, one seeks a test $E$ that minimizes the type-II error $\sup_{\tau \in \sE}\tr[E\tau]$ subject to a condition on the type-I error $\sup_{\rho \in \sP} \tr[(I - E)\rho] \leq \ve$. The subspace-constrained variant imposes an additional requirement $E \perp V(\sE)$, which forces $\tr[E\tau]$ to be identical for all $\tau \in \sE$. In other words, the test must remain oblivious to which candidate state is the true one. This constrained quantity determines precisely the extractable work as given below.

\begin{theorem}
    \label{thm:one-shot clean battery work extraction}
    Let $\ve \in [0,1)$, and let $\sP, \sE \subseteq \density$. The one-shot extractable work from the uncertain athermal state $(\sP, \sE)$ into a clean battery under $\GPO$ is given by
    $ \beta W_{\GPO, \ve}(\sP, \sE) = D_{\min,\ve}^{\sE}(\sP\|\sE)$.
\end{theorem}

This theorem extends the known connections between work extraction and hypothesis testing~\cite{liu2019oneshot,gour2022role,watanabe2024black} to the general framework developed here. If $\sE$ is a singleton, it reduces to the black-box work extraction in Ref.~\cite{watanabe2024black}. If $\sP$ is also a singleton, it further reduces to the standard work extraction formula in Ref.~\cite{gour2022role}. In the presence of equilibrium uncertainty, however, the test must satisfy the additional orthogonality condition. This condition captures the penalty imposed by equilibrium uncertainty: it  enforces a calibration-free test that cannot exploit knowledge of which equilibrium state is realized. Moreover, a reformulation~\cite{zhang2026sm} gives $D_{\min,\ve}^{\sE}(\sP\|\sE) = D_{\min,\ve}(\conv(\sP)\|\aff(\sE))$, which explains why the geometric condition appears in the no-go results.

We next consider the reverse task of preparing an uncertain athermal state from a clean battery. The following theorem relates the one-shot work cost to the standard smoothed max-relative entropy~\cite{datta2009min,fang2025generalized},
\begin{align*}
    D_{\max,\ve}(\sP\|\sE):= \log \inf_{\substack{\tau \in \sE\\\omega \in \sB_\ve(\sP)}} \{M: M\tau \geq \omega\},
\end{align*}
where $\sB_\ve(\sP) := \{\omega \in \density : T(\omega,\rho) \leq \ve \text{ for some } \rho \in \sP\}$ denotes the $\ve$-ball around $\sP$.

\begin{theorem}
    \label{thm:one-shot clean battery work cost}
    Let $\ve \in [0,1)$, and let $\sP, \sE \subseteq \density$. The one-shot work cost of preparing the uncertain athermal state $(\sP,\sE)$ from a clean battery under $\GPO$ is given by
    $\beta C_{\GPO, \ve}(\sP,\sE) = D_{\max, \ve}(\sP\|\sE)$.
\end{theorem}

When both $\sP$ and $\sE$ are singletons, this theorem recovers the standard one-shot work-cost formula in~\cite{gour2022role}. More generally, it shows that the work cost of preparing an uncertain athermal state is determined by the easiest candidate within the target set. It also gives the smoothed max-relative entropy between two sets of quantum states, previously studied as a mathematical quantity in Ref.~\cite{fang2025generalized}, a direct operational interpretation.

\paragraph{Dirty-battery work transformations.---} The preceding no-go results and clean-battery analysis reveal a sharp obstruction: under equilibrium uncertainty, the subspace constraint on the test can severely limit work extraction and, in generic cases (Corollary~\ref{cor:no-go for work extraction to clean battery}), collapse the usual work-error tradeoff. Since uncertainty in the equilibrium state of the input system is already unavoidable, it is reasonable to allow uncertainty in the target system as well and ask whether this obstruction can be circumvented by using a dirty battery. The following theorem answers this question affirmatively: work extraction to a dirty battery $\widebar{B}_M$ admits a nontrivial operational regime and is precisely characterized by the standard smoothed min-relative entropy $D_{\min,\ve}(\sP\|\sE)$~\cite{watanabe2024black,fang2025generalized}. 

\begin{theorem}
    \label{thm:one-shot dirty battery work extraction}
    Let $\ve \in [0,1)$, and let $\sP, \sE \subseteq \density$. The one-shot extractable work from the uncertain athermal state $(\sP,\sE)$ into a dirty battery under $\GPO$ is given by
    $\beta \widebar{W}_{\GPO,\ve}(\sP, \sE) = D_{\min, \ve}(\sP\|\sE)$.
\end{theorem}

If $\sE$ is a singleton, the equilibrium state of the output battery is uniquely determined, and this knowledge effectively makes the dirty battery a clean one. In this case, the theorem also recovers the known expressions for both standard work extraction~\cite{gour2022role} and the recent black-box work extraction~\cite{watanabe2024black}. For a general set $\sE$, the equilibrium uncertainty in the dirty battery removes the subspace constraint that suppresses work extraction in the clean-battery setting, yielding a generally larger extractable work: $D_{\min,\ve}(\sP\|\sE) \geq D_{\min,\ve}^{\sE}(\sP\|\sE)$. Moreover, this result gives the smoothed min-relative entropy between two sets of quantum states, a mathematical quantity studied in~\cite{fang2025generalized}, an operational interpretation, complementing the role of the max-relative entropy in Theorem~\ref{thm:one-shot clean battery work cost}.

For the reverse task, however, the work stored in a dirty battery becomes difficult to reinvest for state preparation. In this case, the work cost is determined by the new \emph{subspace-constrained max-relative entropy} between $\sP$ and $\sE$ with respect to a set of states $\sK\subseteq\density$,
\begin{align*}
    D_{\max,\ve}^{\sK}(\sP\|\sE) := \log\!\inf_{\substack{\gamma \in \cl(\sK)\\\omega \in \sB_\ve(\sP)}} \!\{M : \sC(\gamma, \omega, 1/M) \subseteq \sE\},
\end{align*}
where $\cl(\cK)$ denotes the closure of $\cK$ in $\density$, and $\sC(\gamma, \omega, 1/M) := \{(1-\lambda)\gamma + \lambda \omega : \lambda \in (0, 1/M]\}$ denotes the initial $1/M$ portion of the half open segment from $\gamma$ toward $\omega$. This quantity relates to the cone-restricted max-relative entropy in~\cite{george2024conerestricted} for closed and convex $\sE$~\cite{zhang2026sm}. 

\begin{theorem}
    \label{thm:one-shot dirty battery work cost}
    Let $\ve \in [0,1)$, and let $\sP, \sE \subseteq \density$. The one-shot work cost of preparing the uncertain athermal state $(\sP,\sE)$ from a dirty battery under $\GPO$ is given by
    $\beta \widebar{C}_{\GPO,\ve}(\sP,\sE) = D_{\max,\ve}^{\,\sE}(\sP\|\sE)$.
\end{theorem}

In contrast to Theorem~\ref{thm:one-shot clean battery work cost}, $D_{\max,\ve}^{\,\sE}(\sP\|\sE)$ depends explicitly on the geometry of $\sE$, reflecting the constraint imposed by uncertainty in the battery's equilibrium state. Theorem~\ref{thm:one-shot dirty battery work cost} also admits a clear geometric interpretation: the conversion $\widebar{B}_M \xrightarrow[]{\cF,\;\ve} (\sP, \sE)$ is achievable by some $\cF\in\GPO$ if and only if there exist $\omega \in \sB_\ve(\sP)$ and $\tau \in \cl(\sE)$ such that the initial portion of length $1/M$ of the line segment from $\tau$ toward $\omega$, $\sC(\tau,\omega,1/M)$, lies entirely within $\sE$. The work cost is then determined by the smallest such $M$. Intuitively, the narrower the set $\sE$, the shorter the segment it can accommodate, and hence the larger the required work cost. In the extreme case where $\sE$ is a singleton lying outside $\sB_\ve(\sP)$, no valid pair $(\omega, \tau)$ exists and the work cost diverges. A precise characterization of this intuition can be found in~\cite{zhang2026sm}.

\paragraph{Irreversibility.---} A hallmark of the standard resource theory of athermality is the asymptotic reversibility: the work extraction rate coincides with the work cost rate, so a cyclic process of extraction and formation incurs no net loss~\cite{brandao2013resource,gour2022role}. To examine whether this reversibility persists under equilibrium uncertainty, we consider a sequence of uncertain athermal states $\{(\sP_n,\sE_n)\}_{n=1}^\infty$, and define the asymptotic rates $\beta \WW_{\GPO, \ve}^\infty\left(\sP, \sE\right) := \lim_{n\to\infty} \beta \WW_{\GPO, \ve}(\sP_n,\sE_n)/n$ and $\beta \CC_{\GPO, \ve}^\infty\left(\sP, \sE\right) := \lim_{n\to\infty} \beta \CC_{\GPO, \ve}(\sP_n,\sE_n)/n$, where $\WW \in \left\{W, \widebar{W}\right\}$ and $\CC \in \left\{C, \widebar{C}\right\}$ denote extractable work and work cost in the clean- and dirty-battery settings, respectively.

A simple example already reveals a strict separation between these rates. Let $\ve < 1/2$ and $\delta >0$, and consider the athermal state with precise unequilibrim state $\sP_n = \{\ket{1}\bra{1}^{\ox n}\}$ and uncertain equilibrium state $\sE_n = \left\{\pi_M^{\ox n}: M \in [2, 2+\delta]\right\}$, which describes $n$ copies of an excited dirty battery whose capacity is known only up to precision $\delta$. A direct evaluation~\cite{zhang2026sm} gives, for clean batteries, $\beta W_{\GPO,\ve}^\infty(\sP, \sE) = 0 < 1 = \beta C_{\GPO,\ve}^\infty(\sP, \sE)$. Thus the extractable work vanishes while the formation cost remains strictly positive, giving rise to a thermodynamic analogue of bound entanglement~\cite{horodecki1998mixedstate}. By contrast, for dirty batteries, $\beta \widebar{W}_{\GPO,\ve}^\infty(\sP, \sE) = 1 < \infty = \beta \widebar{C}_{\GPO,\ve}^\infty(\sP, \sE)$, so work can be extracted at an nonzero rate, but the formation cost diverges. These extreme behaviors show that even infinitesimal imprecision on the equilibrium can fundamentally reshape the asymptotic landscape, turning reversibility to severe irreversibility. Further asymptotic analysis is provided in~\cite{zhang2026sm}.

\paragraph{Discussion.---} We challenge the long-standing idealization in quantum thermodynamics by showing that imperfect knowledge of equilibrium, a practically unavoidable source of uncertainty, can impose fundamental limitations on resource manipulation. Rather than giving small corrections to the standard theory, this uncertainty leads to new phenomena and techniques, including generic no-go behaviors, strong forms of irreversibility and new entropic tools. These contributions fit naturally with the continuing interest in resource-theoretic quantum thermodynamics, e.g.,~\cite{shiraishi2021quantum,marvian2022operational,safranek2023work,watanabe2024black,tajima2025gibbspreserving}. More broadly, it opens a new direction for exploring quantum resource theories in general, one that incorporates realistic imperfections and asks how uncertainty reshapes the ultimate limits of quantum information processing.

\let\oldaddcontentsline\addcontentsline
\renewcommand{\addcontentsline}[3]{}

\begin{acknowledgments}
\paragraph{Acknowledgments.---} We thank Gerardo Adesso, Zi-Wen Liu, Kaito Watanabe, and Yunlong Xiao for helpful discussions. We are particularly grateful to Ryuji Takagi for stimulating discussions during the Quantum Resources 2026 workshop in Tokyo and for insightful comments on the affine-hull reformulation of the subspace-constrained min-relative entropy. K.F. and M.Z. are supported by the National Natural Science Foundation of China (Grant No. 92470113 and 12404569), the Shenzhen Science and Technology Program (Grant No. QNXMB20250701091826036 and JCYJ20240813113519025), the Shenzhen Fundamental Research Program (Grant No. JCYJ20241202124023031), the General R\&D Projects of 1+1+1 CUHK-CUHK(SZ)-GDST Joint Collaboration Fund (Grant No. GRDP2025-022), and the University Development Fund (Grant No. UDF01003565). 
\end{acknowledgments}

\nocite{*}

\bibliography{Ref}

\begin{thebibliography}{58}%
\makeatletter
\providecommand \@ifxundefined [1]{%
 \@ifx{#1\undefined}
}%
\providecommand \@ifnum [1]{%
 \ifnum #1\expandafter \@firstoftwo
 \else \expandafter \@secondoftwo
 \fi
}%
\providecommand \@ifx [1]{%
 \ifx #1\expandafter \@firstoftwo
 \else \expandafter \@secondoftwo
 \fi
}%
\providecommand \natexlab [1]{#1}%
\providecommand \enquote  [1]{``#1''}%
\providecommand \bibnamefont  [1]{#1}%
\providecommand \bibfnamefont [1]{#1}%
\providecommand \citenamefont [1]{#1}%
\providecommand \href@noop [0]{\@secondoftwo}%
\providecommand \href [0]{\begingroup \@sanitize@url \@href}%
\providecommand \@href[1]{\@@startlink{#1}\@@href}%
\providecommand \@@href[1]{\endgroup#1\@@endlink}%
\providecommand \@sanitize@url [0]{\catcode `\\12\catcode `\$12\catcode `\&12\catcode `\#12\catcode `\^12\catcode `\_12\catcode `\%12\relax}%
\providecommand \@@startlink[1]{}%
\providecommand \@@endlink[0]{}%
\providecommand \url  [0]{\begingroup\@sanitize@url \@url }%
\providecommand \@url [1]{\endgroup\@href {#1}{\urlprefix }}%
\providecommand \urlprefix  [0]{URL }%
\providecommand \Eprint [0]{\href }%
\providecommand \doibase [0]{https://doi.org/}%
\providecommand \selectlanguage [0]{\@gobble}%
\providecommand \bibinfo  [0]{\@secondoftwo}%
\providecommand \bibfield  [0]{\@secondoftwo}%
\providecommand \translation [1]{[#1]}%
\providecommand \BibitemOpen [0]{}%
\providecommand \bibitemStop [0]{}%
\providecommand \bibitemNoStop [0]{.\EOS\space}%
\providecommand \EOS [0]{\spacefactor3000\relax}%
\providecommand \BibitemShut  [1]{\csname bibitem#1\endcsname}%
\let\auto@bib@innerbib\@empty
\bibitem [{\citenamefont {Brand{\~a}o}\ \emph {et~al.}(2013)\citenamefont {Brand{\~a}o}, \citenamefont {Horodecki}, \citenamefont {Oppenheim}, \citenamefont {Renes},\ and\ \citenamefont {Spekkens}}]{brandao2013resource}%
  \BibitemOpen
  \bibfield  {author} {\bibinfo {author} {\bibfnamefont {F.~G. S.~L.}\ \bibnamefont {Brand{\~a}o}}, \bibinfo {author} {\bibfnamefont {M.}~\bibnamefont {Horodecki}}, \bibinfo {author} {\bibfnamefont {J.}~\bibnamefont {Oppenheim}}, \bibinfo {author} {\bibfnamefont {J.~M.}\ \bibnamefont {Renes}},\ and\ \bibinfo {author} {\bibfnamefont {R.~W.}\ \bibnamefont {Spekkens}},\ }\bibfield  {title} {\bibinfo {title} {Resource theory of quantum states out of thermal equilibrium},\ }\href {https://doi.org/10.1103/PhysRevLett.111.250404} {\bibfield  {journal} {\bibinfo  {journal} {Physical Review Letters}\ }\textbf {\bibinfo {volume} {111}},\ \bibinfo {pages} {250404} (\bibinfo {year} {2013})}\BibitemShut {NoStop}%
\bibitem [{\citenamefont {Horodecki}\ and\ \citenamefont {Oppenheim}(2013)}]{horodecki2013fundamental}%
  \BibitemOpen
  \bibfield  {author} {\bibinfo {author} {\bibfnamefont {M.}~\bibnamefont {Horodecki}}\ and\ \bibinfo {author} {\bibfnamefont {J.}~\bibnamefont {Oppenheim}},\ }\bibfield  {title} {\bibinfo {title} {Fundamental limitations for quantum and nanoscale thermodynamics},\ }\href {https://doi.org/10.1038/ncomms3059} {\bibfield  {journal} {\bibinfo  {journal} {Nature Communications}\ }\textbf {\bibinfo {volume} {4}},\ \bibinfo {pages} {2059} (\bibinfo {year} {2013})}\BibitemShut {NoStop}%
\bibitem [{\citenamefont {Lostaglio}(2019)}]{lostaglio2019introductory}%
  \BibitemOpen
  \bibfield  {author} {\bibinfo {author} {\bibfnamefont {M.}~\bibnamefont {Lostaglio}},\ }\bibfield  {title} {\bibinfo {title} {An introductory review of the resource theory approach to thermodynamics},\ }\href {https://doi.org/10.1088/1361-6633/ab46e5} {\bibfield  {journal} {\bibinfo  {journal} {Reports on Progress in Physics}\ }\textbf {\bibinfo {volume} {82}},\ \bibinfo {pages} {114001} (\bibinfo {year} {2019})}\BibitemShut {NoStop}%
\bibitem [{\citenamefont {Chitambar}\ and\ \citenamefont {Gour}(2019)}]{chitambar2019quantum}%
  \BibitemOpen
  \bibfield  {author} {\bibinfo {author} {\bibfnamefont {E.}~\bibnamefont {Chitambar}}\ and\ \bibinfo {author} {\bibfnamefont {G.}~\bibnamefont {Gour}},\ }\bibfield  {title} {\bibinfo {title} {Quantum resource theories},\ }\href {https://doi.org/10.1103/RevModPhys.91.025001} {\bibfield  {journal} {\bibinfo  {journal} {Reviews of Modern Physics}\ }\textbf {\bibinfo {volume} {91}},\ \bibinfo {pages} {025001} (\bibinfo {year} {2019})}\BibitemShut {NoStop}%
\bibitem [{\citenamefont {Skrzypczyk}\ \emph {et~al.}(2014)\citenamefont {Skrzypczyk}, \citenamefont {Short},\ and\ \citenamefont {Popescu}}]{skrzypczyk2014work}%
  \BibitemOpen
  \bibfield  {author} {\bibinfo {author} {\bibfnamefont {P.}~\bibnamefont {Skrzypczyk}}, \bibinfo {author} {\bibfnamefont {A.~J.}\ \bibnamefont {Short}},\ and\ \bibinfo {author} {\bibfnamefont {S.}~\bibnamefont {Popescu}},\ }\bibfield  {title} {\bibinfo {title} {Work extraction and thermodynamics for individual quantum systems},\ }\href {https://doi.org/10.1038/ncomms5185} {\bibfield  {journal} {\bibinfo  {journal} {Nature Communications}\ }\textbf {\bibinfo {volume} {5}},\ \bibinfo {pages} {4185} (\bibinfo {year} {2014})}\BibitemShut {NoStop}%
\bibitem [{\citenamefont {Gour}\ \emph {et~al.}(2018)\citenamefont {Gour}, \citenamefont {Jennings}, \citenamefont {Buscemi}, \citenamefont {Duan},\ and\ \citenamefont {Marvian}}]{gour2018quantum}%
  \BibitemOpen
  \bibfield  {author} {\bibinfo {author} {\bibfnamefont {G.}~\bibnamefont {Gour}}, \bibinfo {author} {\bibfnamefont {D.}~\bibnamefont {Jennings}}, \bibinfo {author} {\bibfnamefont {F.}~\bibnamefont {Buscemi}}, \bibinfo {author} {\bibfnamefont {R.}~\bibnamefont {Duan}},\ and\ \bibinfo {author} {\bibfnamefont {I.}~\bibnamefont {Marvian}},\ }\bibfield  {title} {\bibinfo {title} {Quantum majorization and a complete set of entropic conditions for quantum thermodynamics},\ }\href {https://doi.org/10.1038/s41467-018-06261-7} {\bibfield  {journal} {\bibinfo  {journal} {Nature Communications}\ }\textbf {\bibinfo {volume} {9}},\ \bibinfo {pages} {5352} (\bibinfo {year} {2018})}\BibitemShut {NoStop}%
\bibitem [{\citenamefont {{\'C}wikli{\'n}ski}\ \emph {et~al.}(2015)\citenamefont {{\'C}wikli{\'n}ski}, \citenamefont {Studzi{\'n}ski}, \citenamefont {Horodecki},\ and\ \citenamefont {Oppenheim}}]{cwiklinski2015limitations}%
  \BibitemOpen
  \bibfield  {author} {\bibinfo {author} {\bibfnamefont {P.}~\bibnamefont {{\'C}wikli{\'n}ski}}, \bibinfo {author} {\bibfnamefont {M.}~\bibnamefont {Studzi{\'n}ski}}, \bibinfo {author} {\bibfnamefont {M.}~\bibnamefont {Horodecki}},\ and\ \bibinfo {author} {\bibfnamefont {J.}~\bibnamefont {Oppenheim}},\ }\bibfield  {title} {\bibinfo {title} {Limitations on the evolution of quantum coherences: Towards fully quantum second laws of thermodynamics},\ }\href {https://doi.org/10.1103/PhysRevLett.115.210403} {\bibfield  {journal} {\bibinfo  {journal} {Physical Review Letters}\ }\textbf {\bibinfo {volume} {115}},\ \bibinfo {pages} {210403} (\bibinfo {year} {2015})}\BibitemShut {NoStop}%
\bibitem [{\citenamefont {Brand{\~a}o}\ \emph {et~al.}(2015)\citenamefont {Brand{\~a}o}, \citenamefont {Horodecki}, \citenamefont {Ng}, \citenamefont {Oppenheim},\ and\ \citenamefont {Wehner}}]{brandao2015second}%
  \BibitemOpen
  \bibfield  {author} {\bibinfo {author} {\bibfnamefont {F.}~\bibnamefont {Brand{\~a}o}}, \bibinfo {author} {\bibfnamefont {M.}~\bibnamefont {Horodecki}}, \bibinfo {author} {\bibfnamefont {N.}~\bibnamefont {Ng}}, \bibinfo {author} {\bibfnamefont {J.}~\bibnamefont {Oppenheim}},\ and\ \bibinfo {author} {\bibfnamefont {S.}~\bibnamefont {Wehner}},\ }\bibfield  {title} {\bibinfo {title} {The second laws of quantum thermodynamics},\ }\href {https://doi.org/10.1073/pnas.1411728112} {\bibfield  {journal} {\bibinfo  {journal} {Proceedings of the National Academy of Sciences}\ }\textbf {\bibinfo {volume} {112}},\ \bibinfo {pages} {3275} (\bibinfo {year} {2015})}\BibitemShut {NoStop}%
\bibitem [{\citenamefont {Lostaglio}\ \emph {et~al.}(2015{\natexlab{a}})\citenamefont {Lostaglio}, \citenamefont {Jennings},\ and\ \citenamefont {Rudolph}}]{lostaglio2015description}%
  \BibitemOpen
  \bibfield  {author} {\bibinfo {author} {\bibfnamefont {M.}~\bibnamefont {Lostaglio}}, \bibinfo {author} {\bibfnamefont {D.}~\bibnamefont {Jennings}},\ and\ \bibinfo {author} {\bibfnamefont {T.}~\bibnamefont {Rudolph}},\ }\bibfield  {title} {\bibinfo {title} {Description of quantum coherence in thermodynamic processes requires constraints beyond free energy},\ }\href {https://doi.org/10.1038/ncomms7383} {\bibfield  {journal} {\bibinfo  {journal} {Nature Communications}\ }\textbf {\bibinfo {volume} {6}},\ \bibinfo {pages} {6383} (\bibinfo {year} {2015}{\natexlab{a}})}\BibitemShut {NoStop}%
\bibitem [{\citenamefont {Lostaglio}\ \emph {et~al.}(2015{\natexlab{b}})\citenamefont {Lostaglio}, \citenamefont {Korzekwa}, \citenamefont {Jennings},\ and\ \citenamefont {Rudolph}}]{lostaglio2015quantum}%
  \BibitemOpen
  \bibfield  {author} {\bibinfo {author} {\bibfnamefont {M.}~\bibnamefont {Lostaglio}}, \bibinfo {author} {\bibfnamefont {K.}~\bibnamefont {Korzekwa}}, \bibinfo {author} {\bibfnamefont {D.}~\bibnamefont {Jennings}},\ and\ \bibinfo {author} {\bibfnamefont {T.}~\bibnamefont {Rudolph}},\ }\bibfield  {title} {\bibinfo {title} {Quantum coherence, time-translation symmetry, and thermodynamics},\ }\href {https://doi.org/10.1103/PhysRevX.5.021001} {\bibfield  {journal} {\bibinfo  {journal} {Physical Review X}\ }\textbf {\bibinfo {volume} {5}},\ \bibinfo {pages} {021001} (\bibinfo {year} {2015}{\natexlab{b}})}\BibitemShut {NoStop}%
\bibitem [{\citenamefont {Sapienza}\ \emph {et~al.}(2019)\citenamefont {Sapienza}, \citenamefont {Cerisola},\ and\ \citenamefont {Roncaglia}}]{sapienza2019correlations}%
  \BibitemOpen
  \bibfield  {author} {\bibinfo {author} {\bibfnamefont {F.}~\bibnamefont {Sapienza}}, \bibinfo {author} {\bibfnamefont {F.}~\bibnamefont {Cerisola}},\ and\ \bibinfo {author} {\bibfnamefont {A.~J.}\ \bibnamefont {Roncaglia}},\ }\bibfield  {title} {\bibinfo {title} {Correlations as a resource in quantum thermodynamics},\ }\href {https://doi.org/10.1038/s41467-019-10572-8} {\bibfield  {journal} {\bibinfo  {journal} {Nature Communications}\ }\textbf {\bibinfo {volume} {10}},\ \bibinfo {pages} {2492} (\bibinfo {year} {2019})}\BibitemShut {NoStop}%
\bibitem [{\citenamefont {Marvian}(2020)}]{marvian2020coherence}%
  \BibitemOpen
  \bibfield  {author} {\bibinfo {author} {\bibfnamefont {I.}~\bibnamefont {Marvian}},\ }\bibfield  {title} {\bibinfo {title} {Coherence distillation machines are impossible in quantum thermodynamics},\ }\href {https://doi.org/10.1038/s41467-019-13846-3} {\bibfield  {journal} {\bibinfo  {journal} {Nature Communications}\ }\textbf {\bibinfo {volume} {11}},\ \bibinfo {pages} {25} (\bibinfo {year} {2020})}\BibitemShut {NoStop}%
\bibitem [{\citenamefont {Marvian}(2022)}]{marvian2022operational}%
  \BibitemOpen
  \bibfield  {author} {\bibinfo {author} {\bibfnamefont {I.}~\bibnamefont {Marvian}},\ }\bibfield  {title} {\bibinfo {title} {Operational interpretation of quantum fisher information in quantum thermodynamics},\ }\href {https://doi.org/10.1103/PhysRevLett.129.190502} {\bibfield  {journal} {\bibinfo  {journal} {Physical Review Letters}\ }\textbf {\bibinfo {volume} {129}},\ \bibinfo {pages} {190502} (\bibinfo {year} {2022})}\BibitemShut {NoStop}%
\bibitem [{\citenamefont {Gour}(2022)}]{gour2022role}%
  \BibitemOpen
  \bibfield  {author} {\bibinfo {author} {\bibfnamefont {G.}~\bibnamefont {Gour}},\ }\bibfield  {title} {\bibinfo {title} {Role of quantum coherence in thermodynamics},\ }\href {https://doi.org/10.1103/PRXQuantum.3.040323} {\bibfield  {journal} {\bibinfo  {journal} {PRX Quantum}\ }\textbf {\bibinfo {volume} {3}},\ \bibinfo {pages} {040323} (\bibinfo {year} {2022})}\BibitemShut {NoStop}%
\bibitem [{\citenamefont {Tajima}\ and\ \citenamefont {Takagi}(2025)}]{tajima2025gibbspreserving}%
  \BibitemOpen
  \bibfield  {author} {\bibinfo {author} {\bibfnamefont {H.}~\bibnamefont {Tajima}}\ and\ \bibinfo {author} {\bibfnamefont {R.}~\bibnamefont {Takagi}},\ }\bibfield  {title} {\bibinfo {title} {Gibbs-preserving operations requiring infinite amount of quantum coherence},\ }\href {https://doi.org/10.1103/PhysRevLett.134.170201} {\bibfield  {journal} {\bibinfo  {journal} {Physical Review Letters}\ }\textbf {\bibinfo {volume} {134}},\ \bibinfo {pages} {170201} (\bibinfo {year} {2025})}\BibitemShut {NoStop}%
\bibitem [{\citenamefont {{\v S}afr{\'a}nek}\ \emph {et~al.}(2023)\citenamefont {{\v S}afr{\'a}nek}, \citenamefont {Rosa},\ and\ \citenamefont {Binder}}]{safranek2023work}%
  \BibitemOpen
  \bibfield  {author} {\bibinfo {author} {\bibfnamefont {D.}~\bibnamefont {{\v S}afr{\'a}nek}}, \bibinfo {author} {\bibfnamefont {D.}~\bibnamefont {Rosa}},\ and\ \bibinfo {author} {\bibfnamefont {F.~C.}\ \bibnamefont {Binder}},\ }\bibfield  {title} {\bibinfo {title} {Work extraction from unknown quantum sources},\ }\href {https://doi.org/10.1103/PhysRevLett.130.210401} {\bibfield  {journal} {\bibinfo  {journal} {Physical Review Letters}\ }\textbf {\bibinfo {volume} {130}},\ \bibinfo {pages} {210401} (\bibinfo {year} {2023})}\BibitemShut {NoStop}%
\bibitem [{\citenamefont {Watanabe}\ and\ \citenamefont {Takagi}(2024)}]{watanabe2024black}%
  \BibitemOpen
  \bibfield  {author} {\bibinfo {author} {\bibfnamefont {K.}~\bibnamefont {Watanabe}}\ and\ \bibinfo {author} {\bibfnamefont {R.}~\bibnamefont {Takagi}},\ }\bibfield  {title} {\bibinfo {title} {Black box work extraction and composite hypothesis testing},\ }\href {https://doi.org/10.1103/PhysRevLett.133.250401} {\bibfield  {journal} {\bibinfo  {journal} {Physical Review Letters}\ }\textbf {\bibinfo {volume} {133}},\ \bibinfo {pages} {250401} (\bibinfo {year} {2024})}\BibitemShut {NoStop}%
\bibitem [{\citenamefont {Watanabe}\ and\ \citenamefont {Takagi}(2026)}]{watanabe2026universal}%
  \BibitemOpen
  \bibfield  {author} {\bibinfo {author} {\bibfnamefont {K.}~\bibnamefont {Watanabe}}\ and\ \bibinfo {author} {\bibfnamefont {R.}~\bibnamefont {Takagi}},\ }\bibfield  {title} {\bibinfo {title} {Universal work extraction in quantum thermodynamics},\ }\href {https://doi.org/10.1038/s41467-026-69143-3} {\bibfield  {journal} {\bibinfo  {journal} {Nature Communications}\ }\textbf {\bibinfo {volume} {17}},\ \bibinfo {pages} {1857} (\bibinfo {year} {2026})}\BibitemShut {NoStop}%
\bibitem [{\citenamefont {Giovannetti}\ \emph {et~al.}(2011)\citenamefont {Giovannetti}, \citenamefont {Lloyd},\ and\ \citenamefont {Maccone}}]{giovannetti2011advances}%
  \BibitemOpen
  \bibfield  {author} {\bibinfo {author} {\bibfnamefont {V.}~\bibnamefont {Giovannetti}}, \bibinfo {author} {\bibfnamefont {S.}~\bibnamefont {Lloyd}},\ and\ \bibinfo {author} {\bibfnamefont {L.}~\bibnamefont {Maccone}},\ }\bibfield  {title} {\bibinfo {title} {Advances in quantum metrology},\ }\href {https://doi.org/10.1038/nphoton.2011.35} {\bibfield  {journal} {\bibinfo  {journal} {Nature Photonics}\ }\textbf {\bibinfo {volume} {5}},\ \bibinfo {pages} {222} (\bibinfo {year} {2011})}\BibitemShut {NoStop}%
\bibitem [{\citenamefont {Montenegro}\ \emph {et~al.}(2025)\citenamefont {Montenegro}, \citenamefont {Mukhopadhyay}, \citenamefont {Yousefjani}, \citenamefont {Sarkar}, \citenamefont {Mishra}, \citenamefont {Paris},\ and\ \citenamefont {Bayat}}]{montenegro2025review}%
  \BibitemOpen
  \bibfield  {author} {\bibinfo {author} {\bibfnamefont {V.}~\bibnamefont {Montenegro}}, \bibinfo {author} {\bibfnamefont {C.}~\bibnamefont {Mukhopadhyay}}, \bibinfo {author} {\bibfnamefont {R.}~\bibnamefont {Yousefjani}}, \bibinfo {author} {\bibfnamefont {S.}~\bibnamefont {Sarkar}}, \bibinfo {author} {\bibfnamefont {U.}~\bibnamefont {Mishra}}, \bibinfo {author} {\bibfnamefont {M.~G.~A.}\ \bibnamefont {Paris}},\ and\ \bibinfo {author} {\bibfnamefont {A.}~\bibnamefont {Bayat}},\ }\bibfield  {title} {\bibinfo {title} {Review: Quantum metrology and sensing with many-body systems},\ }\href {https://doi.org/10.1016/j.physrep.2025.05.005} {\bibfield  {journal} {\bibinfo  {journal} {Physics Reports}\ }\textbf {\bibinfo {volume} {1134}},\ \bibinfo {pages} {1} (\bibinfo {year} {2025})}\BibitemShut {NoStop}%
\bibitem [{\citenamefont {Gerster}\ \emph {et~al.}(2022)\citenamefont {Gerster}, \citenamefont {{Mart{\'i}nez-Garc{\'i}a}}, \citenamefont {Hrmo}, \citenamefont {{van Mourik}}, \citenamefont {Wilhelm}, \citenamefont {Vodola}, \citenamefont {M{\"u}ller}, \citenamefont {Blatt}, \citenamefont {Schindler},\ and\ \citenamefont {Monz}}]{gerster2022experimental}%
  \BibitemOpen
  \bibfield  {author} {\bibinfo {author} {\bibfnamefont {L.}~\bibnamefont {Gerster}}, \bibinfo {author} {\bibfnamefont {F.}~\bibnamefont {{Mart{\'i}nez-Garc{\'i}a}}}, \bibinfo {author} {\bibfnamefont {P.}~\bibnamefont {Hrmo}}, \bibinfo {author} {\bibfnamefont {M.~W.}\ \bibnamefont {{van Mourik}}}, \bibinfo {author} {\bibfnamefont {B.}~\bibnamefont {Wilhelm}}, \bibinfo {author} {\bibfnamefont {D.}~\bibnamefont {Vodola}}, \bibinfo {author} {\bibfnamefont {M.}~\bibnamefont {M{\"u}ller}}, \bibinfo {author} {\bibfnamefont {R.}~\bibnamefont {Blatt}}, \bibinfo {author} {\bibfnamefont {P.}~\bibnamefont {Schindler}},\ and\ \bibinfo {author} {\bibfnamefont {T.}~\bibnamefont {Monz}},\ }\bibfield  {title} {\bibinfo {title} {Experimental bayesian calibration of trapped-ion entangling operations},\ }\href {https://doi.org/10.1103/PRXQuantum.3.020350} {\bibfield  {journal} {\bibinfo  {journal} {PRX Quantum}\ }\textbf {\bibinfo {volume} {3}},\ \bibinfo {pages} {020350} (\bibinfo {year} {2022})}\BibitemShut {NoStop}%
\bibitem [{\citenamefont {Berritta}\ \emph {et~al.}(2025)\citenamefont {Berritta}, \citenamefont {Benestad}, \citenamefont {Pahl}, \citenamefont {Mathews}, \citenamefont {Krzywda}, \citenamefont {Assouly}, \citenamefont {Sung}, \citenamefont {Kim}, \citenamefont {Niedzielski}, \citenamefont {Serniak}, \citenamefont {Schwartz}, \citenamefont {Yoder}, \citenamefont {Chatterjee}, \citenamefont {Grover}, \citenamefont {Danon}, \citenamefont {Oliver},\ and\ \citenamefont {Kuemmeth}}]{berritta2025efficient}%
  \BibitemOpen
  \bibfield  {author} {\bibinfo {author} {\bibfnamefont {F.}~\bibnamefont {Berritta}}, \bibinfo {author} {\bibfnamefont {J.}~\bibnamefont {Benestad}}, \bibinfo {author} {\bibfnamefont {L.}~\bibnamefont {Pahl}}, \bibinfo {author} {\bibfnamefont {M.}~\bibnamefont {Mathews}}, \bibinfo {author} {\bibfnamefont {J.~A.}\ \bibnamefont {Krzywda}}, \bibinfo {author} {\bibfnamefont {R.}~\bibnamefont {Assouly}}, \bibinfo {author} {\bibfnamefont {Y.}~\bibnamefont {Sung}}, \bibinfo {author} {\bibfnamefont {D.~K.}\ \bibnamefont {Kim}}, \bibinfo {author} {\bibfnamefont {B.~M.}\ \bibnamefont {Niedzielski}}, \bibinfo {author} {\bibfnamefont {K.}~\bibnamefont {Serniak}}, \bibinfo {author} {\bibfnamefont {M.~E.}\ \bibnamefont {Schwartz}}, \bibinfo {author} {\bibfnamefont {J.~L.}\ \bibnamefont {Yoder}}, \bibinfo {author} {\bibfnamefont {A.}~\bibnamefont {Chatterjee}}, \bibinfo {author} {\bibfnamefont {J.~A.}\ \bibnamefont {Grover}}, \bibinfo {author} {\bibfnamefont {J.}~\bibnamefont {Danon}}, \bibinfo {author} {\bibfnamefont {W.~D.}\ \bibnamefont {Oliver}},\ and\ \bibinfo {author} {\bibfnamefont {F.}~\bibnamefont {Kuemmeth}},\ }\bibfield  {title} {\bibinfo {title} {Efficient qubit calibration by binary-search hamiltonian tracking},\ }\href {https://doi.org/10.1103/77qg-p68k} {\bibfield  {journal} {\bibinfo  {journal} {PRX Quantum}\ }\textbf {\bibinfo {volume} {6}},\ \bibinfo {pages} {030335} (\bibinfo {year} {2025})}\BibitemShut {NoStop}%
\bibitem [{\citenamefont {Huang}\ \emph {et~al.}(2023)\citenamefont {Huang}, \citenamefont {Tong}, \citenamefont {Fang},\ and\ \citenamefont {Su}}]{huang2023learning}%
  \BibitemOpen
  \bibfield  {author} {\bibinfo {author} {\bibfnamefont {H.-Y.}\ \bibnamefont {Huang}}, \bibinfo {author} {\bibfnamefont {Y.}~\bibnamefont {Tong}}, \bibinfo {author} {\bibfnamefont {D.}~\bibnamefont {Fang}},\ and\ \bibinfo {author} {\bibfnamefont {Y.}~\bibnamefont {Su}},\ }\bibfield  {title} {\bibinfo {title} {Learning many-body hamiltonians with heisenberg-limited scaling},\ }\href {https://doi.org/10.1103/PhysRevLett.130.200403} {\bibfield  {journal} {\bibinfo  {journal} {Physical Review Letters}\ }\textbf {\bibinfo {volume} {130}},\ \bibinfo {pages} {200403} (\bibinfo {year} {2023})}\BibitemShut {NoStop}%
\bibitem [{\citenamefont {Hangleiter}\ \emph {et~al.}(2024)\citenamefont {Hangleiter}, \citenamefont {Roth}, \citenamefont {Fuksa}, \citenamefont {Eisert},\ and\ \citenamefont {Roushan}}]{hangleiter2024robustly}%
  \BibitemOpen
  \bibfield  {author} {\bibinfo {author} {\bibfnamefont {D.}~\bibnamefont {Hangleiter}}, \bibinfo {author} {\bibfnamefont {I.}~\bibnamefont {Roth}}, \bibinfo {author} {\bibfnamefont {J.}~\bibnamefont {Fuksa}}, \bibinfo {author} {\bibfnamefont {J.}~\bibnamefont {Eisert}},\ and\ \bibinfo {author} {\bibfnamefont {P.}~\bibnamefont {Roushan}},\ }\bibfield  {title} {\bibinfo {title} {Robustly learning the hamiltonian dynamics of a superconducting quantum processor},\ }\href {https://doi.org/10.1038/s41467-024-52629-3} {\bibfield  {journal} {\bibinfo  {journal} {Nature Communications}\ }\textbf {\bibinfo {volume} {15}},\ \bibinfo {pages} {9595} (\bibinfo {year} {2024})}\BibitemShut {NoStop}%
\bibitem [{\citenamefont {Proctor}\ \emph {et~al.}(2020)\citenamefont {Proctor}, \citenamefont {Revelle}, \citenamefont {Nielsen}, \citenamefont {Rudinger}, \citenamefont {Lobser}, \citenamefont {Maunz}, \citenamefont {{Blume-Kohout}},\ and\ \citenamefont {Young}}]{proctor2020detecting}%
  \BibitemOpen
  \bibfield  {author} {\bibinfo {author} {\bibfnamefont {T.}~\bibnamefont {Proctor}}, \bibinfo {author} {\bibfnamefont {M.}~\bibnamefont {Revelle}}, \bibinfo {author} {\bibfnamefont {E.}~\bibnamefont {Nielsen}}, \bibinfo {author} {\bibfnamefont {K.}~\bibnamefont {Rudinger}}, \bibinfo {author} {\bibfnamefont {D.}~\bibnamefont {Lobser}}, \bibinfo {author} {\bibfnamefont {P.}~\bibnamefont {Maunz}}, \bibinfo {author} {\bibfnamefont {R.}~\bibnamefont {{Blume-Kohout}}},\ and\ \bibinfo {author} {\bibfnamefont {K.}~\bibnamefont {Young}},\ }\bibfield  {title} {\bibinfo {title} {Detecting and tracking drift in quantum information processors},\ }\href {https://doi.org/10.1038/s41467-020-19074-4} {\bibfield  {journal} {\bibinfo  {journal} {Nature Communications}\ }\textbf {\bibinfo {volume} {11}},\ \bibinfo {pages} {5396} (\bibinfo {year} {2020})}\BibitemShut {NoStop}%
\bibitem [{\citenamefont {Capannelli}\ \emph {et~al.}(2025)\citenamefont {Capannelli}, \citenamefont {Undseth}, \citenamefont {{Fern{\'a}ndez de Fuentes}}, \citenamefont {Raymenants}, \citenamefont {Unseld}, \citenamefont {{Pietx-Casas}}, \citenamefont {Philips}, \citenamefont {M{\k a}dzik}, \citenamefont {Amitonov}, \citenamefont {Tryputen}, \citenamefont {Scappucci},\ and\ \citenamefont {Vandersypen}}]{capannelli2025tracking}%
  \BibitemOpen
  \bibfield  {author} {\bibinfo {author} {\bibfnamefont {K.}~\bibnamefont {Capannelli}}, \bibinfo {author} {\bibfnamefont {B.}~\bibnamefont {Undseth}}, \bibinfo {author} {\bibfnamefont {I.}~\bibnamefont {{Fern{\'a}ndez de Fuentes}}}, \bibinfo {author} {\bibfnamefont {E.}~\bibnamefont {Raymenants}}, \bibinfo {author} {\bibfnamefont {F.~K.}\ \bibnamefont {Unseld}}, \bibinfo {author} {\bibfnamefont {O.}~\bibnamefont {{Pietx-Casas}}}, \bibinfo {author} {\bibfnamefont {S.~G.~J.}\ \bibnamefont {Philips}}, \bibinfo {author} {\bibfnamefont {M.~T.}\ \bibnamefont {M{\k a}dzik}}, \bibinfo {author} {\bibfnamefont {S.~V.}\ \bibnamefont {Amitonov}}, \bibinfo {author} {\bibfnamefont {L.}~\bibnamefont {Tryputen}}, \bibinfo {author} {\bibfnamefont {G.}~\bibnamefont {Scappucci}},\ and\ \bibinfo {author} {\bibfnamefont {L.~M.~K.}\ \bibnamefont {Vandersypen}},\ }\bibfield  {title} {\bibinfo {title} {Tracking spin qubit frequency variations over 912 days},\ }\href {https://doi.org/10.1038/s41534-025-01134-6} {\bibfield  {journal} {\bibinfo  {journal} {npj Quantum Information}\ }\textbf {\bibinfo {volume} {11}},\ \bibinfo {pages} {192} (\bibinfo {year} {2025})}\BibitemShut {NoStop}%
\bibitem [{\citenamefont {{Google Quantum AI and Collaborators}}(2025)}]{acharya2025quantum}%
  \BibitemOpen
  \bibfield  {author} {\bibinfo {author} {\bibnamefont {{Google Quantum AI and Collaborators}}},\ }\bibfield  {title} {\bibinfo {title} {Quantum error correction below the surface code threshold},\ }\href {https://doi.org/10.1038/s41586-024-08449-y} {\bibfield  {journal} {\bibinfo  {journal} {Nature}\ }\textbf {\bibinfo {volume} {638}},\ \bibinfo {pages} {920} (\bibinfo {year} {2025})}\BibitemShut {NoStop}%
\bibitem [{\citenamefont {Bluvstein}\ \emph {et~al.}(2026)\citenamefont {Bluvstein}, \citenamefont {Geim}, \citenamefont {Li}, \citenamefont {Evered}, \citenamefont {Bonilla~Ataides}, \citenamefont {Baranes}, \citenamefont {Gu}, \citenamefont {Manovitz}, \citenamefont {Xu}, \citenamefont {Kalinowski}, \citenamefont {Majidy}, \citenamefont {Kokail}, \citenamefont {Maskara}, \citenamefont {Trapp}, \citenamefont {Stewart}, \citenamefont {Hollerith}, \citenamefont {Zhou}, \citenamefont {Gullans}, \citenamefont {Yelin}, \citenamefont {Greiner}, \citenamefont {Vuleti{\'c}}, \citenamefont {Cain},\ and\ \citenamefont {Lukin}}]{bluvstein2026faulttolerant}%
  \BibitemOpen
  \bibfield  {author} {\bibinfo {author} {\bibfnamefont {D.}~\bibnamefont {Bluvstein}}, \bibinfo {author} {\bibfnamefont {A.~A.}\ \bibnamefont {Geim}}, \bibinfo {author} {\bibfnamefont {S.~H.}\ \bibnamefont {Li}}, \bibinfo {author} {\bibfnamefont {S.~J.}\ \bibnamefont {Evered}}, \bibinfo {author} {\bibfnamefont {J.~P.}\ \bibnamefont {Bonilla~Ataides}}, \bibinfo {author} {\bibfnamefont {G.}~\bibnamefont {Baranes}}, \bibinfo {author} {\bibfnamefont {A.}~\bibnamefont {Gu}}, \bibinfo {author} {\bibfnamefont {T.}~\bibnamefont {Manovitz}}, \bibinfo {author} {\bibfnamefont {M.}~\bibnamefont {Xu}}, \bibinfo {author} {\bibfnamefont {M.}~\bibnamefont {Kalinowski}}, \bibinfo {author} {\bibfnamefont {S.}~\bibnamefont {Majidy}}, \bibinfo {author} {\bibfnamefont {C.}~\bibnamefont {Kokail}}, \bibinfo {author} {\bibfnamefont {N.}~\bibnamefont {Maskara}}, \bibinfo {author} {\bibfnamefont {E.~C.}\ \bibnamefont {Trapp}}, \bibinfo {author} {\bibfnamefont {L.~M.}\ \bibnamefont {Stewart}}, \bibinfo {author} {\bibfnamefont {S.}~\bibnamefont {Hollerith}}, \bibinfo {author} {\bibfnamefont {H.}~\bibnamefont {Zhou}}, \bibinfo {author} {\bibfnamefont {M.~J.}\ \bibnamefont {Gullans}}, \bibinfo {author} {\bibfnamefont {S.~F.}\ \bibnamefont {Yelin}}, \bibinfo {author} {\bibfnamefont {M.}~\bibnamefont {Greiner}}, \bibinfo {author} {\bibfnamefont {V.}~\bibnamefont {Vuleti{\'c}}}, \bibinfo {author} {\bibfnamefont {M.}~\bibnamefont {Cain}},\ and\ \bibinfo {author} {\bibfnamefont {M.~D.}\ \bibnamefont {Lukin}},\ }\bibfield  {title} {\bibinfo {title} {A fault-tolerant neutral-atom architecture for universal quantum computation},\ }\href {https://doi.org/10.1038/s41586-025-09848-5} {\bibfield  {journal} {\bibinfo  {journal} {Nature}\ }\textbf {\bibinfo {volume} {649}},\ \bibinfo {pages} {39} (\bibinfo {year} {2026})}\BibitemShut {NoStop}%
\bibitem [{\citenamefont {Mok}\ \emph {et~al.}(2021)\citenamefont {Mok}, \citenamefont {Bharti}, \citenamefont {Kwek},\ and\ \citenamefont {Bayat}}]{mok2021optimal}%
  \BibitemOpen
  \bibfield  {author} {\bibinfo {author} {\bibfnamefont {W.-K.}\ \bibnamefont {Mok}}, \bibinfo {author} {\bibfnamefont {K.}~\bibnamefont {Bharti}}, \bibinfo {author} {\bibfnamefont {L.-C.}\ \bibnamefont {Kwek}},\ and\ \bibinfo {author} {\bibfnamefont {A.}~\bibnamefont {Bayat}},\ }\bibfield  {title} {\bibinfo {title} {Optimal probes for global quantum thermometry},\ }\href {https://doi.org/10.1038/s42005-021-00572-w} {\bibfield  {journal} {\bibinfo  {journal} {Communications Physics}\ }\textbf {\bibinfo {volume} {4}},\ \bibinfo {pages} {62} (\bibinfo {year} {2021})}\BibitemShut {NoStop}%
\bibitem [{\citenamefont {Liu}\ \emph {et~al.}(2019)\citenamefont {Liu}, \citenamefont {Bu},\ and\ \citenamefont {Takagi}}]{liu2019oneshot}%
  \BibitemOpen
  \bibfield  {author} {\bibinfo {author} {\bibfnamefont {Z.-W.}\ \bibnamefont {Liu}}, \bibinfo {author} {\bibfnamefont {K.}~\bibnamefont {Bu}},\ and\ \bibinfo {author} {\bibfnamefont {R.}~\bibnamefont {Takagi}},\ }\bibfield  {title} {\bibinfo {title} {One-shot operational quantum resource theory},\ }\href {https://doi.org/10.1103/PhysRevLett.123.020401} {\bibfield  {journal} {\bibinfo  {journal} {Physical Review Letters}\ }\textbf {\bibinfo {volume} {123}},\ \bibinfo {pages} {020401} (\bibinfo {year} {2019})}\BibitemShut {NoStop}%
\bibitem [{\citenamefont {Horodecki}\ \emph {et~al.}(1998)\citenamefont {Horodecki}, \citenamefont {Horodecki},\ and\ \citenamefont {Horodecki}}]{horodecki1998mixedstate}%
  \BibitemOpen
  \bibfield  {author} {\bibinfo {author} {\bibfnamefont {M.}~\bibnamefont {Horodecki}}, \bibinfo {author} {\bibfnamefont {P.}~\bibnamefont {Horodecki}},\ and\ \bibinfo {author} {\bibfnamefont {R.}~\bibnamefont {Horodecki}},\ }\bibfield  {title} {\bibinfo {title} {Mixed-state entanglement and distillation: Is there a ``bound'' entanglement in nature?},\ }\href {https://doi.org/10.1103/PhysRevLett.80.5239} {\bibfield  {journal} {\bibinfo  {journal} {Physical Review Letters}\ }\textbf {\bibinfo {volume} {80}},\ \bibinfo {pages} {5239} (\bibinfo {year} {1998})}\BibitemShut {NoStop}%
\bibitem [{zha()}]{zhang2026sm}%
  \BibitemOpen
  \href@noop {} {}\bibinfo {note} {See Supplemental Material for background material, technical derivations, and extended discussions of the main results.}\BibitemShut {Stop}%
\bibitem [{\citenamefont {Janzing}\ \emph {et~al.}(2000)\citenamefont {Janzing}, \citenamefont {Wocjan}, \citenamefont {Zeier}, \citenamefont {Geiss},\ and\ \citenamefont {Beth}}]{janzing2000thermodynamic}%
  \BibitemOpen
  \bibfield  {author} {\bibinfo {author} {\bibfnamefont {D.}~\bibnamefont {Janzing}}, \bibinfo {author} {\bibfnamefont {P.}~\bibnamefont {Wocjan}}, \bibinfo {author} {\bibfnamefont {R.}~\bibnamefont {Zeier}}, \bibinfo {author} {\bibfnamefont {R.}~\bibnamefont {Geiss}},\ and\ \bibinfo {author} {\bibfnamefont {{\relax Th}.}~\bibnamefont {Beth}},\ }\bibfield  {title} {\bibinfo {title} {Thermodynamic cost of reliability and low temperatures: Tightening landauer's principle and the second law},\ }\href {https://doi.org/10.1023/A:1026422630734} {\bibfield  {journal} {\bibinfo  {journal} {International Journal of Theoretical Physics}\ }\textbf {\bibinfo {volume} {39}},\ \bibinfo {pages} {2717} (\bibinfo {year} {2000})}\BibitemShut {NoStop}%
\bibitem [{\citenamefont {Faist}\ \emph {et~al.}(2015)\citenamefont {Faist}, \citenamefont {Oppenheim},\ and\ \citenamefont {Renner}}]{faist2015gibbspreserving}%
  \BibitemOpen
  \bibfield  {author} {\bibinfo {author} {\bibfnamefont {P.}~\bibnamefont {Faist}}, \bibinfo {author} {\bibfnamefont {J.}~\bibnamefont {Oppenheim}},\ and\ \bibinfo {author} {\bibfnamefont {R.}~\bibnamefont {Renner}},\ }\bibfield  {title} {\bibinfo {title} {Gibbs-preserving maps outperform thermal operations in the quantum regime},\ }\href {https://doi.org/10.1088/1367-2630/17/4/043003} {\bibfield  {journal} {\bibinfo  {journal} {New Journal of Physics}\ }\textbf {\bibinfo {volume} {17}},\ \bibinfo {pages} {043003} (\bibinfo {year} {2015})}\BibitemShut {NoStop}%
\bibitem [{\citenamefont {Faist}\ and\ \citenamefont {Renner}(2018)}]{faist2018fundamental}%
  \BibitemOpen
  \bibfield  {author} {\bibinfo {author} {\bibfnamefont {P.}~\bibnamefont {Faist}}\ and\ \bibinfo {author} {\bibfnamefont {R.}~\bibnamefont {Renner}},\ }\bibfield  {title} {\bibinfo {title} {Fundamental work cost of quantum processes},\ }\href {https://doi.org/10.1103/PhysRevX.8.021011} {\bibfield  {journal} {\bibinfo  {journal} {Physical Review X}\ }\textbf {\bibinfo {volume} {8}},\ \bibinfo {pages} {021011} (\bibinfo {year} {2018})}\BibitemShut {NoStop}%
\bibitem [{\citenamefont {Faist}\ \emph {et~al.}(2019)\citenamefont {Faist}, \citenamefont {Berta},\ and\ \citenamefont {Brand{\~a}o}}]{faist2019thermodynamic}%
  \BibitemOpen
  \bibfield  {author} {\bibinfo {author} {\bibfnamefont {P.}~\bibnamefont {Faist}}, \bibinfo {author} {\bibfnamefont {M.}~\bibnamefont {Berta}},\ and\ \bibinfo {author} {\bibfnamefont {F.}~\bibnamefont {Brand{\~a}o}},\ }\bibfield  {title} {\bibinfo {title} {Thermodynamic capacity of quantum processes},\ }\href {https://doi.org/10.1103/PhysRevLett.122.200601} {\bibfield  {journal} {\bibinfo  {journal} {Physical Review Letters}\ }\textbf {\bibinfo {volume} {122}},\ \bibinfo {pages} {200601} (\bibinfo {year} {2019})}\BibitemShut {NoStop}%
\bibitem [{\citenamefont {Shiraishi}\ and\ \citenamefont {Sagawa}(2021)}]{shiraishi2021quantum}%
  \BibitemOpen
  \bibfield  {author} {\bibinfo {author} {\bibfnamefont {N.}~\bibnamefont {Shiraishi}}\ and\ \bibinfo {author} {\bibfnamefont {T.}~\bibnamefont {Sagawa}},\ }\bibfield  {title} {\bibinfo {title} {Quantum thermodynamics of correlated-catalytic state conversion at small scale},\ }\href {https://doi.org/10.1103/PhysRevLett.126.150502} {\bibfield  {journal} {\bibinfo  {journal} {Physical Review Letters}\ }\textbf {\bibinfo {volume} {126}},\ \bibinfo {pages} {150502} (\bibinfo {year} {2021})}\BibitemShut {NoStop}%
\bibitem [{\citenamefont {Alhambra}\ \emph {et~al.}(2016)\citenamefont {Alhambra}, \citenamefont {Oppenheim},\ and\ \citenamefont {Perry}}]{alhambra2016fluctuating}%
  \BibitemOpen
  \bibfield  {author} {\bibinfo {author} {\bibfnamefont {{\'A}.~M.}\ \bibnamefont {Alhambra}}, \bibinfo {author} {\bibfnamefont {J.}~\bibnamefont {Oppenheim}},\ and\ \bibinfo {author} {\bibfnamefont {C.}~\bibnamefont {Perry}},\ }\bibfield  {title} {\bibinfo {title} {Fluctuating states: What is the probability of a thermodynamical transition?},\ }\href {https://doi.org/10.1103/PhysRevX.6.041016} {\bibfield  {journal} {\bibinfo  {journal} {Physical Review X}\ }\textbf {\bibinfo {volume} {6}},\ \bibinfo {pages} {041016} (\bibinfo {year} {2016})}\BibitemShut {NoStop}%
\bibitem [{\citenamefont {Takagi}\ \emph {et~al.}(2024)\citenamefont {Takagi}, \citenamefont {Yuan}, \citenamefont {Regula},\ and\ \citenamefont {Gu}}]{takagi2024virtual}%
  \BibitemOpen
  \bibfield  {author} {\bibinfo {author} {\bibfnamefont {R.}~\bibnamefont {Takagi}}, \bibinfo {author} {\bibfnamefont {X.}~\bibnamefont {Yuan}}, \bibinfo {author} {\bibfnamefont {B.}~\bibnamefont {Regula}},\ and\ \bibinfo {author} {\bibfnamefont {M.}~\bibnamefont {Gu}},\ }\bibfield  {title} {\bibinfo {title} {Virtual quantum resource distillation: General framework and applications},\ }\href {https://doi.org/10.1103/PhysRevA.109.022403} {\bibfield  {journal} {\bibinfo  {journal} {Physical Review A}\ }\textbf {\bibinfo {volume} {109}},\ \bibinfo {pages} {022403} (\bibinfo {year} {2024})}\BibitemShut {NoStop}%
\bibitem [{\citenamefont {Yuan}\ \emph {et~al.}(2024)\citenamefont {Yuan}, \citenamefont {Regula}, \citenamefont {Takagi},\ and\ \citenamefont {Gu}}]{yuan2024virtual}%
  \BibitemOpen
  \bibfield  {author} {\bibinfo {author} {\bibfnamefont {X.}~\bibnamefont {Yuan}}, \bibinfo {author} {\bibfnamefont {B.}~\bibnamefont {Regula}}, \bibinfo {author} {\bibfnamefont {R.}~\bibnamefont {Takagi}},\ and\ \bibinfo {author} {\bibfnamefont {M.}~\bibnamefont {Gu}},\ }\bibfield  {title} {\bibinfo {title} {Virtual quantum resource distillation},\ }\href {https://doi.org/10.1103/PhysRevLett.132.050203} {\bibfield  {journal} {\bibinfo  {journal} {Physical Review Letters}\ }\textbf {\bibinfo {volume} {132}},\ \bibinfo {pages} {050203} (\bibinfo {year} {2024})}\BibitemShut {NoStop}%
\bibitem [{\citenamefont {Zhu}\ \emph {et~al.}(2024)\citenamefont {Zhu}, \citenamefont {Mo}, \citenamefont {Chen},\ and\ \citenamefont {Wang}}]{zhu2024reversing}%
  \BibitemOpen
  \bibfield  {author} {\bibinfo {author} {\bibfnamefont {C.}~\bibnamefont {Zhu}}, \bibinfo {author} {\bibfnamefont {Y.}~\bibnamefont {Mo}}, \bibinfo {author} {\bibfnamefont {Y.-A.}\ \bibnamefont {Chen}},\ and\ \bibinfo {author} {\bibfnamefont {X.}~\bibnamefont {Wang}},\ }\bibfield  {title} {\bibinfo {title} {Reversing unknown quantum processes via virtual combs for channels with limited information},\ }\href {https://doi.org/10.1103/PhysRevLett.133.030801} {\bibfield  {journal} {\bibinfo  {journal} {Physical Review Letters}\ }\textbf {\bibinfo {volume} {133}},\ \bibinfo {pages} {030801} (\bibinfo {year} {2024})}\BibitemShut {NoStop}%
\bibitem [{\citenamefont {Fang}\ and\ \citenamefont {Liu}(2020)}]{fang2020nogo}%
  \BibitemOpen
  \bibfield  {author} {\bibinfo {author} {\bibfnamefont {K.}~\bibnamefont {Fang}}\ and\ \bibinfo {author} {\bibfnamefont {Z.-W.}\ \bibnamefont {Liu}},\ }\bibfield  {title} {\bibinfo {title} {No-go theorems for quantum resource purification},\ }\href {https://doi.org/10.1103/PhysRevLett.125.060405} {\bibfield  {journal} {\bibinfo  {journal} {Physical Review Letters}\ }\textbf {\bibinfo {volume} {125}},\ \bibinfo {pages} {060405} (\bibinfo {year} {2020})}\BibitemShut {NoStop}%
\bibitem [{\citenamefont {Fang}\ and\ \citenamefont {Liu}(2022)}]{fang2022nogo}%
  \BibitemOpen
  \bibfield  {author} {\bibinfo {author} {\bibfnamefont {K.}~\bibnamefont {Fang}}\ and\ \bibinfo {author} {\bibfnamefont {Z.-W.}\ \bibnamefont {Liu}},\ }\bibfield  {title} {\bibinfo {title} {No-go theorems for quantum resource purification: New approach and channel theory},\ }\href {https://doi.org/10.1103/PRXQuantum.3.010337} {\bibfield  {journal} {\bibinfo  {journal} {PRX Quantum}\ }\textbf {\bibinfo {volume} {3}},\ \bibinfo {pages} {010337} (\bibinfo {year} {2022})}\BibitemShut {NoStop}%
\bibitem [{\citenamefont {Liu}\ \emph {et~al.}(2025)\citenamefont {Liu}, \citenamefont {Du}, \citenamefont {Cai},\ and\ \citenamefont {Liu}}]{liu2025no}%
  \BibitemOpen
  \bibfield  {author} {\bibinfo {author} {\bibfnamefont {Z.}~\bibnamefont {Liu}}, \bibinfo {author} {\bibfnamefont {Z.}~\bibnamefont {Du}}, \bibinfo {author} {\bibfnamefont {Z.}~\bibnamefont {Cai}},\ and\ \bibinfo {author} {\bibfnamefont {Z.-W.}\ \bibnamefont {Liu}},\ }\href {https://doi.org/10.48550/arXiv.2509.21111} {\bibinfo {title} {No universal purification in quantum mechanics}} (\bibinfo {year} {2025}),\ \Eprint {https://arxiv.org/abs/2509.21111} {arXiv:2509.21111 [quant-ph]} \BibitemShut {NoStop}%
\bibitem [{\citenamefont {Matsumoto}\ and\ \citenamefont {Hayashi}(2007)}]{matsumoto2007universal}%
  \BibitemOpen
  \bibfield  {author} {\bibinfo {author} {\bibfnamefont {K.}~\bibnamefont {Matsumoto}}\ and\ \bibinfo {author} {\bibfnamefont {M.}~\bibnamefont {Hayashi}},\ }\bibfield  {title} {\bibinfo {title} {Universal distortion-free entanglement concentration},\ }\href {https://doi.org/10.1103/PhysRevA.75.062338} {\bibfield  {journal} {\bibinfo  {journal} {Physical Review A}\ }\textbf {\bibinfo {volume} {75}},\ \bibinfo {pages} {062338} (\bibinfo {year} {2007})}\BibitemShut {NoStop}%
\bibitem [{Note1()}]{Note1}%
  \BibitemOpen
  \bibinfo {note} {An equivalent description is $V({{\protect \mathscr {K}}}) = \protect \operatorname {aff}({{\protect \mathscr {K}}}) - \tau _0$ for any fixed $\tau _0 \in {{\protect \mathscr {K}}}$.}\BibitemShut {Stop}%
\bibitem [{\citenamefont {Fang}\ \emph {et~al.}(2025)\citenamefont {Fang}, \citenamefont {Fawzi},\ and\ \citenamefont {Fawzi}}]{fang2025generalized}%
  \BibitemOpen
  \bibfield  {author} {\bibinfo {author} {\bibfnamefont {K.}~\bibnamefont {Fang}}, \bibinfo {author} {\bibfnamefont {H.}~\bibnamefont {Fawzi}},\ and\ \bibinfo {author} {\bibfnamefont {O.}~\bibnamefont {Fawzi}},\ }\href {https://doi.org/10.48550/arXiv.2411.04035} {\bibinfo {title} {Generalized quantum asymptotic equipartition}} (\bibinfo {year} {2025}),\ \Eprint {https://arxiv.org/abs/2411.04035} {arXiv:2411.04035 [quant-ph]} \BibitemShut {NoStop}%
\bibitem [{\citenamefont {Datta}(2009)}]{datta2009min}%
  \BibitemOpen
  \bibfield  {author} {\bibinfo {author} {\bibfnamefont {N.}~\bibnamefont {Datta}},\ }\bibfield  {title} {\bibinfo {title} {Min- and max-relative entropies and a new entanglement monotone},\ }\href {https://doi.org/10.1109/TIT.2009.2018325} {\bibfield  {journal} {\bibinfo  {journal} {IEEE Transactions on Information Theory}\ }\textbf {\bibinfo {volume} {55}},\ \bibinfo {pages} {2816} (\bibinfo {year} {2009})}\BibitemShut {NoStop}%
\bibitem [{\citenamefont {George}\ and\ \citenamefont {Chitambar}(2024)}]{george2024conerestricted}%
  \BibitemOpen
  \bibfield  {author} {\bibinfo {author} {\bibfnamefont {I.}~\bibnamefont {George}}\ and\ \bibinfo {author} {\bibfnamefont {E.}~\bibnamefont {Chitambar}},\ }\bibfield  {title} {\bibinfo {title} {Cone-restricted information theory},\ }\href {https://doi.org/10.1088/1751-8121/ad52d5} {\bibfield  {journal} {\bibinfo  {journal} {Journal of Physics A: Mathematical and Theoretical}\ }\textbf {\bibinfo {volume} {57}},\ \bibinfo {pages} {265302} (\bibinfo {year} {2024})}\BibitemShut {NoStop}%
\bibitem [{\citenamefont {Wang}\ and\ \citenamefont {Wilde}(2019)}]{wang2019resource}%
  \BibitemOpen
  \bibfield  {author} {\bibinfo {author} {\bibfnamefont {X.}~\bibnamefont {Wang}}\ and\ \bibinfo {author} {\bibfnamefont {M.~M.}\ \bibnamefont {Wilde}},\ }\bibfield  {title} {\bibinfo {title} {Resource theory of asymmetric distinguishability},\ }\href {https://doi.org/10.1103/PhysRevResearch.1.033170} {\bibfield  {journal} {\bibinfo  {journal} {Physical Review Research}\ }\textbf {\bibinfo {volume} {1}},\ \bibinfo {pages} {033170} (\bibinfo {year} {2019})}\BibitemShut {NoStop}%
\bibitem [{\citenamefont {Renes}(2016)}]{renes2016relative}%
  \BibitemOpen
  \bibfield  {author} {\bibinfo {author} {\bibfnamefont {J.~M.}\ \bibnamefont {Renes}},\ }\bibfield  {title} {\bibinfo {title} {Relative submajorization and its use in quantum resource theories},\ }\href {https://doi.org/10.1063/1.4972295} {\bibfield  {journal} {\bibinfo  {journal} {Journal of Mathematical Physics}\ }\textbf {\bibinfo {volume} {57}},\ \bibinfo {pages} {122202} (\bibinfo {year} {2016})}\BibitemShut {NoStop}%
\bibitem [{\citenamefont {{Lipka-Bartosik}}\ \emph {et~al.}(2024)\citenamefont {{Lipka-Bartosik}}, \citenamefont {Chubb}, \citenamefont {Renes}, \citenamefont {Tomamichel},\ and\ \citenamefont {Korzekwa}}]{lipka-bartosik2024quantum}%
  \BibitemOpen
  \bibfield  {author} {\bibinfo {author} {\bibfnamefont {P.}~\bibnamefont {{Lipka-Bartosik}}}, \bibinfo {author} {\bibfnamefont {C.~T.}\ \bibnamefont {Chubb}}, \bibinfo {author} {\bibfnamefont {J.~M.}\ \bibnamefont {Renes}}, \bibinfo {author} {\bibfnamefont {M.}~\bibnamefont {Tomamichel}},\ and\ \bibinfo {author} {\bibfnamefont {K.}~\bibnamefont {Korzekwa}},\ }\bibfield  {title} {\bibinfo {title} {Quantum dichotomies and coherent thermodynamics beyond first-order asymptotics},\ }\href {https://doi.org/10.1103/PRXQuantum.5.020335} {\bibfield  {journal} {\bibinfo  {journal} {PRX Quantum}\ }\textbf {\bibinfo {volume} {5}},\ \bibinfo {pages} {020335} (\bibinfo {year} {2024})}\BibitemShut {NoStop}%
\bibitem [{\citenamefont {Regula}\ \emph {et~al.}(2020)\citenamefont {Regula}, \citenamefont {Bu}, \citenamefont {Takagi},\ and\ \citenamefont {Liu}}]{regula2020benchmarking}%
  \BibitemOpen
  \bibfield  {author} {\bibinfo {author} {\bibfnamefont {B.}~\bibnamefont {Regula}}, \bibinfo {author} {\bibfnamefont {K.}~\bibnamefont {Bu}}, \bibinfo {author} {\bibfnamefont {R.}~\bibnamefont {Takagi}},\ and\ \bibinfo {author} {\bibfnamefont {Z.-W.}\ \bibnamefont {Liu}},\ }\bibfield  {title} {\bibinfo {title} {Benchmarking one-shot distillation in general quantum resource theories},\ }\href {https://doi.org/10.1103/PhysRevA.101.062315} {\bibfield  {journal} {\bibinfo  {journal} {Physical Review A}\ }\textbf {\bibinfo {volume} {101}},\ \bibinfo {pages} {062315} (\bibinfo {year} {2020})}\BibitemShut {NoStop}%
\bibitem [{\citenamefont {Regula}\ and\ \citenamefont {Takagi}(2021)}]{regula2021oneshot}%
  \BibitemOpen
  \bibfield  {author} {\bibinfo {author} {\bibfnamefont {B.}~\bibnamefont {Regula}}\ and\ \bibinfo {author} {\bibfnamefont {R.}~\bibnamefont {Takagi}},\ }\bibfield  {title} {\bibinfo {title} {One-shot manipulation of dynamical quantum resources},\ }\href {https://doi.org/10.1103/PhysRevLett.127.060402} {\bibfield  {journal} {\bibinfo  {journal} {Physical Review Letters}\ }\textbf {\bibinfo {volume} {127}},\ \bibinfo {pages} {060402} (\bibinfo {year} {2021})}\BibitemShut {NoStop}%
\bibitem [{\citenamefont {Fang}\ \emph {et~al.}(2019)\citenamefont {Fang}, \citenamefont {Wang}, \citenamefont {Tomamichel},\ and\ \citenamefont {Duan}}]{fang2019nonasymptotic}%
  \BibitemOpen
  \bibfield  {author} {\bibinfo {author} {\bibfnamefont {K.}~\bibnamefont {Fang}}, \bibinfo {author} {\bibfnamefont {X.}~\bibnamefont {Wang}}, \bibinfo {author} {\bibfnamefont {M.}~\bibnamefont {Tomamichel}},\ and\ \bibinfo {author} {\bibfnamefont {R.}~\bibnamefont {Duan}},\ }\bibfield  {title} {\bibinfo {title} {Non-asymptotic entanglement distillation},\ }\href {https://doi.org/10.1109/TIT.2019.2914688} {\bibfield  {journal} {\bibinfo  {journal} {IEEE Transactions on Information Theory}\ }\textbf {\bibinfo {volume} {65}},\ \bibinfo {pages} {6454} (\bibinfo {year} {2019})}\BibitemShut {NoStop}%
\bibitem [{\citenamefont {Regula}\ \emph {et~al.}(2018)\citenamefont {Regula}, \citenamefont {Fang}, \citenamefont {Wang},\ and\ \citenamefont {Adesso}}]{regula2018oneshot}%
  \BibitemOpen
  \bibfield  {author} {\bibinfo {author} {\bibfnamefont {B.}~\bibnamefont {Regula}}, \bibinfo {author} {\bibfnamefont {K.}~\bibnamefont {Fang}}, \bibinfo {author} {\bibfnamefont {X.}~\bibnamefont {Wang}},\ and\ \bibinfo {author} {\bibfnamefont {G.}~\bibnamefont {Adesso}},\ }\bibfield  {title} {\bibinfo {title} {One-shot coherence distillation},\ }\href {https://doi.org/10.1103/PhysRevLett.121.010401} {\bibfield  {journal} {\bibinfo  {journal} {Physical Review Letters}\ }\textbf {\bibinfo {volume} {121}},\ \bibinfo {pages} {010401} (\bibinfo {year} {2018})}\BibitemShut {NoStop}%
\bibitem [{\citenamefont {{Hisaharu Umegaki}}(1954)}]{umegaki1954conditional}%
  \BibitemOpen
  \bibfield  {author} {\bibinfo {author} {\bibnamefont {{Hisaharu Umegaki}}},\ }\bibfield  {title} {\bibinfo {title} {Conditional expectation in an operator algebra},\ }\href {https://doi.org/10.2748/tmj/1178245177} {\bibfield  {journal} {\bibinfo  {journal} {Tohoku Mathematical Journal}\ }\textbf {\bibinfo {volume} {6}},\ \bibinfo {pages} {177} (\bibinfo {year} {1954})}\BibitemShut {NoStop}%
\bibitem [{\citenamefont {Fang}\ and\ \citenamefont {Hayashi}(2025)}]{fang2025error}%
  \BibitemOpen
  \bibfield  {author} {\bibinfo {author} {\bibfnamefont {K.}~\bibnamefont {Fang}}\ and\ \bibinfo {author} {\bibfnamefont {M.}~\bibnamefont {Hayashi}},\ }\href {https://doi.org/10.48550/arXiv.2508.12901} {\bibinfo {title} {Error exponents of quantum state discrimination with composite correlated hypotheses}} (\bibinfo {year} {2025}),\ \Eprint {https://arxiv.org/abs/2508.12901} {arXiv:2508.12901 [quant-ph]} \BibitemShut {NoStop}%
\end{thebibliography}%

\clearpage
\onecolumngrid

\newgeometry{lmargin=1.05in,rmargin=1.05in,tmargin=1.in,bmargin=1in}

\let\addcontentsline\oldaddcontentsline

\setcounter{secnumdepth}{2}
\setcounter{tocdepth}{2}

\setcounter{section}{0}
\setcounter{figure}{0}
\setcounter{table}{0}
\setcounter{equation}{0}
\setcounter{theorem}{0}
\setcounter{definition}{0}

\renewcommand{\thesection}{S\arabic{section}}
\renewcommand{\thefigure}{S\arabic{figure}}
\renewcommand{\thetable}{S\arabic{table}}
\renewcommand{\thetheorem}{S\arabic{theorem}}
\renewcommand{\thedefinition}{S\arabic{definition}}

\begin{center}
    {\large \textbf{Supplemental Material}}
\end{center}

\vspace{0.5cm}

In this Supplemental Material, we provide more detailed expositions, proofs and discussions of the results in the main text. We may reiterate some of the steps and statements to ensure that the Supplemental Material are explicit and self-contained.

\tableofcontents

\section{Preliminaries}

\subsection{Notation}

The sets of real and natural numbers are denoted by $\RR$ and $\NN$, respectively. Unless otherwise specified, all logarithms are taken to base two and written as $\log$. Quantum systems are labeled by capital letters such as $A$, $B$, and $C$, and the finite-dimensional Hilbert space associated with system $A$ is denoted by $\cH_A$. The set of linear operators on $\cH_A$ is denoted by $\sL(A)$. The sets of Hermitian and positive semidefinite operators on $\cH_A$ are denoted by $\HERM(A)$ and $\PSD(A)$, respectively. The set of density operators (i.e., positive semidefinite operators with unit trace) on $\cH_A$ is denoted by $\density(A)$. When the underlying system is clear from context, we omit the system label and simply write $\sL, \HERM, \PSD$ and $\density$. For $X,Y\in\sL$, we write $X\geq Y$ if $X-Y \in \PSD$. Script letters such as $\sP, \sE$, and $\sK$ denote sets of linear operators, or sequences of such sets when the context is clear. For any set of density operators $\sP \subseteq \density$, we define its convex hull and affine hull by
\begin{align}
    \conv(\sP) & := \left\{\sum_{i=1}^k c_i \rho_i: k\in\NN,\, \rho_i \in \sP,  c_i \geq 0,\,\, \sum_{i=1}^k c_i = 1\right\},\\
    \aff(\sP) & := \left\{\sum_{i=1}^k a_i \rho_i : k\in\NN, \rho_i \in \sP, a_i \in \RR,\, \sum_{i=1}^k a_i = 1\right\},
\end{align}
respectively. For $X, Y \in \sL$, their trace distance is defined as 
\begin{align} 
    T(X,Y) := \tfrac{1}{2}\|X - Y\|_1.
\end{align} 
Given $\sP \subseteq \density$ and $\ve \in [0,1)$, we define the $\ve$-ball around $\sP$ as 
\begin{align} 
    \sB_\ve(\sP) := \{\omega \in \density : T(\omega,\rho) \leq \ve \text{ for some } \rho \in \sP\}.
\end{align}

We use $\fL(A\to A^\prime)$ to denote the set of all linear maps from $\sL(A)$ to $\sL(A^\prime)$, and $\fF(A\to A^\prime)$ to denote a specified subset of such maps. The set of completely positive trace-preserving maps from $\sL(A)$ to $\sL(A^\prime)$ is denoted by $\CPTP(A\to A^\prime)$. When the input and output systems are clear from context, we simply write $\fL$, $\fF$ and $\CPTP$. Calligraphic letters like $\cL, \cF$ and $\cE$ denote individual maps in these classes.

\subsection{Max- and min-relative entropies}

Throughout this work, we used the standard smoothed max- and min-relative entropies between two sets of quantum states, together with their new subspace-constrained variants. For convenience, we collect the definitions of the standard smoothed max- and min-relative entropies in this section. The subspace-constrained variants will be introduced later when they are required.

\begin{definition}[Smoothed max-relative entropy between two sets~\cite{fang2025generalized}]
    \label{def:smoothed max relative entropy}
    Let $\sP, \sE \subseteq \density$ be two sets of quantum states. For any $\ve \in [0,1)$, the smoothed max-relative entropy between $\sP$ and $\sE$ is defined as
    \begin{align}\label{eq:smoothed max relative entropy}
        D_{\max,\ve}(\sP\|\sE):= \log\inf_{\substack{\tau\in\sE \\ \omega\in\sB_{\ve}(\sP)}}\left\{M > 1: M\tau - \omega \in \PSD \right\}.
    \end{align}
\end{definition}

When both $\sP$ and $\sE$ are singletons, this definition reduces to the standard smoothed max-relative entropy between two states~\cite{datta2009min}.

\begin{definition}[Smoothed min-relative entropy between two sets~\cite{watanabe2024black,fang2025generalized}]
    \label{def:smoothed min relative entropy}
    Let $\sP, \sE \subseteq \density$ be two sets of quantum states. For any $\ve \in [0,1)$, the smoothed min-relative entropy between $\sP$ and $\sE$ is defined as
    \begin{align}
        \label{eq:smoothed min relative entropy}
        D_{\min,\ve}(\sP\|\sE) := -\log \min_{\substack{0 \leq E \leq I}} \left\{\sup_{\tau \in \sE}\tr[\tau E] : \sup_{\rho \in \sP} \tr[\rho(I-E)] \leq \ve\right\}.
    \end{align}
\end{definition}

It is straightforward to verify that $D_{\min,\ve}(\sP\|\sE)$ is unchanged when $\sP$ and $\sE$ are replaced by their convex hulls, i.e., $D_{\min,\ve}(\sP\|\sE) = D_{\min,\ve}(\conv(\sP)\|\conv(\sE))$. In case where $\sP = \{\rho\}$ and $\sE =\{\tau\}$ are singletons, this definition recovers the standard smoothed min-relative entropy (also known as hypothesis testing relative entropy) between two states $D_{\min,\ve}(\rho\|\tau)$. 

We will also use the same optimization formula in an extended sense, where the second argument is allowed to be a Hermitian trace-one operator, or more generally a set of such operators. If the optimized value inside the logarithm is nonpositive, we set the corresponding quantity to be $+\infty$. The following lemma establishes a faithfulness property for this extended quantity.

\begin{lemma}[Faithfulness of min-relative entropy]
\label{lem:faithfulness of min-relative entropy}
    Let $\ve \in (0,1)$, $\rho \in \density$, and $X \in \HERM$ with $\tr[X]=1$. Then $D_{\min,\ve}(\rho\|X) = -\log(1-\ve)$ if and only if $\rho = X$.
\end{lemma}
\begin{proof}
    ($\Rightarrow$, by contrapositive) Suppose that $\rho\neq X$. Define $\Delta:=\rho-X\neq 0$ and $H:=\Delta - \frac{\tr[\Delta\rho]+\tr[\Delta X]}{2}\,I$. Since $\tr[\rho]=\tr[X]=1$, a direct calculation shows that
    \begin{align}
        \tr[H\rho]=\tfrac{1}{2} \tr[\Delta^2]>0, \qquad \tr[H X]=-\tfrac{1}{2}\tr[\Delta^2]<0.
    \end{align}
    For $t > 0$, define the operator $E_t:= (1-\ve)I+tH$. Since $H$ is Hermitian and bounded, we have
    \begin{align}
        (1-\ve -t\|H\|) I \leq E_t \leq (1-\ve + t\|H\|) I.
    \end{align}
    Thus, choosing any $0< t < \min\left\{\frac{1-\ve}{\|H\|},\,\frac{\ve}{\|H\|}\right\}$ ensures that $0\leq E_t \leq I$. For any such $t$, we have
    \begin{align}
        \tr[(I-E_t)\rho] = \tr[(\ve I- tH)\rho] = \ve - t\tr[H\rho] < \ve,
    \end{align}
    so $E_t$ is a feasible test operator. Now let
    \begin{align}
        \alpha := \min_{\substack{0\leq E\leq I \\ \tr[(I-E)\rho]\leq\ve}} \tr[E X].
    \end{align}
    Since $E_t$ is feasible, 
    \begin{align}
        \alpha \leq \tr[E_t X] = (1-\ve) + t\tr[H X]  < 1-\ve
    \end{align} 
    If $\alpha \leq 0$, then by convention $D_{\min,\ve}(\rho\|X) = +\infty$. If $\alpha > 0$, then 
    \begin{align}
        D_{\min,\ve}(\rho\|X) = -\log \alpha > -\log(1-\ve).
    \end{align}
    Thus, in either case, $D_{\min,\ve}(\rho\|X) \neq -\log(1-\ve)$. This proves the contrapositive. \smallskip

    ($\Leftarrow$) If $\rho = X$, then for every $0\leq E\leq I$ with $\tr[(I-E)\rho]\leq\ve$,
    \begin{align}
        \tr[E X] = \tr[E\rho] = 1-\tr[(I-E)\rho] \geq 1-\ve.
    \end{align}
    Moreover, the choice $E=(1-\ve)I$ is feasible and achieves $\tr[E X]=1-\ve$. Hence, $D_{\min,\ve}(\rho\|X) = -\log(1-\ve)$. This completes the proof.
\end{proof}

\subsection{Standard quantum thermodynamic framework}

In this section, we briefly review the standard resource-theoretic framework of quantum thermodynamics, which motivates some definitions used in the main text and serves as the basis for our subsequent extension to scenarios with uncertainty. \medskip 

\paragraph*{Free states.} Consider a quantum system $A$ with Hamiltonian $H^A$ in contact with a heat bath at inverse temperature $\beta = 1/(k_B T)$, where $T$ is the bath temperature and $k_B$ is the Boltzmann constant. In this setting, the unique free state is the thermal equilibrium state, or Gibbs state, defined as
\begin{align}
    \tau^A := \frac{e^{-\beta H^A}}{\tr[e^{-\beta H^A}]}.
\end{align}

\paragraph*{Free operations.} The golden rule of any well-defined resource theory is that free operations map free states to free states. In quantum thermodynamics, the class of free operations that most directly reflects physical implementability is given by thermal operations (TO)~\cite{brandao2013resource,horodecki2013fundamental,janzing2000thermodynamic}. A CPTP map $\cE\in\CPTP(A\to A^\prime)$ is a thermal operation if it admits the form
\begin{align}
    \cE(\cdot) = \tr_{(AE)\setminus A^\prime}[U(\cdot \ox \tau^E)U^\dagger],
\end{align}
where $\tau^E$ is the Gibbs state of an ancillary system $E$ with Hamiltonian $H^E$, and $U$ is a joint unitary that preserves total energy:
\begin{align}
    [U, H^A \ox I^E + I^A \ox H^E] = 0.
\end{align}
Every thermal operation is Gibbs-preserving $\cE(\tau^A) = \tau^{A^\prime}$, and time-translation covariant, i.e.,
\begin{align}
    \cE(e^{-iH^A t}\rho^A e^{iH^A t}) = e^{-iH^{A^\prime}t} \cE(\rho^A)e^{iH^{A^\prime}t}, \qquad \forall t\in\RR.
\end{align}
These two properties motivate successively larger classes of free operations. The class of Gibbs-preserving covariant operations (GPC) consists of all CPTP maps satisfying both Gibbs preservation and time-translation covariance. Dropping the covariance requirement yields the broader class of Gibbs-preserving operations (GPO),
\begin{align}
    \GPO(A\to A^\prime) := \left\{\cE \in \CPTP(A\to A^\prime): \cE(\tau^A) = \tau^{A^\prime}\right\}.
\end{align}

\paragraph*{Athermal states.} In this framework, the thermodynamic resource of a system $A$ is not determined solely by its density operator $\rho^A$, but also by the equilibrium state $\tau^A$ against which the deviation is measured. An athermal state is therefore represented by a pair $(\rho^A,\tau^A)$ of nonequilibrium and equilibrium states. Under this convention, a free state is represented by $(\tau^A, \tau^A)$. \medskip

\paragraph*{Athermality transformations.} Given a class of free operations $\fF(A \to A^\prime)$, let $(\rho, \tau) \in \density(A) \times \density(A)$ and $(\rho^\prime, \tau^\prime) \in \density(A^\prime) \times \density(A^\prime)$ be two athermal states. For any $\ve \in [0,1)$, we write
\begin{align}
    (\rho,\tau) \xrightarrow[]{\cF,\; \ve} (\rho^\prime, \tau^\prime),
\end{align}
if there exists $\cF \in \fF(A \to A^\prime)$ such that $\cF(\rho)\in \sB_\ve(\rho^\prime)$ and $\cF(\tau) = \tau^\prime$. \medskip

\paragraph*{(Clean) battery.} To connect athermal state transformations with an operational notion of work, we introduce an explicit work storage system, i.e., a battery. Following Refs.~\cite{gour2022role,watanabe2024black}, we model the battery as a two-level system $B$ with Hilbert space $\cH_B = \operatorname{span}\{\ket{0}, \ket{1}\}$ and Hamiltonian parameterized by a real number $M > 1$,
\begin{align}
    H^B_M = E_{M,0} \ket{0}\bra{0}^B + E_{M,1} \ket{1}\bra{1}^B, \qquad E_{M,1} - E_{M,0} = \frac{1}{\beta}\log(M-1).
\end{align}
With this choice, the Gibbs state of the battery takes the form
\begin{align}
    \pi_M^B = \left(1-\frac{1}{M}\right) \ket{0}\bra{0}^B + \frac{1}{M}\ket{1}\bra{1}^B. 
\end{align}
For a fixed inverse temperature $\beta$, the parameter $M$ determines the energy gap of the battery Hamiltonian and thereby quantifies its work capacity. Using the above notation, a fully charged battery is represented by the athermal state $(\ket{1}\bra{1}^B, \pi_M^B)$, whereas an empty battery corresponds to the free state $(\pi_M^B, \pi_M^B)$. We refer to this model as a ``clean'' battery, since its Gibbs state is exactly known. \medskip

\paragraph*{Work transformations.} With the battery model in place, we can formulate work transformation tasks. We focus on two converse processes: work extraction and work of formation. Together, these tasks provide an operational way to quantify the usefulness of athermality resources and to assess the reversibility of thermodynamic processes. Concretely, consider a primary system $P$ with finite-dimensional Hilbert space $\cH_P$ and Hamiltonian $H^P$, together with the battery system $B$ introduced above. We assume that the total Hamiltonian of the joint system $PB$ is non-interacting both initially and finally, i.e.,
\begin{align}
    H^{PB}_M = H^P \ox I^B + I^P \ox H_M^B.
\end{align}
Thus the corresponding Gibbs state factorizes as
\begin{align}
    \tau^{PB}_M = \tau^P \ox \pi_M^B.
\end{align}
Suppose the primary system is initially in state $\rho^P$ and the battery starts in its Gibbs state $\pi_M^B$. If there exists a free operation $\cF \in \fF(PB \to PB)$ such that
\begin{align}
    (\rho^P \ox \pi_M^B, \tau^P \ox \pi_M^B) \xrightarrow[]{\cF, \;\ve} (\tau^P\ox\ket{1}\bra{1}^B, \tau^P\ox \pi^B_M),
\end{align}
then the nonequilibrium resource of $\rho^P$ has been converted into the excitation of the battery, and we say that the battery with capacity $M$ can be charged by $\rho^P$ under $\fF$. Conversely, if there exists a free operation $\cF \in \fF(PB \to PB)$ such that
\begin{align}
    (\tau^P\ox\ket{1}\bra{1}^B, \tau^P\ox \pi^B_M) \xrightarrow[]{\cF, \;\ve} (\rho^P \ox \pi_M^B, \tau^P \ox \pi_M^B),
\end{align}
then the resource stored in the battery is consumed to prepare the state $\rho^P$, and we say that $\rho^P$ can be formed from a battery with capacity $M$ under $\fF$. 

Since appending and discarding Gibbs subsystems are Gibbs-preserving operations, the above description is equivalent to a reduced formulation in which the transformation acts directly between the primary system and the battery~\cite{gour2022role}. This is the formulation adopted in the main text. In this reduced picture, the work extraction and work of formation tasks take the form
\begin{align}
    (\rho^P, \tau^P) \xrightarrow[]{\cF, \; \ve} (\ket{1}\bra{1}^B, \pi_M^B) \quad \text{and} \quad (\ket{1}\bra{1}^B, \pi_M^B) \xrightarrow[]{\cF^\prime, \; \ve} (\rho^P, \tau^P),
\end{align}
where the two transformations are implemented by operations $\cF\in\fF(P\to B)$ and $\cF^\prime\in\fF(B\to P)$, respectively. When the input and output systems are clear from context, we omit superscripts for simplicity.

\section{Quantum thermodynamics with uncertainty}

We now introduce our framework for quantum thermodynamics with uncertainty. In contrast to the standard setting, we consider scenarios in which the experimenter has only partial information about both the nonequilibrium and equilibrium states of the primary system $P$.

\subsection{Framework with uncertainty}

\paragraph*{Uncertain nonequilibrium.} Such uncertainty may arise in several ways. The nonequilibrium state $\rho$ may be prepared imperfectly, so that the realized state deviates from the intended one. Alternatively, the experimenter may have no direct access to the preparation procedure; for instance, the system may be prepared by an external agent and later handed to the experimenter, who is only informed that the state belongs to a prescribed set of states $\sP\subseteq\density$. This description is known as a black-box or state-agnostic setting studied in recent works~\cite{watanabe2024black,watanabe2026universal}. \medskip

\paragraph*{Uncertain thermal equilibrium.} The equilibrium reference may itself be uncertain. Owing to unavoidable noise and incomplete knowledge of the experimental conditions, the underlying Hamiltonian may be known only to belong to a given class, or the bath temperature may be known only to lie within a prescribed interval. In either case, the corresponding Gibbs state is no longer uniquely specified and must instead be regarded as an element of a set of candidate equilibrium states $\sE\subseteq\density$. This description is a central ingredient of our framework and is essential for capturing the limitations of thermodynamic protocols under realistic conditions. \medskip

\paragraph*{Uncertain athermal state.} Since athermality is defined relative to an equilibrium reference, uncertainty in both the nonequilibrium state and the equilibrium state implies that the resource can no longer be represented by a single pair $(\rho,\tau)$. Instead, it must be described by a family of candidate pairs. To capture this situation, we model an uncertain athermal state as a black box specified by two sets of density operators $(\sP,\sE)$, where $\sP\subseteq\density$ and $\sE\subseteq\density$ encode the knowledge of the possible nonequilibrium and equilibrium states, respectively. A comparison of the standard athermality and uncertain athermality is given in Fig.~\ref{fig:uncertainty}. \medskip

\begin{figure}[H]
    \centering
    \includegraphics[width=0.45\linewidth]{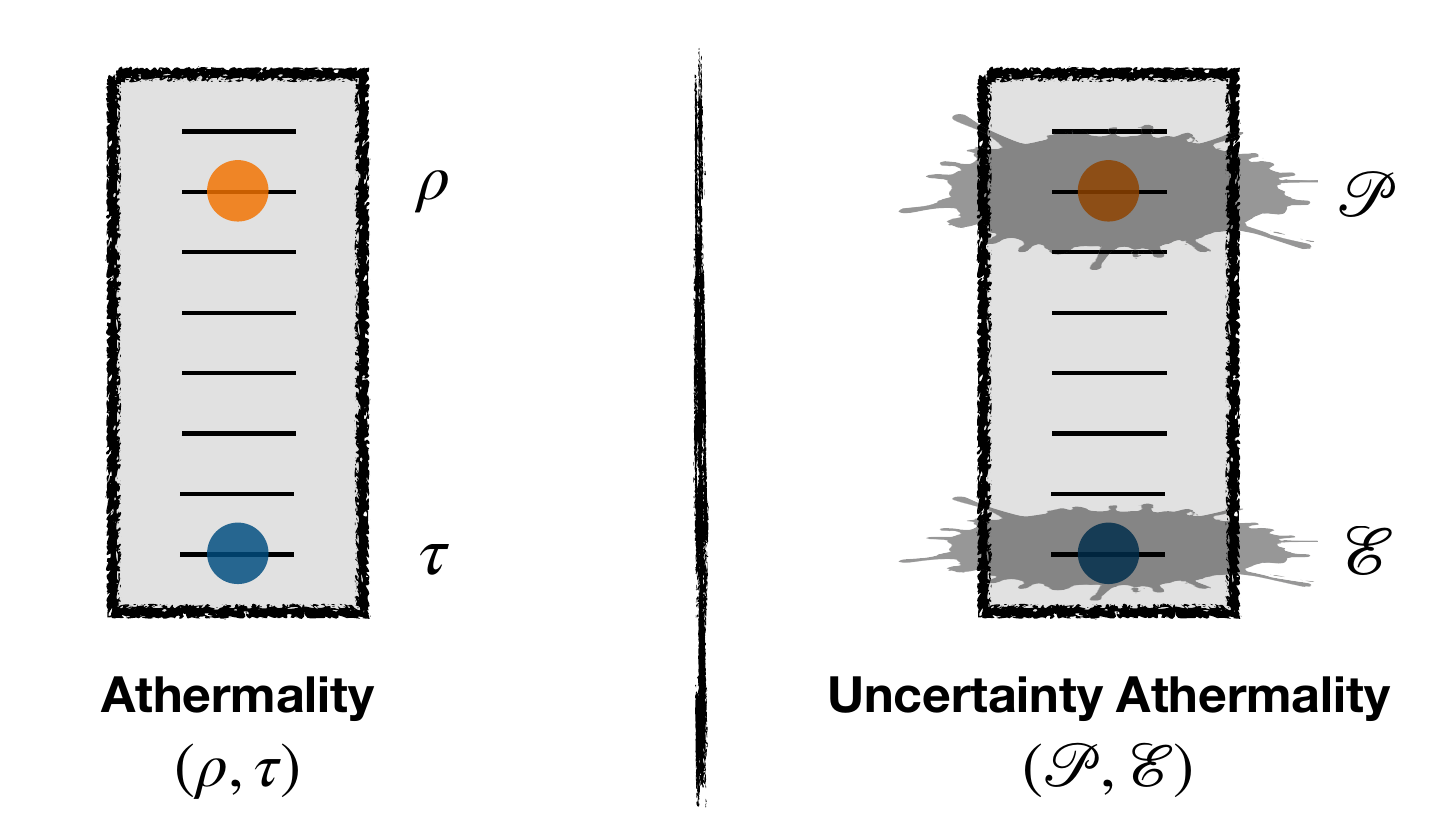}
    \caption{Comparison of the standard and uncertain athermal states. (Left) A standard athermal state is described by a single pair $(\rho, \tau) \in \density\times \density$ of nonequilibrium and equilibrium states. (Right) An uncertain athermal state is described by a pair of candidate sets $(\sP, \sE)$, where $\sP\subseteq\density$ and $\sE\subseteq\density$ represent the experimenter's partial knowledge of the possible nonequilibrium and equilibrium states, respectively.}
    \label{fig:uncertainty}
\end{figure}

\paragraph*{(Dirty) battery.} In practice, the physical system used to implement the battery may itself be imperfectly calibrated or subject to parameter drift. This motivates the notion of a \emph{dirty battery}, whose equilibrium state is not specified exactly but is instead described by a set of candidate Gibbs states. In particular, we consider the dirty battery
\begin{align}
    (\ket{1}\bra{1}, \Pi_M), \quad \Pi_M := \{\pi_{M^\prime}: M^\prime \in [M, \infty)\}.
\end{align}
Operationally, this one-sided uncertainty means that the experimenter can certify only a lower bound on the battery's work capacity, while the exact value is unknown. \medskip

\paragraph*{Uncertain athermality transformation.} Since the realized input $(\rho,\tau)\in\sP\times\sE$ is not known in advance, a valid transformation must be implemented by a single operation that works for all admissible candidates. Formally, let $\sP, \sE \subseteq \density(A)$ and $\sP^\prime, \sE^\prime \subseteq \density(A^\prime)$. Given a class of free operations $\fF(A \to A^\prime)$ and an error tolerance $\ve \in [0,1)$, we write
\begin{align}\label{eq:sets conversion}
    (\sP,\sE) \xrightarrow[]{\cF,\; \ve} (\sP^\prime, \sE^\prime),
\end{align}
if there exists a single operation $\cF\in\fF$ such that, for every $(\rho, \tau) \in \sP \times \sE$, one can find a target pair $(\rho^\prime, \tau^\prime) \in \sP^\prime \times \sE^\prime$ satisfying $(\rho,\tau) \xrightarrow[]{\cF,\; \ve} (\rho^\prime, \tau^\prime)$. In particular, if any of the sets involved is a singleton, we identify it with its unique element whenever this causes no confusion. For example, we write $\rho$ in place of $\{\rho\}$ for simplicity. 

As in the standard setting, work transformations can be formulated within this general framework. An illustration is provided in Figure~\ref{fig:settings}. For a clean battery, the work extraction and work of formation tasks are,
\begin{align}
    (\sP, \sE) \xrightarrow[]{\cF, \;\ve} (\ket{1}\bra{1}, \pi_M),  \quad \text{and} \quad (\ket{1}\bra{1}, \pi_M) \xrightarrow[]{\cF^\prime, \;\ve} (\sP, \sE),
\end{align}
respectively. For a dirty battery, the corresponding tasks take the form
\begin{align}
    (\sP, \sE) \xrightarrow[]{\cF, \;\ve} (\ket{1}\bra{1}, \Pi_M),  \quad \text{and} \quad (\ket{1}\bra{1}, \Pi_M) \xrightarrow[]{\cF^\prime, \;\ve} (\sP, \sE).
\end{align}

\begin{figure}[H]
    \centering
    \includegraphics[width=0.85\linewidth]{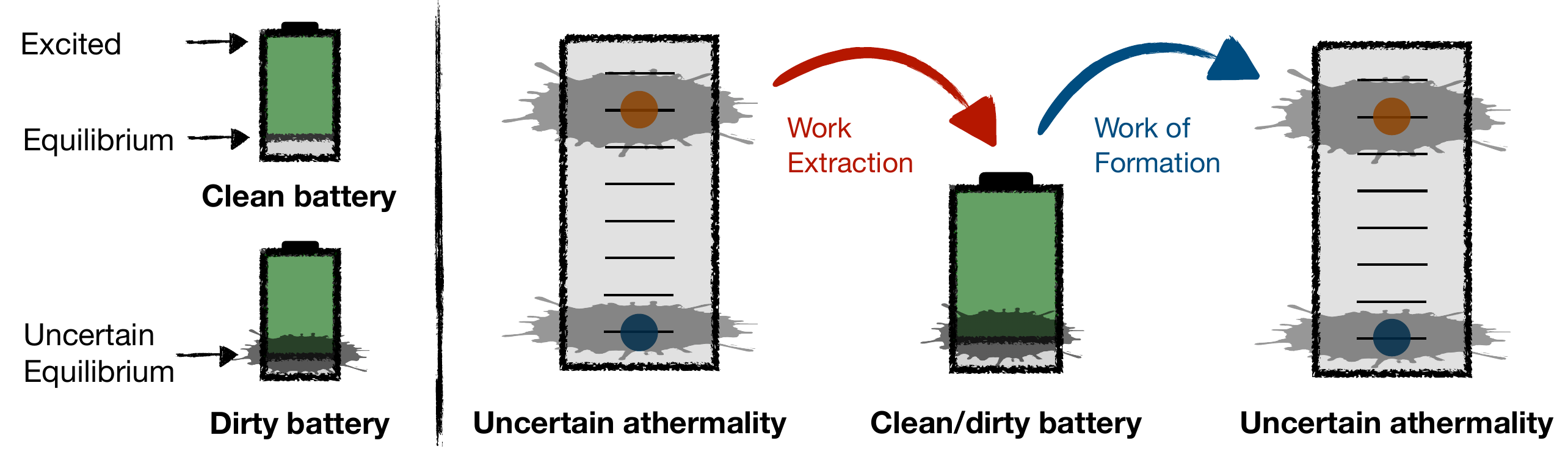}
    \caption{Illustration of work transformations with clean and dirty batteries. (Left) A clean battery has precisely known equilibrium state, while a dirty battery has uncertain equilibrium state. (Right) Work extraction converts an athermal state into an excited battery, whereas work of formation consumes an excited battery to prepare an athermal state.}
    \label{fig:settings}
\end{figure}

\paragraph*{Gibbs-preserving linear maps.} To establish no-go limitations in maximal generality, we impose only the minimal structural requirements on the allowed transformations. Since Gibbs preservation is the minimal thermodynamic constraint underlying the standard classes of free operations, we introduce the class of \emph{Gibbs-preserving linear maps} (GPL), which retains this requirement while dropping all additional operational constraints, including complete positivity, trace preservation, and time-translation covariance.

\begin{definition}[Gibbs-preserving linear maps]
    For an input system $A$ and an output system $A^\prime$, the class of Gibbs-preserving linear maps from $A$ to $A^\prime$ is defined as
    \begin{align}
        \GPL(A\to A^\prime) := \left\{\cL \in \fL(A\to A^\prime): \cL(\tau^A) = \tau^{A^\prime}\right\}.
    \end{align}
\end{definition}

Any impossibility result established under GPL immediately rules out the same transformation under all smaller, physically motivated classes of free operations. Moreover, GPL naturally accommodates potential extensions of free operations, such as post-selected probabilistic protocols~\cite{alhambra2016fluctuating} and virtual operations that combine quantum operations with classical statistical post-processing~\cite{takagi2024virtual,yuan2024virtual,zhu2024reversing}. Together with the standard classes introduced above, one has the inclusion hierarchy
\begin{align}
    \TO(A\to A^\prime) \subsetneq \GPC(A\to A^\prime) \subsetneq \GPO(A\to A^\prime) \subsetneq \GPL(A\to A^\prime).
\end{align}
Throughout, our no-go results are established for the largest class GPL, ensuring they hold a fortiori for all smaller classes. Conversely, our achievability results are constructed within conventional classes such as $\TO$ or $\GPO$.

\subsection{No-go theorem for athermality ``purification''}

\begin{restate}{Theorem}{1}[No-go theorem for athermality ``purification'']
    \label{thm:no-go theorem for transformation with uncertainty}
    Let $\ve \in [0, 1)$, and let $(\rho^\prime, \tau^\prime) \in \density \times \density$. Suppose that $\sP, \sE \subseteq \density$ satisfy $\conv(\sP) \cap \aff(\sE) \neq \emptyset$. Then $(\sP, \sE) \xrightarrow[]{\cL,\;\ve} (\rho^\prime, \tau^\prime)$ is achievable by some $\cL \in \GPL$ if and only if $\ve \geq T(\rho^\prime, \tau^\prime)$. 
\end{restate}
\begin{proof}
    ($\Rightarrow$) Suppose that there exists a state $\omega \in \conv(\sP) \cap \aff(\sE)$, so $\omega$ admits both convex decomposition over $\sP$ and affine decomposition over $\sE$,
    \begin{align}
        \omega = \sum_i c_i \rho_i = \sum_j a_j \tau_j,
    \end{align}
    where $\rho_i \in \sP$, $\tau_j \in \sE$, $c_i \ge 0$, and $\textstyle\sum_i c_i = \textstyle\sum_j a_j = 1$. Now suppose there exists a Gibbs-preserving linear map $\cL \in \GPL$ such that for any pair $(\rho, \tau) \in \sP \times \sE$,
    \begin{align}
        T(\cL(\rho), \rho^\prime) \leq \ve, \qquad \cL(\tau) = \tau^\prime.
    \end{align}
    Applying $\cL$ to the convex decomposition of $\omega$ over $\sP$, and using linearity of $\cL$ together with convexity of the trace distance, we obtain
    \begin{align}
        T(\cL(\omega), \rho^\prime) = T\left(\sum_i c_i \cL(\rho_i), \rho^\prime\right) \leq \sum_i c_i T(\cL(\rho_i), \rho^\prime) \leq \sum_i c_i \ve =\ve.
    \end{align}
    Similarly, applying $\cL$ to $\omega$ with the affine decomposition over $\sE$ and using the linearity of $\cL$ gives
    \begin{align}
        \cL(\omega) = \sum_j a_j \cL(\tau_j) = \sum_j a_j \tau^\prime = \tau^\prime.
    \end{align}
    Combining the above two relations, we get
    \begin{align}
        T(\rho^\prime, \tau^\prime) = T(\rho^\prime, \cL(\omega)) \leq \ve.
    \end{align}
    
   ($\Leftarrow$) Suppose $T(\rho^\prime, \tau^\prime) \leq \ve$ and consider the replacer map $\cL(\cdot) = \tr(\cdot)\tau^\prime$. For any $(\rho, \tau) \in \sP \times \sE$, it holds that
    \begin{align}
        T(\cL(\rho), \rho^\prime) = T(\tau^\prime, \rho^\prime) \leq \ve, \qquad \cL(\tau) = \tau^\prime.
    \end{align}
    Thus the desired conversion is achievable. This completes the proof.
\end{proof}

\begin{remark}
    We claim that the replacer map $\cL(\cdot) = \tr(\cdot)\tau^\prime$ is a thermal operation. To see this, choose an ancillary system $E$ whose Gibbs state is $\tau^E = \tau^\prime$. Then, for any input state $\rho^P$,
    \begin{align}
        \cL(\rho^P) = \tr_P[I(\rho^P\ox \tau^E)I] = \tau^E =\tau^\prime
    \end{align}
    where $I$ is the identity operator, which trivially satisfies $[I, H^P\ox I^E+I^P\ox H^E] = 0$. 
\end{remark}

\begin{remark}
    For exact conversion ($\ve = 0$), the condition can be relaxed to $\aff(\sP) \cap \aff(\sE) \neq \emptyset$. The proof is the same except that the convex decomposition of $\omega$ over $\sP$ is replaced by an affine decomposition, and the convexity of trace distance is not longer needed.
\end{remark}

This theorem yields a sharp limitation on athermality ``purification'': under the stated geometric condition, the conversion is either trivial or impossible, with no intermediate tradeoff. An illustration is shown in Fig.~\ref{fig:no-go}. Specifically, if the target athermal state satisfies $T(\rho^\prime, \tau^\prime) \leq \ve$, then
\begin{align}
    (\rho^\prime, \tau^\prime) \xrightarrow[]{\id,\;\ve} (\tau^\prime, \tau^\prime), \qquad (\tau^\prime, \tau^\prime) \xrightarrow[]{\id,\;\ve} (\rho^\prime, \tau^\prime),
\end{align}
where $\id$ is the identity operation. Hence, the achievable target $(\rho^\prime,\tau^\prime)$ is operationally equivalent, up to error $\ve$, to the free state $(\tau^\prime,\tau^\prime)$. In this case, the conversion can be achieved by the trivial thermal operation $\cL(\cdot) = \tr(\cdot) \tau^\prime$, which simply discards the input and prepares $\tau^\prime$. Conversely, for any target with nontrivial athermality $T(\rho^\prime,\tau^\prime) > \ve$, Theorem~\ref{thm:no-go theorem for transformation with uncertainty} precludes conversion from any uncertain primitive $(\sP,\sE)$ satisfying the geometric condition. In particular, since any achievable conversion to a target $(\rho^\prime,\tau^\prime)$ must incur error at least $T(\rho^\prime,\tau^\prime)$, exact purification to any non-free target is impossible. 

\begin{figure}[H]
    \centering
    \includegraphics[width=0.65\linewidth]{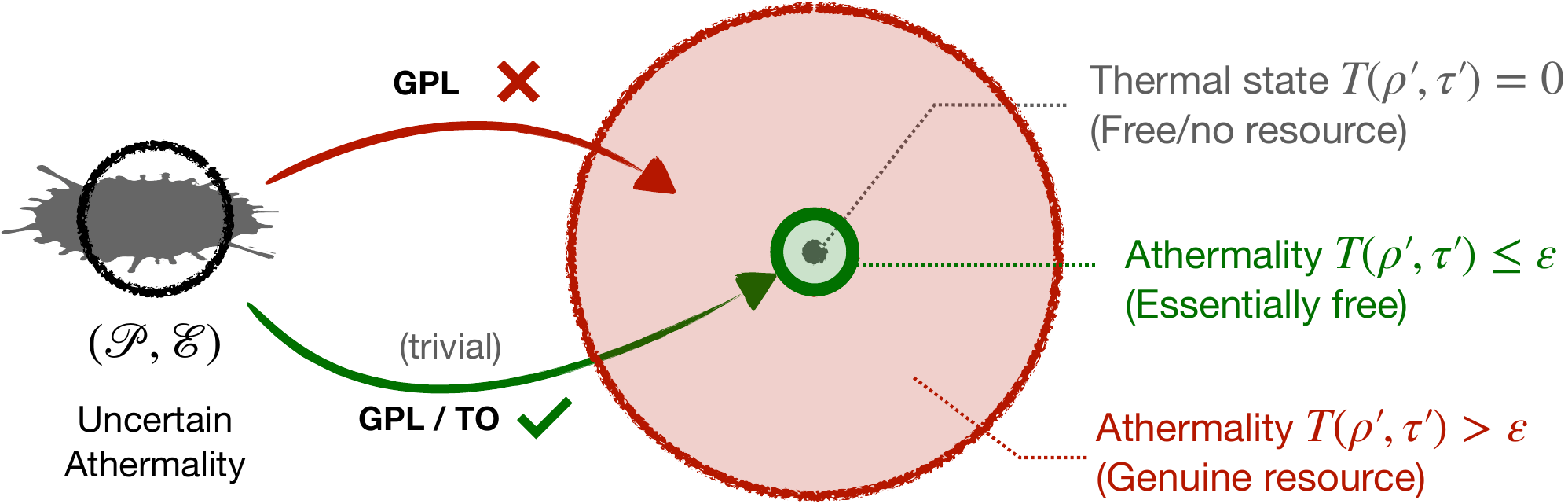}
    \caption{Illustration of the no-go theorem for athermality ``purification''. The conversion is achievable if and only if the target state is $\ve$-close to the thermal state (free state).}
    \label{fig:no-go}
\end{figure}

The geometric condition $\conv(\sP) \cap \aff(\sE) \neq \emptyset$ may appear technical at first sight, but it can be very generic in practice. As shown in the following example, even when the nominal Hamiltonian is known, unavoidable fluctuations or incomplete knowledge of the control parameters may leave an arbitrarily small but full-dimensional uncertainty region around it. The induced family of Gibbs states then has full affine span, $\aff(\sE)=\aff(\density)$, making the condition automatic.

\begin{example}
    Consider a qubit coupled to an external field with Hamiltonian
    \begin{align}
        H(\bm{h}) = -\bm{h}\cdot\bm{\sigma} = -h_x \sigma_x - h_y \sigma_y - h_z \sigma_z,
    \end{align}
    where $\bm{h} \in\RR^3$ denotes the field vector, and $\sigma_x, \sigma_y, \sigma_z$ are the Pauli matrices. Let
    \begin{align}
        |\bm{h}|=\sqrt{h_x^2+h_y^2+h_z^2}, \qquad \hat{\bm r}=\frac{\bm h}{|\bm h|}=\frac{1}{|\bm h|}(h_x,h_y,h_z),
    \end{align}
    so that $\hat{\bm r}$ is a unit vector and $\bm h\cdot \bm \sigma = |\bm h|\,\hat{\bm r}\cdot \bm \sigma$. Since $\hat{\bm r}$ is a unit vector, we have the identity $(\hat{\bm r}\cdot\bm\sigma)^2=I$, which implies $(\hat{\bm r}\cdot\bm\sigma)^{2k}=I$ and $(\hat{\bm r}\cdot\bm\sigma)^{2k+1}=\hat{\bm r}\cdot\bm\sigma$ for every non-negative integer $k$. Expanding the matrix exponential as a power series and separating even and odd powers, we obtain
    \begin{align}
        e^{x \hat{\bm r}\cdot\bm\sigma} = \sum_{n=0}^{\infty}\frac{x^n}{n!}(\hat{\bm r}\cdot\bm\sigma)^n = \sum_{k=0}^{\infty}\frac{x^{2k}}{(2k)!}\,I + \sum_{k=0}^{\infty}\frac{x^{2k+1}}{(2k+1)!}\,\hat{\bm r}\cdot\bm\sigma = \cosh(x)\,I+\sinh(x)\,\hat{\bm r}\cdot\bm\sigma.
    \end{align}
    Hence,
    \begin{align}
        e^{-\beta H(\bm h)}
        =e^{\beta \bm h\cdot \bm\sigma}
        =e^{\beta |\bm h|\,\hat{\bm r}\cdot\bm\sigma}
        =\cosh(\beta|\bm h|)I+\sinh(\beta|\bm h|)\,\hat{\bm r}\cdot\bm\sigma.
    \end{align}
    Taking the trace and using $\tr[\sigma_x]=\tr[\sigma_y]=\tr[\sigma_z]=0$, we find
    \begin{align}
        \tr[e^{-\beta H(\bm h)}]
        =2\cosh(\beta|\bm h|).
    \end{align}
    Therefore, the Gibbs state at inverse temperature $\beta$ is given by
    \begin{align}
        \tau(\bm h) = \frac{e^{-\beta H(\bm h)}}{\tr[e^{-\beta H(\bm h)}]} = \frac{1}{2}\left(I+\tanh(\beta|\bm h|)\,\frac{\bm h}{|\bm h|}\cdot\bm\sigma \right),
    \end{align}
    with Bloch vector
    \begin{align}
        \bm{r}(\bm{h}) = \tanh(\beta|\bm{h}|) \frac{\bm{h}}{|\bm{h}|}.
    \end{align}
    
    Suppose that the experimenter intends to implement the nominal field $h_0 \hat{\bm z}$, corresponding to the nominal Gibbs state $\tau(h_0 \hat{\bm z})$ with Bloch vector $\bm{r}(h_0 \hat{\bm z}) = (0,0,\tanh(\beta h_0))$. In practice, however, the actual field may deviate slightly from this nominal value due to fluctuations, calibration errors, or incomplete knowledge of the control parameters. It is therefore natural to describe the possible fields by a small Euclidean ball around $h_0 \hat{\bm z}$,
    \begin{align}
        B_{\delta}(h_0 \hat{\bm z}) = \{\bm{h} \in \RR^3: |\bm{h} - h_0 \hat{\bm z}| \leq \delta\},
    \end{align}
    where $\delta > 0$ may be arbitrarily small. As illustrated in Fig.~\ref{fig:qubit field example}, the map $\bm{h} \mapsto \bm{r}(\bm{h})$ sends this ball of candidate field vectors to a neighborhood of Gibbs states in the Bloch ball around the nominal Gibbs state $\tau(h_0 \hat{\bm z})$. Since the uncertainty in $\bm{h}$ spans three independent directions, the resulting family of Gibbs states $\sE$ is locally three-dimensional, and its affine hull therefore covers the full state space $\aff(\sE) = \aff(\density)$. Consequently, even if the nonequilibrium state is known exactly, $\sP = \{\rho\}$, the geometric condition is satisfied automatically, and nontrivial athermality purification is impossible. Notably, this conclusion holds for arbitrarily small but nonzero uncertainty $\delta$.

    \begin{figure}[H]
        \centering
        \includegraphics[width=0.58\linewidth]{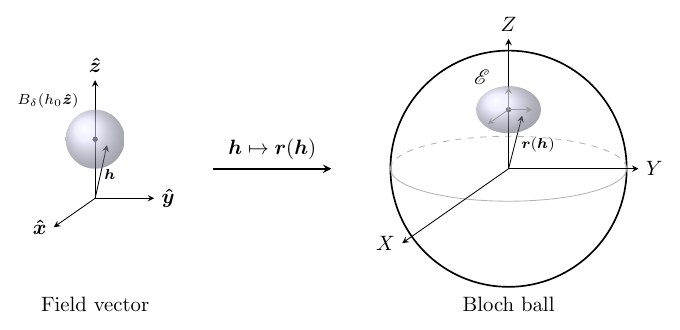}
        \caption{Geometric illustration of the Hamiltonian uncertainty in the qubit-field example. A small ball of possible field vectors around $h_0\bm{\hat{z}}$ is mapped to a neighborhood of Gibbs states in the Bloch ball around $\tau(h_0\bm{\hat{z}})$. The resulting set $\sE$ spans three independent directions, so its affine hull $\aff(\sE)$ covers the full state space.}
        \label{fig:qubit field example}
    \end{figure}
\end{example}

\subsection{No-go theorem for work extraction into a clean battery}

A particularly important instance of athermality ``purification'' is work extraction into a clean battery. The following corollary is obtained directly from Theorem~\ref{thm:no-go theorem for transformation with uncertainty} by setting the target athermal state to be the clean battery $(\ket{1}\bra{1}, \pi_M)$.

\begin{restate}{Corollary}{2}[No-go theorem for work extraction into a clean battery]
    \label{cor:no-go for work extraction to clean battery}
    Let $\ve \in [0,1)$. Suppose that $\sP, \sE \subseteq \density$ satisfies $\conv(\sP) \cap \aff(\sE) \neq \emptyset$. Then $(\sP, \sE) \xrightarrow[]{\cL,\; \ve} (\ket{1}\bra{1}, \pi_M)$ is achievable by some $\cL \in \GPL$ if and only if $\ve \geq 1 -1/M$. Moreover, whenever the transformation is achievable, the optimal work extraction is trivially realized by the thermal operation $\cL(\cdot) = \tr[\cdot] \pi_{1/(1-\ve)}$.
\end{restate} 

An illustration of Corollary~\ref{cor:no-go for work extraction to clean battery} is given in Fig.~\ref{fig:forbidden region precise battery}. The region $\ve < 1-\frac{1}{M}$ is completely forbidden, excluding the the operationally desirable regime of simultaneously small error and large extractable work. In particular, any nontrivial work extraction with $M>1$ necessarily incurs a nonzero error. Outside this region, the transformation can be achieved by the trivial thermal operation $\cL(\cdot) = \tr[\cdot]\pi_{M}$. Thus the boundary $\ve = 1-\frac{1}{M}$ is sharp: work extraction into a clean battery is either impossible or trivial, with no intermediate work-error tradeoff.

\begin{figure}[H]
    \centering
    \includegraphics[width=0.36\linewidth]{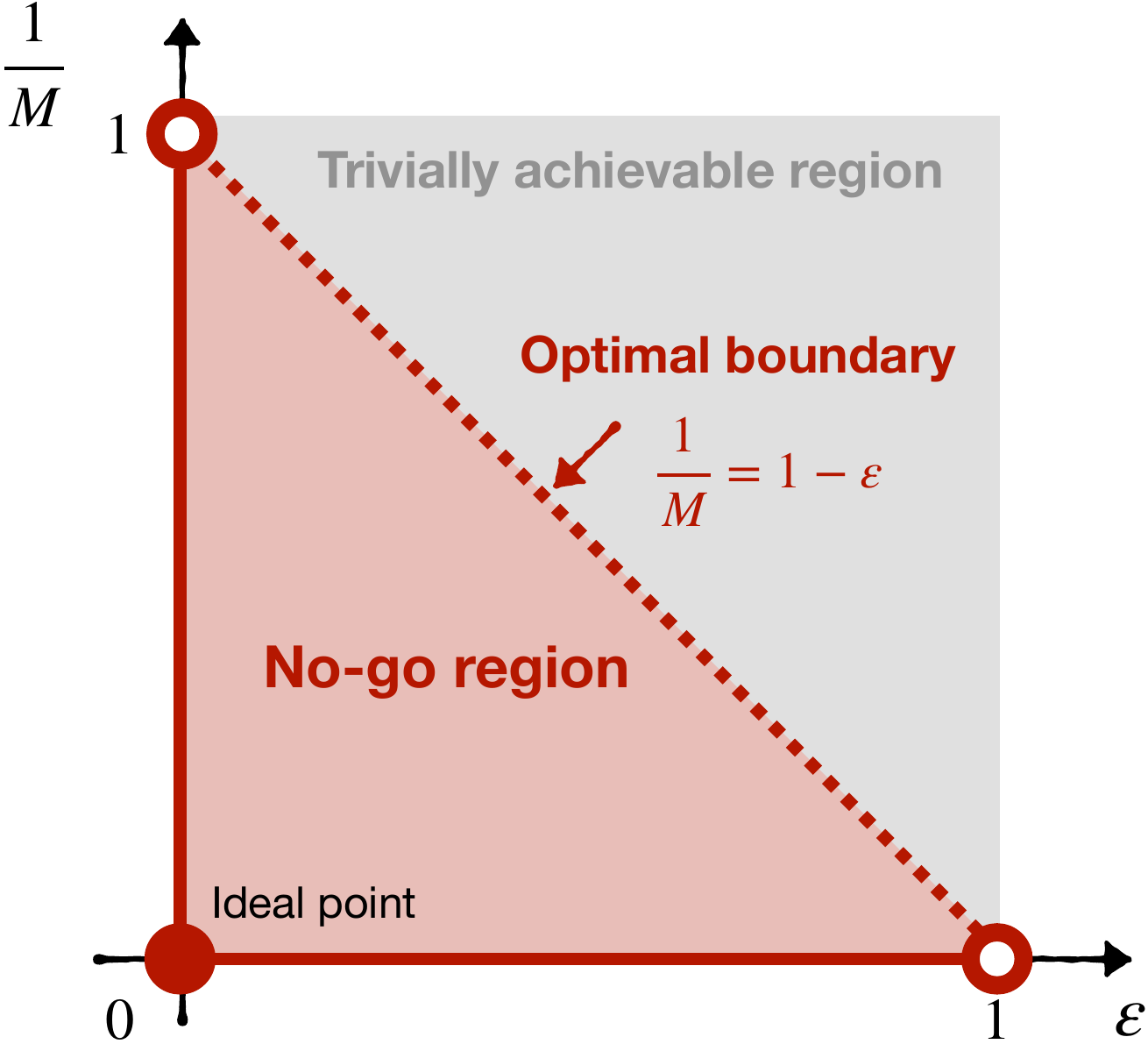}
    \caption{Illustration of the no-go theorem for work extraction into a clean battery. The region $1/M < 1-\ve$ is forbidden, while the complementary region $1/M \geq 1-\ve$ is trivially achievable.}
    \label{fig:forbidden region precise battery}
\end{figure}

When $\sE = \{\tau\}$ is a singleton, the geometric condition $\conv(\sP) \cap \aff(\sE) \neq \emptyset$ reduces to $\tau \in \conv(\sP)$. In this black-box setting, the optimal extractable work under GPO is given by~\cite[Theorem 2]{watanabe2024black}
\begin{align}
    D_{\min,\ve}(\sP\|\tau) = D_{\min,\ve}(\conv(\sP)\|\tau) = \inf_{\rho\in\conv(\sP)} D_{\min,\ve}(\rho\|\tau),
\end{align}
where the second equality follows from~\cite[Lemma 31]{fang2025generalized}. Since $\tau \in \conv(\sP)$, the faithfulness of the smoothed min-relative entropy in Lemma~\ref{lem:faithfulness of min-relative entropy} then implies $D_{\min,\ve}(\conv(\sP)\|\tau) = -\log(1-\ve)$. The present theorem extends this conclusion to general equilibrium uncertainty sets $\sE$ and, moreover, to the larger class $\GPL$. Since the optimal extractable work is already achieved by a thermal operation whenever $\conv(\sP) \cap \aff(\sE) \neq \emptyset$, all free operation classes TO, GPC, GPO, GPL yield the same clean-battery extractable work. In this sense, equilibrium uncertainty collapses the usual operational distinctions among free-operation classes and eliminates the meaningful work extraction regime familiar from the standard theory~\cite{horodecki2013fundamental,gour2022role,watanabe2024black}.

\subsection{No-go theorem for battery energy truncation}

The implications of the aforementioned no-go theorems become particularly transparent when applied to battery systems. In the standard setting, a battery with larger capacity is always at least as useful as one with smaller capacity since any surplus energy can be discarded by a free truncation operation. This intuition is formalized in Lemma~\ref{lem:yesgo truncation}. Under equilibrium uncertainty, however, this monotonicity breaks down completely. As shown in Corollary~\ref{col:no-go truncation}, a universal energy truncation from an uncertain battery to a precisely calibrated one becomes strictly impossible outside of the trivial regime, regardless of how small the uncertainty may be. The contrast between the two scenarios is illustrated in Fig.~\ref{fig:truncation}.

\begin{figure}[H]
    \centering
    \includegraphics[width=0.6\linewidth]{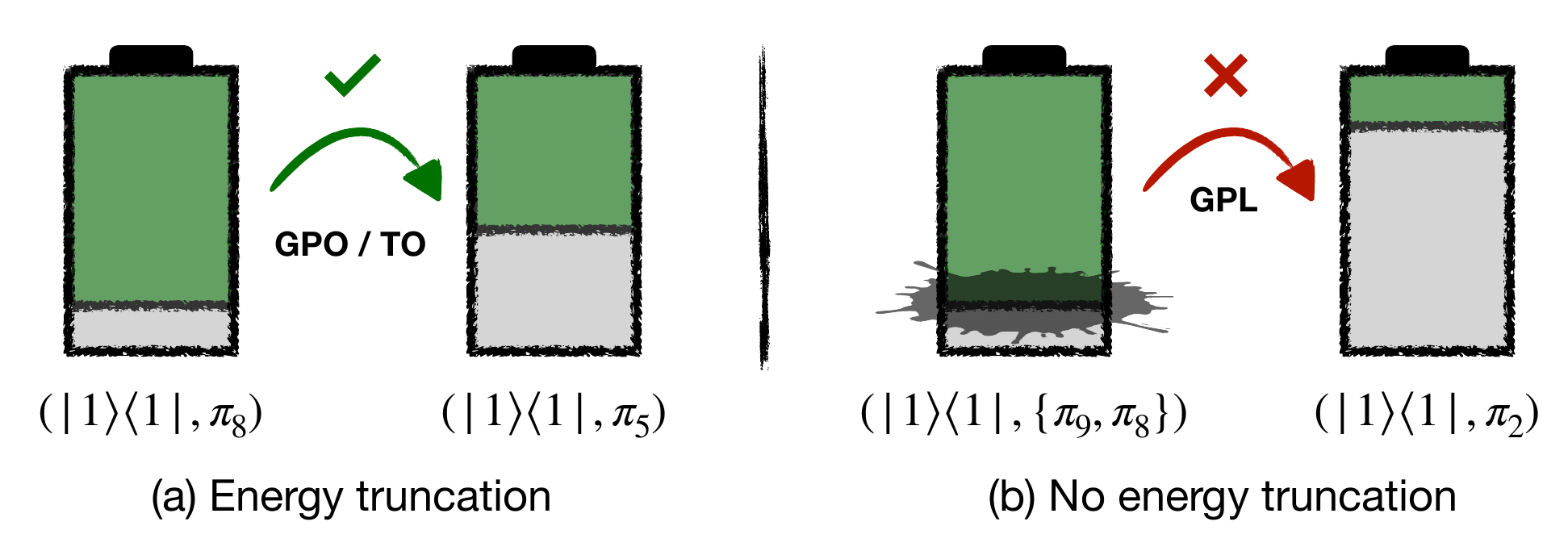}
    \caption{Illustration of the no-go theorem for battery energy truncation. (a) In the standard setting, a battery with higher capacity can be freely truncated to one with lower capacity. (b) Under equilibrium uncertainty, such truncation is impossible for any nontrivial target capacity.}
    \label{fig:truncation}
\end{figure}

\begin{lemma}[Energy truncation for clean batteries]
    \label{lem:yesgo truncation}
    Let $\ve \in [0,1)$ and $M, N > 1$. The energy truncation $(\ket{1}\bra{1}, \pi_M) \xrightarrow[]{\cF,\; \ve} (\ket{1}\bra{1}, \pi_N)$ is achievable by some $\cF\in \TO$ or $\GPO$ if and only if $M\geq N(1-\ve)$. 
\end{lemma}
\begin{proof}
    ($\Leftrightarrow$ for TO) Since the transformation is between two semi-classical states, the condition for convertibility under TO is known to be equivalent to (see~\cite{renes2016relative} or~\cite[Eq.~(6)]{lipka-bartosik2024quantum}),
    \begin{align}
        \beta_x(\ket{1}\bra{1}\|\pi_M) \leq \beta_{x-\ve}(\ket{1}\bra{1}\|\pi_N) \quad \text{for all } x \in (\ve,1),
    \end{align}
     where
    \begin{align}
        \beta_x(\rho\|\sigma) = \min \{\tr(\sigma Q): 0\leq Q \leq I,\; \tr(\rho Q) \geq 1-x\}.
    \end{align}
    A straightforward calculation yields $\beta_x(\ket{1}\bra{1}\|\pi_M) = (1-x)/M$. Hence, the condition becomes
    \begin{align}
        \frac{1-x}{M} \leq \frac{1-(x-\ve)}{N} \quad \text{for all } x \in (\ve,1).
    \end{align}
    It is easy to verify that this inequality holds for all such $x$ if and only if $M \geq N(1-\ve)$. \medskip

    ($\Rightarrow$ for GPO) Suppose that there exists $\cF\in\GPO$ such that $T(\cF(\ket{1}\bra{1}),\ket{1}\bra{1})\le \ve$ and $\cF(\pi_M)=\pi_N$. Since
    \begin{align}
        \pi_M-\frac{1}{M}\ket{1}\bra{1}=\frac{M-1}{M}\ket{0}\bra{0}\ge 0,
    \end{align}
    positivity and linearity of $\cF$ imply
    \begin{align}
        M\pi_N-\cF(\ket{1}\bra{1})\ge 0.
    \end{align}
    Let $\Delta$ be the completely dephasing channel in the computational basis, and write
    \begin{align}
        \Delta(\cF(\ket{1}\bra{1}))=(1-x)\ket{0}\bra{0}+x\ket{1}\bra{1}
    \end{align}
    for some $x\in[0,1]$. By contractivity of the trace distance under $\Delta$,
    \begin{align}
        1-x
        =
        T(\Delta(\cF(\ket{1}\bra{1})),\ket{1}\bra{1})
        \le
        T(\cF(\ket{1}\bra{1}),\ket{1}\bra{1})
        \le \ve,
    \end{align}
    so $x\ge 1-\ve$. On the other hand, dephasing the operator inequality above gives
    \begin{align}
        M\pi_N-\bigl[(1-x)\ket{0}\bra{0}+x\ket{1}\bra{1}\bigr]\ge 0,
    \end{align}
    hence in particular $x\le M/N$. So $1-\ve\le x\le M/N$, which implies $M\ge N(1-\ve)$. \medskip
    
    ($\Leftarrow$ for GPO) Assume $M\ge N(1-\ve)$. If $\ve\ge 1-\frac{1}{N}$, then the replacer channel $\cF(\cdot)=\tr[\cdot]\pi_N$ already satisfies the requirements. It remains to consider the case $\ve<1-\frac{1}{N}$. Define
    \begin{align}
        q:=\frac{M-N(1-\ve)}{N(M-1)}, \quad
        \gamma_0:=(1-q)\ket{0}\bra{0}+q\ket{1}\bra{1}, \quad
        \gamma_1:=\ve\ket{0}\bra{0}+(1-\ve)\ket{1}\bra{1}.
    \end{align}
    One checks that $q\in[0,1)$: the lower bound follows from $M\ge N(1-\ve)$, and the upper bound $q<1$ is equivalent to $N\ve<(N-1)M$, which holds since $\ve<1-\tfrac{1}{N}$ gives $N\ve<N-1\le(N-1)M$.
    Now consider the measure-and-prepare channel $\cF(\cdot) := \tr[\ket{0}\bra{0}(\cdot)]\gamma_0 + \tr[\ket{1}\bra{1}(\cdot)]\gamma_1$. Since both $\gamma_0$ and $\gamma_1$ are valid quantum states, $\cF$ is CPTP. It is straightforward to verify that
    \begin{align}
        T(\cF(\ket{1}\bra{1}),\ket{1}\bra{1})=T(\gamma_1,\ket{1}\bra{1})=\ve,
    \end{align}
    and
    \begin{align}
        \cF(\pi_M)=\frac{M-1}{M}\gamma_0+\frac{1}{M}\gamma_1.
    \end{align}
    By the choice of $q$, we have $\cF(\pi_M)=\pi_N$. This completes the proof.
\end{proof}

We now consider the same task when the input battery capacity is not exactly known. Suppose the experimenter is given an excited battery whose nonequilibrium state is known to be $\sP = \{\ket{1}\bra{1}\}$, but whose exact capacity is uncertain, for instance due to slow drift in the parameter of the battery Hamiltonian. Assume the actual capacity is either $M_1$ or $M_2$, with $M_2 > M_1$, so that the candidate Gibbs state is either $\pi_{M_1}$ or $\pi_{M_2}$. This gives the dirty battery $(\ket{1}\bra{1}, \{\pi_{M_1}, \pi_{M_2}\})$. If both possible capacities satisfy $M_1, M_2 \geq N$, then each candidate battery is individually sufficient, in the exact setting ($\ve = 0$), to produce an excited clean battery of capacity $N$ by Lemma~\ref{lem:yesgo truncation}. One might therefore expect that the excess energy in the dirty battery can be discarded uniformly, yielding the lower-capacity clean battery $(\ket{1}\bra{1}, \pi_N)$. The following corollary shows that this intuition fails: equilibrium uncertainty rules out any nontrivial universal energy truncation.

\begin{restate}{Corollary}{3}[No universal energy truncation for dirty batteries]
    \label{col:no-go truncation}
    Let $\ve\in[0,1)$ and $M_2 > M_1 \geq N > 1$. Then the universal energy truncation $(\ket{1}\bra{1}, \{\pi_{M_1}, \pi_{M_2}\}) \xrightarrow[]{\cL,\;\ve} (\ket{1}\bra{1}, \pi_N)$ is achievable by some $\cL \in \GPL$ if and only if $\ve \geq 1-1/N$.
\end{restate}

This result follows directly from Corollary~\ref{cor:no-go for work extraction to clean battery} since $\ket{1}\bra{1} \in \aff(\{\pi_{M_1}, \pi_{M_2}\})$. Its operational implication, however, is striking: even if the nonequilibrium state is known exactly, an \emph{arbitrarily small} uncertainty in the equilibrium reference, namely any $M_2 - M_1 > 0$, is sufficient to eliminate any nontrivial energy truncation. No matter how large the possible input capacities $M_1$ and $M_2$ are, or how modest the target capacity $N$ is, the conversion is possible only in the trivial regime where the target battery's excited state is already $\ve$-close to its own Gibbs state. This behavior sharply contrasts with the black-box setting of Ref.~\cite{watanabe2024black}, where uncertainty is confined to the nonequilibrium state and nontrivial energy truncation can still be possible. The comparison reveals a fundamental asymmetry between the two sources of uncertainty in quantum thermodynamics: equilibrium uncertainty can be far more detrimental than nonequilibrium-state uncertainty. It also contrasts with entanglement transformation, where universal distillation from an unknown resource can be achieved at the optimal rate~\cite{matsumoto2007universal}, and universal entanglement truncation is therefore possible. This highlights a structural difference between the resource theories of entanglement and athermality.

The no-go theorems in this section motivate a more flexible treatment of battery models when exploring the possibilities and limitations of work transformation under uncertainty. As summarized in Table~\ref{tab:one-shot summary}, we study two settings: a clean battery, whose equilibrium state is precisely known, and a dirty battery, whose equilibrium state is itself uncertain. For both scenarios, we derive exact one-shot entropic characterizations of work extraction and work of formation. These characterizations involve both the standard smoothed min- and max-relative entropies, shown in cyan, and the new subspace-constrained variants introduced in this work, shown in orange.

\vspace{-0.8em}

\begin{table}[H]
    \centering         
    \renewcommand{\arraystretch}{1.5} 
    
    \smallskip
    \begin{tabular}{cccc}
        \hline\hline
        & \cuscol{4cm}{Extractable work} & \cuscol{4cm}{Work cost} &  \cuscol{3.8cm}{Reversibility} \\
        \hline
        \cuscol{3.8cm}{Clean battery} & \cellcolor{myorange!18} $D_{\min,\ve}^{\sE}(\sP\|\sE)$ & \cellcolor{mycyan!18} $D_{\max,\ve}(\sP\|\sE)$ & No \\
        \hline
        \cuscol{3.8cm}{Dirty battery} & \cellcolor{mycyan!18} $D_{\min,\ve}(\sP\|\sE)$ & \cellcolor{myorange!18} $D_{\max,\ve}^{\,\sE}(\sP\|\sE)$ & No \\
        \hline\hline
    \end{tabular}
    \caption{One-shot entropic characterizations of work transformations under uncertain athermality $(\sP, \sE)$ where both $\sP$ and $\sE$ are sets of quantum states. Here $D_{\min,\varepsilon}$ and $D_{\max,\varepsilon}$ denote the standard smoothed min- and max-relative entropies, while $D_{\min,\varepsilon}^{\sE}$ and $D_{\max,\varepsilon}^{\sE}$ denote their subspace-constrained variants introduced in this work. See Sections~\ref{sec:clean-battery} and~\ref{sec:dirty-battery} for definitions.}
    \label{tab:one-shot summary}
\end{table}

A notable feature is that the conventional correspondence between operational tasks and entropic quantities~\cite{liu2019oneshot} breaks down: the standard relative entropies no longer govern both tasks within a fixed battery model. For a clean battery, formation is characterized by the standard max-relative entropy, whereas extraction involves the subspace-constrained min-relative entropy. For a dirty battery, the pattern is reversed: extraction is governed by the standard min-relative entropy, while formation involves the subspace-constrained max-relative entropy. We discuss the clean-battery and dirty-battery scenarios in Sections~\ref{sec:clean-battery} and~\ref{sec:dirty-battery}, respectively. Finally, in Section~\ref{sec: asymptotic} we provide asymptotic analyses and give an explicit example demonstrating the irreversibility of work transformations under both models.

\section{Clean-battery work transformations}
\label{sec:clean-battery}

This section concerns work transformations through a clean battery and provides detailed proofs of the corresponding results stated in the main text.

\subsection{Work extraction into a clean battery}

We start with work extraction from an uncertain athermal state into a clean battery.

\begin{definition}[One-shot extractable work into a clean battery]
    \label{def:one-shot clean battery extractable work}
    Let $\ve \in [0,1)$, and let $(\sP, \sE)$ be an uncertain athermal state with $\sP,\sE\subseteq\density$. Given a class of free operations $\fF$, the one-shot extractable work from $(\sP,\sE)$ into a clean battery is defined as
    \begin{align}
        \beta W_{\fF,\ve}(\sP,\sE) := \log\sup_{\cF\in\fF}\left\{M:(\sP,\sE)\xrightarrow[]{\cF,\;\ve} (\ket{1}\bra{1}, \pi_M)\right\}.
    \end{align}
\end{definition}

This definition reduces to the standard one-shot extractable work~\cite{horodecki2013fundamental,brandao2015second,gour2022role} when both $\sP$ and $\sE$ are singletons, and to black-box work extraction~\cite{watanabe2024black} when only $\sE$ is a singleton. 

To characterize the extractable work in this setting, we introduce the following quantity, which can be regarded as a constrained version of the standard smoothed min-relative entropy.

\begin{definition}[Subspace-constrained min-relative entropy]
\label{def:constrained min relative entropy}
    Let $\ve \in [0,1)$, $\sP, \sE \subseteq \density$, and let $\sK \subseteq \sL$. The subspace-constrained min-relative entropy between $\sP$ and $\sE$ with respect to $\sK$ is defined as
    \begin{align}
        D_{\min,\ve}^{\sK}(\sP\|\sE) := -\log \min_{\substack{0 \leq E \leq I \\ E \perp V(\sK)}} \left\{\sup_{\tau \in \sE}\tr[\tau E] : \sup_{\rho \in \sP} \tr[\rho(I-E)] \leq \ve\right\}.
    \end{align}
    Here $V(\sK) := \operatorname{span}\{\tau - \tau^\prime : \tau, \tau^\prime \in \sK\}$ is the subspace spanned by the differences of elements in $\sK$, and the condition $E \perp V(\sK)$ means that $\tr[E X] = 0$ for all $X \in V(\sK)$.
\end{definition}

The subspace $V(\sK)$ can equivalently be written as
\begin{align}
    V(\sK) = \aff(\sK) - \tau_0 := \{X-\tau_0 : X \in \aff(\sK)\}
\end{align}
for any fixed $\tau_0 \in \sK$. In particular, if $\sK$ is a singleton, then $V(\sK) = \{0\}$, and the constraint $E \perp V(\sK)$ becomes vacuous. In this case, this quantity reduces to the standard smoothed min-relative entropy between two sets of quantum states in Definition~\ref{def:smoothed min relative entropy}. Another case of particular interest is $\sK = \sE$, for which the quantity acquires a natural interpretation as a \emph{constrained} hypothesis testing problem. Recall that in standard hypothesis testing between two sets $\sP$ and $\sE$, one seeks a test operator $0 \leq E \leq I$ that minimizes the worst-case type-II error $\sup_{\tau \in \sE}\tr[\tau E]$ subject to a constant threshold on the worst-case type-I error $\sup_{\rho \in \sP} \tr[\rho(I - E)] \leq \ve$. Since $V(\sE) = \aff(\sE) - \tau_0$ for any fixed $\tau_0 \in \sE$, the additional condition $E \perp V(\sE)$ here is equivalent to requiring $\tr[\tau E] = \tr[\tau_0 E]$ for all $\tau \in \sE$, and hence the type-II error is identical for every $\tau \in \sE$. In other words, the test must remain oblivious to which particular candidate $\tau \in \sE$ is realized.

The following theorem provides an exact characterization of the one-shot extractable work to a clean battery under GPO in terms of the subspace-constrained min-relative entropy.

\begin{restate}{Theorem}{4}[One-shot extractable work to a clean battery under GPO]
    \label{thm:one-shot clean battery work extraction}
    Let $\ve \in [0,1)$, and let $(\sP, \sE)$ be an uncertain athermal state with $\sP,\sE\subseteq\density$. The one-shot extractable work from $(\sP,\sE)$ into a clean battery under GPO is given by
    \begin{align}
        \beta W_{\GPO, \ve}(\sP, \sE) = D_{\min,\ve}^{\sE}(\sP\|\sE).
    \end{align}
\end{restate}
\begin{proof}
We show that the conversion $(\sP, \sE) \xrightarrow[]{\cF,\;\ve} (\ket{1}\bra{1}, \pi_M)$ is achievable by some $\cF \in \GPO$ if and only if there exists a test operator $0 \leq E \leq I$ satisfying $E \perp V(\sE)$, $\tr[\tau E] = 1/M$ for all $\tau \in \sE$, and $\sup_{\rho \in \sP}\tr[\rho(I - E)] \leq \ve$. The desired entropic formula then follows by optimizing over $M$. \smallskip

($\Rightarrow$) Suppose that $\cF \in \GPO$ achieves $(\sP, \sE) \xrightarrow[]{\cF,\;\ve} (\ket{1}\bra{1}, \pi_M)$. Then for every $(\rho,\tau)\in\sP\times\sE$, we have
\begin{align}
    T(\cF(\rho), \ket{1}\bra{1}) \leq \ve,\qquad\cF(\tau) = \pi_M.
\end{align}
Let $\Delta$ denote the completely dephasing channel in the computational basis $\{\ket{0}, \ket{1}\}$ and define $\cF^\prime := \Delta \circ \cF$. Since both $\cF$ and $\Delta$ are CPTP, the map $\cF^\prime$ is also CPTP. Moreover, for every $\tau \in \sE$,
\begin{align}
    \cF^\prime(\tau) = \Delta(\cF(\tau)) = \Delta(\pi_M) = \pi_M.
\end{align}
For every $\rho \in \sP$, contractivity of the trace distance under $\Delta$ gives
\begin{align}
    T(\cF^\prime(\rho), \ket{1}\bra{1}) = T(\Delta(\cF(\rho)), \Delta(\ket{1}\bra{1})) \leq T(\cF(\rho), \ket{1}\bra{1}) \leq \ve.
\end{align}
Now define $E := (\cF^\prime)^\dagger(\ket{1}\bra{1})$, where $(\cF^\prime)^\dagger$ denotes the adjoint map. Since $\cF^\prime$ is CPTP, its adjoint is completely positive and unital, and therefore $0 \leq E \leq I$. It remains to verify that $E$ satisfies the required conditions. \smallskip

First, for every $\tau\in\sE$,
\begin{align}
    \tr[E\tau] = \tr[(\cF^\prime)^\dagger(\ket{1}\bra{1})\tau] = \tr[\ket{1}\bra{1}\cF^\prime(\tau)] = \bra{1}\pi_M\ket{1} = \frac{1}{M}.
\end{align}
Hence $\tr[E\tau]$ is independent of $\tau \in \sE$. It follows that $\tr[(\tau - \tau^\prime)E] = 0$ for all $\tau, \tau^\prime \in \sE$, and therefore $E \perp V(\sE)$. \smallskip

Second, observe that $\cF^\prime(\rho)$ is diagonal in the computational basis, so a direct calculation gives
\begin{align}
    T(\cF^\prime(\rho), \ket{1}\bra{1}) & = \frac{1}{2} \bigl\| \tr[\cF^\prime(\rho)\ket{1}\bra{1}] \ket{1}\bra{1} + \tr[\cF^\prime(\rho)\ket{0}\bra{0}] \ket{0}\bra{0} - \ket{1}\bra{1} \bigr\|_1 \\
    & = \frac{1}{2}\bigl\|\tr[\rho E]\ket{1}\bra{1} + \tr[\rho(I-E)]\ket{0}\bra{0} -\ket{1}\bra{1}\bigr\|_1 \\
    & = \frac{1}{2}\bigl\|({\tr[\rho E] - 1})\ket{1}\bra{1} + \tr[\rho(I-E)]\ket{0}\bra{0}\bigr\|_1 \\
    & = \tr[\rho(I-E)].
\end{align}
Since $T(\cF^\prime(\rho), \ket{1}\bra{1}) \leq \ve$ for all $\rho \in \sP$, we have $\sup_{\rho \in \sP}\tr[\rho(I-E)] \leq \ve$. \medskip

($\Leftarrow$) Suppose there exists $0 \leq E \leq I$ such that $E \perp V(\sE)$, $\tr[\tau_0 E] = 1/M$ for some (hence all) $\tau_0 \in \sE$, and $\sup_{\rho \in \sP}\tr[\rho(I-E)] \leq \ve$. Define the measure-and-prepare channel $\cE(\cdot) := \tr[(I - E)(\cdot)]\ket{0}\bra{0} + \tr[E(\cdot)]\ket{1}\bra{1}$. This is CPTP by construction. We now verify the desired properties. For every $\tau \in \sE$, 
\begin{align} 
    \cE(\tau) = \tr[\tau E]\ket{1}\bra{1} + (1 - \tr[\tau E])\ket{0}\bra{0} = \frac{1}{M}\ket{1}\bra{1} + \left(1 - \frac{1}{M}\right)\ket{0}\bra{0} = \pi_M.
\end{align}
Moreover, for every $\rho \in \sP$, the same trace-distance computation as above yields $T(\cE(\rho), \ket{1}\bra{1}) = \tr[\rho(I-E)] \leq \ve$. Thus $\cE \in \GPO$ achieves $(\sP, \sE) \xrightarrow[]{\cE,\;\ve} (\ket{1}\bra{1}, \pi_M)$. This proves the claimed equivalence, and the theorem follows by optimizing over $M$.
\end{proof}

\begin{remark}
    The subspace constraint $E\perp V(\sE)$ appearing in $D_{\min,\ve}^{\sE}(\sP\|\sE)$ can be equivalently encoded by replacing the second argument $\sE$ with its affine hull. More precisely, we have
    \begin{align} 
        \beta W_{\GPO, \ve}(\sP, \sE) = D_{\min,\ve}^{\sE}(\sP\|\sE) = D_{\min,\ve}(\sP\|\aff(\sE)),
    \end{align}
    where the last quantity is interpreted using the same formula as in Definition~\ref{def:smoothed min relative entropy}, with the second argument extended from sets of states to affine sets of Hermitian trace-one operators. To see the equivalence, fix a test operator $E$. If there exist $\tau_1,\tau_2\in\sE$ such that $\tr(E\tau_1) \neq \tr(E\tau_2)$, then the affine line $\tau_{\lambda} = \lambda \tau_1 + (1-\lambda) \tau_2$ for $\lambda \in \mathbb{R}$ is contained in $\aff(\sE)$, and the objective function $\tr[E\tau_\lambda]$ is unbounded above as $\lambda$ varies. Hence any test operator $E$ with finite objective value must satisfy $\tr(E\tau_1) = \tr(E\tau_2)$ for all $\tau_1, \tau_2 \in \sE$, which is precisely the condition $E\perp V(\sE)$. Conversely, if $E\perp V(\sE)$, then $\tr[E\tau]$ is constant over $\aff(\sE)$ and therefore coincides with its value on $\sE$. Similar affine-hull formulation also appears in discussions of channel resource distillation~\cite{regula2020benchmarking,regula2021oneshot,watanabe2024black}. \smallskip

    This reformulation endows the entropic quantity $D_{\min,\ve}(\sP\|\aff(\sE))$ with a concrete operational meaning: it equals the one-shot extractable work to a clean battery under GPO in the presence of equilibrium uncertainty. Notably, the second argument $\aff(\sE)$ generally contains operators that are not positive semidefinite, so this divergence falls outside the standard domain of quantum relative entropies. While divergences with nonpositive semidefinite second arguments have appeared in specific contexts such as entanglement distillation~\cite{fang2019nonasymptotic} and coherence distillation~\cite{regula2018oneshot}, the present result provides a fully general operational interpretation for arbitrary sets $\sP$ and $\sE$. \smallskip
    
    This affine-hull reformulation also explains why the geometric condition $\conv(\sP) \cap \aff(\sE) \neq \emptyset$ arises naturally in Theorems~\ref{thm:no-go theorem for transformation with uncertainty} and Corollary~\ref{cor:no-go for work extraction to clean battery}. Indeed, we have
    \begin{align} 
        D_{\min,\ve}(\sP\|\aff(\sE)) = D_{\min,\ve}(\conv(\sP)\|\aff(\sE)) = \inf_{\substack{\rho \in \conv(\sP)\\ X \in \aff(\sE)}} D_{\min,\ve}(\rho\|X),
    \end{align}
    where the second equality follows by the same argument as in Ref.~\cite[Lemma 31]{fang2025generalized} and $D_{\min,\ve}(\rho\|X)$ is understood in the extended sense. Combining this identity with the faithfulness of the smoothed min-relative entropy in Lemma~\ref{lem:faithfulness of min-relative entropy} shows that the clean-battery extractable work reduces to the trivial value $\beta W_{\GPO, \ve}(\sP, \sE) = -\log(1-\ve)$ whenever $\conv(\sP) \cap \aff(\sE) \neq \emptyset$. However, it is worth mentioning that the no-go theorems are stronger than this faithfulness argument, since they are established for the larger class $\GPL$ rather than $\GPO$.
\end{remark}

Theorem~\ref{thm:one-shot clean battery work extraction} extends the known connections between work extraction and hypothesis testing~\cite{liu2019oneshot,gour2022role,watanabe2024black} to the general framework developed here. When $\sE$ is a singleton, $V(\sE) = \{0\}$, so the orthogonality constraint becomes vacuous, and the theorem reduces to the black-box work extraction~\cite[Theorem 2]{watanabe2024black}. If $\sP$ is also a singleton, it further reduces to the standard work extraction formula~\cite[Eq.~(50)]{gour2022role}. In the presence of equilibrium uncertainty, however, the same operational task is governed by a constrained hypothesis testing problem, where the test operator must satisfy the additional condition $E \perp V(\sE)$. This condition encodes the penalty imposed by equilibrium uncertainty: it forces the test to be calibration-free, and hence unable to exploit knowledge of which equilibrium state is realized. As shown later in Theorem~\ref{thm:one-shot dirty battery work extraction}, removing this constraint recovers the standard min-relative entropy, which corresponds to work extraction into a dirty battery. The subspace constraint therefore exactly captures the distinction between these two models.

\subsection{Work of formation from a clean battery}

We now turn to the reverse task of preparing an uncertain athermal state from a clean battery.

\begin{definition}[One-shot work cost from a clean battery]
    \label{def:one-shot clean battery work cost}
    Let $\ve \in [0,1)$, and let $(\sP, \sE)$ be an uncertain athermal state with $\sP,\sE\subseteq\density$. Given a class of free operations $\fF$, the one-shot work cost of preparing $(\sP, \sE)$ from a clean battery is defined as
    \begin{align}\label{eq:one-shot work cost}
        \beta C_{\fF,\ve}(\sP, \sE) & := \log \inf_{\cF \in \fF} \left\{M: (\ket{1}\bra{1}, \pi_{M}) \xrightarrow[]{\cF,\; \ve} (\sP, \sE)\right\}. 
    \end{align}
\end{definition}

This definition recovers the standard one-shot work cost~\cite{horodecki2013fundamental,gour2022role} when both $\sP$ and $\sE$ are singletons. The following theorem gives an exact characterization of the clean-battery work cost under GPO in terms of the smoothed max-relative entropy between two sets, thereby giving a clear operational interpretation to this quantity.

\begin{restate}{Theorem}{5}[One-shot work cost from a clean battery under GPO]
    \label{thm:one-shot clean battery work cost}
    Let $\ve \in [0,1)$, and let $(\sP, \sE)$ be an uncertain athermal state with $\sP,\sE\subseteq\density$. The one-shot work cost of preparing $(\sP, \sE)$ from a clean battery under $\GPO$ is given by
    \begin{align}
        \beta C_{\GPO, \ve}(\sP,\sE) = D_{\max, \ve}(\sP\|\sE).
    \end{align}
\end{restate}
\begin{proof}
    By definition of uncertain athermality resource transformation in Eq.~\eqref{eq:sets conversion} and the definition of $\beta C_{\GPO, \ve}(\sP,\sE)$ in Eq.~\eqref{eq:one-shot work cost}, we have
    \begin{align}
        \beta C_{\GPO,\ve}(\sP, \sE) & = \log \inf_{\cF \in \GPO} \left\{M: (\ket{1}\bra{1}, \pi_{M}) \xrightarrow[]{\cF,\; \ve} (\sP, \sE)\right\} \\
        & = \log \inf_{(\rho, \tau) \in \sP\times\sE} \inf_{\cF\in\GPO} \left\{M:(\ket{1}\bra{1}, \pi_M) \xrightarrow[]{\cF,\;\ve} (\rho, \tau)\right\} \\
        & = \inf_{(\rho, \tau) \in \sP\times\sE} \log \inf_{\cF\in\GPO} \left\{M:(\ket{1}\bra{1}, \pi_M) \xrightarrow[]{\cF,\;\ve} (\rho, \tau)\right\} \\
        & = \inf_{(\rho, \tau) \in \sP\times\sE} D_{\max,\ve}(\rho\|\tau) \\
        & = D_{\max,\ve}(\sP\|\sE),
    \end{align}
    where the penultimate equality follows from~\cite[Eq.~(48)]{wang2019resource}. This concludes the proof.
\end{proof}

When both $\sP$ and $\sE$ are singletons, this theorem recovers the standard one-shot work-cost formula~\cite[Eq.~(50)]{gour2022role}. More generally, it shows that the work cost of preparing an uncertain athermal state is determined by the easiest candidate within the target set. It also gives the smoothed max-relative entropy between two sets of quantum states, previously studied as a mathematical quantity in Ref.~\cite{fang2025generalized}, a direct operational interpretation in quantum thermodynamics under uncertainty.

\section{Dirty-battery work transformations}
\label{sec:dirty-battery}

In this section, we consider work transformations through a dirty battery and provide detailed proofs of the corresponding results stated in the main text.

\subsection{Work extraction into a dirty battery}

The preceding no-go results and clean-battery analysis reveal a sharp obstruction: under equilibrium uncertainty, the subspace constraint on the test can severely limit work extraction and, in generic cases (Corollary~\ref{cor:no-go for work extraction to clean battery}), collapse the usual work-error tradeoff. Since uncertainty in the equilibrium state of the input system is already unavoidable, it is reasonable to allow uncertainty in the target system as well and ask whether this obstruction can be circumvented by using a dirty battery.

\begin{definition}[One-shot extractable work into a dirty battery]
    \label{def:one-shot dirty battery extractable work}
    Let $\ve \in [0,1)$, and let $(\sP, \sE)$ be an uncertain athermal state with $\sP,\sE\subseteq\density$. Given a class of free operations $\fF$, the one-shot extractable work from $(\sP,\sE)$ into a dirty battery is defined as
    \begin{align}
        \label{eq:one-shot dirty battery extractable work}
        \beta \widebar{W}_{\fF,\ve}(\sP, \sE) := \log \sup_{\cF \in \fF} \left\{M: (\sP, \sE) \xrightarrow[]{\cF,\; \ve} (\ket{1}\bra{1}, \Pi_M)\right\}.
    \end{align}
\end{definition}
As in the clean-battery setting, this definition reduces to the standard one-shot extractable work~\cite{horodecki2013fundamental,brandao2015second,gour2022role} when both $\sP$ and $\sE$ are singletons, and to black-box work extraction~\cite{watanabe2024black} when only $\sE$ is a singleton. Therefore, it remains a reasonable extension of work extraction to the setting with equilibrium uncertainty.
The following theorem gives an exact characterization of the one-shot extractable work from a dirty battery under GPO. In contrast to the clean-battery case, the relevant quantity is the standard smoothed min-relative entropy between two sets of quantum states.

\begin{restate}{Theorem}{6}[One-shot extractable work to a dirty battery under GPO]
    \label{thm:one-shot dirty battery work extraction}
    Let $\ve \in [0,1)$, and let $(\sP, \sE)$ be an uncertain athermal state with $\sP,\sE\subseteq\density$. The one-shot extractable work from $(\sP,\sE)$ into a dirty battery under $\GPO$ is given by
    \begin{align}
        \beta \widebar{W}_{\GPO,\ve}(\sP, \sE) = D_{\min, \ve}(\sP\|\sE).
    \end{align}
\end{restate}
\begin{proof}
    ($\geq$) Let $0\leq E \leq I$ be any test operator satisfying $\sup_{\rho\in\sP}\tr[(I-E)\rho] \leq \varepsilon$.\smallskip

    In the special case where $\sup_{\tau \in \sE}\tr[E\tau] = 0$, fix an arbitrary $M_0 > 1$ and define the measure-and-prepare channel $\cE_{M_0}(\cdot) := \tr[(I-E)(\cdot)]\, \pi_{M_0} + \tr[E(\cdot)]\, \ket{1}\bra{1}$. Then for any $\tau\in\sE$, $\cE_{M_0}(\tau) = \pi_{M_0}$. For any $\rho \in \sP$, let $x = \tr[E\rho]$. Then
    \begin{align}
        \cE_{M_0}(\rho) & = (1-x)(1-1/M_0)\ket{0}\bra{0} + \bigl[(1-x)/M_0 + x\bigr]\ket{1}\bra{1} \\
        & = (1-x)(1-1/M_0)\ket{0}\bra{0} + \bigl[1 - (1-x)(1-1/M_0)\bigr]\ket{1}\bra{1}.
    \end{align}
    Therefore, $\cE_{M_0}(\rho) - \ket{1}\bra{1} = (1-x)(1-1/M_0)\big(\ket{0}\bra{0} - \ket{1}\bra{1}\big)$, and hence
    \begin{align}
        T(\cE_{M_0}(\rho), \ket{1}\bra{1}) = \tfrac{1}{2}\|\cE_{M_0}(\rho) - \ket{1}\bra{1}\|_1 = (1-x)(1-1/M_0) \leq 1-x \leq \ve,
    \end{align}
    where the last inequality follows from the feasibility of $E$. This shows that the conversion $(\sP, \sE) \xrightarrow[]{\cE_{M_0},\;\ve} (\ket{1}\bra{1}, \Pi_{M_0})$ is achievable by $\cE_{M_0} \in \GPO$ for any $M_0 > 1$. Taking $M_0\to\infty$ gives
    \begin{align}
        \beta \widebar{W}_{\GPO,\ve}(\sP, \sE) = \infty = D_{\min, \ve}(\sP\|\sE).
    \end{align}
    
    If $\sup_{\tau \in \sE}\tr[E\tau] > 0$, we define the measure-and-prepare channel $\cE(\cdot) = \tr[(I-E)(\cdot)] \ket{0}\bra{0} + \tr[E(\cdot)] \ket{1}\bra{1}$. For any $\rho \in \sP$,
    \begin{align}
        T(\cE(\rho), \ket{1}\bra{1}) = \frac{1}{2}\|\tr[(I-E)\rho] \ket{0}\bra{0} + \tr[E\rho] \ket{1}\bra{1} - \ket{1}\bra{1}\|_1 = \tr[(I-E)\rho] \leq \ve.
    \end{align}
    For any $\tau \in \sE$,
    \begin{align}
        \cE(\tau) = \tr[(I-E)\tau] \ket{0}\bra{0} + \tr[E\tau] \ket{1}\bra{1} = \pi_{M_{\tau}}
    \end{align}
    where $M_\tau = \frac{1}{\tr[E\tau]}$. Taking infimum over all $\tau \in \sE$ yields
    \begin{align}
       M^\prime:= \inf_{\tau\in\sE} M_{\tau} = \frac{1}{\sup_{\tau\in\sE} \tr[E\tau]}.
    \end{align}
    Thus the channel $\cE$ is feasible for the conversion $(\sP, \sE) \xrightarrow[]{\cE,\; \ve} (\ket{1}\bra{1}, \Pi_{M^\prime})$. It follows that
    \begin{align}
        \beta \widebar{W}_{\GPO,\ve}(\sP, \sE) \geq \log M^\prime = \log \frac{1}{\sup_{\tau\in\sE} \tr[E\tau]}
    \end{align}
    Taking the supremum over all feasible $E$ gives $\beta \widebar{W}_{\GPO,\ve}(\sP, \sE) \geq D_{\min,\ve}(\sP\|\sE)$. \smallskip

    ($\leq$) To obtain the reverse inequality, suppose $\cF\in\GPO$ achieves the conversion $(\sP, \sE) \xrightarrow[]{\cF,\; \ve} (\ket{1}\bra{1}, \Pi_M)$. Then for every $\rho \in \sP$, $T(\cF(\rho),\ket{1}\bra{1}) \leq \ve$, and for every $\tau \in \sE$, $\cF(\tau) = \pi_{M_{\tau}}$ for some $M_{\tau} \geq M$. Define $E := \cF^\dagger(\ket{1}\bra{1})$, where $\cF^\dagger$ denotes the adjoint of $\cF$. Since $\cF$ is CPTP, we have $0 \leq E \leq I$. \smallskip
    
    For any $\rho \in \sP$, the adjoint property gives $\tr[(I-E)\rho] = \tr[(I - \ket{1}\bra{1})\cF(\rho)]$. By the variational formula of trace distance,
    \begin{align}
        T(\cF(\rho), \ket{1}\bra{1}) = \sup_{0\leq F\leq I}\tr[F(\cF(\rho)-\ket{1}\bra{1})].
    \end{align}
    Taking $F = I - \ket{1}\bra{1}$ yields
    \begin{align}
        \tr[(I-E)\rho]  = \tr[(I - \ket{1}\bra{1})(\cF(\rho) - \ket{1}\bra{1})] \leq T(\cF(\rho), \ket{1}\bra{1}) \leq \ve.
    \end{align}
    Hence $\sup_{\rho \in \sP} \tr[(I-E)\rho] \leq \ve$ and $E$ is feasible for $D_{\min,\ve}(\sP\|\sE)$. Moreover, for any $\tau \in \sE$,
    \begin{align}
        \tr[E\tau] = \tr[\ket{1}\bra{1}\cF(\tau)] = \tr[\ket{1}\bra{1}\pi_{M_\tau}] = \frac{1}{M_\tau} \leq \frac{1}{M},
    \end{align}
    which implies $\sup_{\tau \in \sE} \tr[E\tau] \leq \frac{1}{M}$. Consequently
    \begin{align}
        D_{\min, \ve}(\sP\|\sE) \geq -\log \sup_{\tau \in \sE} \tr[E\tau] \geq \log M.
    \end{align}
    As this holds for any transformation $\cF\in\GPO$ achieving $(\sP, \sE) \xrightarrow[]{\cF, \ve} (\ket{1}\bra{1}, \Pi_M)$, taking the supremum over all achievable $M$ gives
    $ D_{\min, \ve}(\sP\|\sE) \geq \beta \widebar{W}_{\GPO,\ve}(\sP, \sE)$. This completes the proof.
\end{proof}

If $\sE$ is a singleton, the equilibrium state of the output battery is uniquely determined, and this knowledge effectively makes the dirty battery a clean one. In this case, Theorem~\ref{thm:one-shot dirty battery work extraction} also recovers the known characterizations for both the standard work extraction~\cite[Eq.~(50)]{gour2022role} and the recent black-box work extraction~\cite[Theorem~2]{watanabe2024black}. For a general equilibrium set $\sE$, the equilibrium uncertainty in the dirty battery removes the subspace constraint $E \perp V(\sE)$ that suppresses work extraction in the clean-battery setting (Theorem~\ref{thm:one-shot clean battery work extraction}), yielding a generically larger extractable work,
\begin{align}
    D_{\min,\ve}^{\sK}(\sP\|\sE) \leq D_{\min,\ve}(\sP\|\sE).
\end{align}
Moreover, this theorem gives the smoothed min-relative entropy between two sets of quantum states, previously studied as a purely mathematical quantity in~\cite{fang2025generalized}, a concrete operational interpretation in quantum thermodynamics, complementing the role of the max-relative entropy in work cost from a clean battery.

\subsection{Work of formation from a dirty battery}

We now consider the reverse task: preparing an uncertain athermality resource from a dirty battery with a worst-case certification of its work capacity. 

\begin{definition}[One-shot work cost from a dirty battery]
    \label{def:one-shot dirty battery work cost}
    Let $\ve \in [0,1)$, and let $(\sP, \sE)$ be an uncertain athermal state with $\sP,\sE\subseteq\density$. Given a class of free operations $\fF$, the one-shot work cost of preparing $(\sP, \sE)$ from a dirty battery is defined as
    \begin{align}\label{eq:one-shot dirty battery work cost}
        \beta \widebar{C}_{\fF,\ve}(\sP,\sE) := \log \inf_{\cF\in\fF} \left\{M:(\ket{1}\bra{1}, \Pi_M) \xrightarrow[]{\cF, \ve} (\sP, \sE)\right\}.
    \end{align}
\end{definition}

To characterize the work cost in this scenario, we introduce a constrained variant of the standard smoothed max-relative entropy.

\begin{definition}[Subspace-constrained max-relative entropy]
    \label{def:segment max relative entropy}
    Let $\ve \in [0,1)$, and let $\sK,\sP,\sE \subseteq \density$. The subspace-constrained max-relative entropy between $\sP$ and $\sE$ with respect to $\sK$ is defined as
    \begin{align}\label{eq:segment max relative entropy}
        D_{\max,\ve}^{\sK}(\sP\|\sE) := \log\inf_{\substack{\gamma \in \cl(\sK)\\\omega \in \sB_\ve(\sP)}} \{M>1: \sC(\gamma, \omega, 1/M) \subseteq \sE\},
    \end{align}
    where $\cl(\sK)$ denotes the closure of $\sK$ in $\density$, and $\sC(\gamma, \omega, 1/M) := \{(1-\lambda)\gamma + \lambda \omega : \lambda \in (0, 1/M]\}$ denotes the initial $1/M$ portion of the half-open segment from $\gamma$ (excluded) toward $\omega$.
\end{definition}

It is straightforward to verify that ${D}_{\max,\ve}^{\,\sK}(\sP\|\sE) \geq {D}_{\max,\ve}(\sP\|\sE)$ in general. When $\sE$ is convex and closed, the condition $\sC(\gamma, \omega, 1/M) \subseteq \sE$ is equivalent to requiring that the two endpoints of the corresponding closed segment lie in $\sE$, namely,
\begin{align}
    \gamma \in \sE, \qquad \tau:=(1-1/M)\gamma + (1/M)\omega \in \sE.
\end{align}
Indeed, $\tau \in \sE$ follows by taking the endpoint $\lambda = 1/M$, while $\gamma\in\sE$ follows from closedness by taking the limit $\lambda \to 0^+$ along the segment. Conversely, if both endpoints lie in $\sE$, then convexity implies that the whole segment between them lies in $\sE$. For any set $\sK\subseteq\density$, let $\cone(\sK):= \{z \rho: z \geq 0, \rho \in \sK\}$ denote the cone generated by $\sK$. Together with the constraint $\gamma \in \cl(\sK)$, the endpoint condition can be rewritten as $M\tau - \omega \in \cone(\cl(\sK)\cap \sE)$. Thus, in this case, the subspace-constrained max-relative entropy admits an equivalent form
\begin{align}
    D_{\max,\ve}^{\,\sK}(\sP\|\sE) = \log \inf_{\substack{\tau \in \sE\\\omega \in \sB_\ve(\sP)}}\left\{M > 1 : M\tau - \omega \in \cone(\cl(\sK)\cap \sE)\right\}.
\end{align}
This expression relates to the cone-restricted max-relative entropy in Ref.~\cite{george2024conerestricted}. 

Although the subspace-constrained max-relative entropy may appear artificial at first sight, it admits a direct operational interpretation as the one-shot work cost from a dirty battery under GPO.

\begin{restate}{Theorem}{7}[One-shot work cost from a dirty battery under GPO]
    \label{thm:one-shot dirty battery work cost}
    Let $\ve \in [0,1)$, and let $(\sP, \sE)$ be an uncertain athermal state with $\sP,\sE\subseteq\density$. The one-shot work cost of preparing $(\sP, \sE)$ from a dirty battery under $\GPO$ is given by
    \begin{align}
        \beta \widebar{C}_{\GPO,\ve}(\sP,\sE) = {D}_{\max,\ve}^{\,\sE}(\sP\|\sE).
    \end{align}
\end{restate}
\begin{proof}
    We show that the conversion $(\ket{1}\bra{1},\Pi_M) \xrightarrow[]{\cF,\;\ve}  (\sP,\sE)$ is achievable by some $\cF \in \GPO$ if and only if there exist $\omega\in\sB_{\ve}(\sP)$ and $\tau\in\cl(\sE)$ such that $\sC(\tau,\omega,1/M) \subseteq \sE$. The stated entropic characterization then follows by optimizing over all such $M$. \smallskip
    
    ($\Rightarrow$) Suppose that $(\ket{1}\bra{1}, \Pi_M) \xrightarrow[]{\cF,\;\ve} (\sP, \sE)$ is achieved by $\cF \in \GPO$. Define $\omega := \cF(\ket{1}\bra{1})$ and $\tau := \cF(\ket{0}\bra{0})$. By the approximation requirement on the nonequilibrium component, there exists $\rho\in\sP$ such that $T(\omega,\rho) \le \ve$, so $\omega \in \sB_\ve(\sP)$. Next, fix any $\lambda\in(0,1/M]$ and set $M^\prime = 1/\lambda$. Then $M^\prime \in[M,\infty)$ and $\pi_{M^\prime} \in \Pi_M$. The linearity of $\cF$ implies $\cF(\pi_{M^\prime}) = (1-\lambda)\tau + \lambda \omega$. By the exact Gibbs-preserving requirement for the target equilibrium set, we have $\cF(\pi_{M^\prime}) \in \sE$. Therefore,
    \begin{align}
        (1-\lambda)\tau + \lambda\omega = \cF(\pi_{M^\prime}) \in \sE.
    \end{align}
    Since this holds for every $\lambda \in (0,1/M]$, we have $\sC(\tau,\omega,1/M) \subseteq \sE$. Finally, taking $\lambda \to 0^+$ along the segment shows that $\tau$ is a limit of points in $\sE$, and hence $\tau \in \cl(\sE)$. \medskip

    ($\Leftarrow$) Suppose there exist $\omega \in \sB_\ve(\sP)$ and $\tau\in\cl(\sE)$ such that $\sC(\tau, \omega, 1/M) \subseteq \sE$. Since $\cl(\sE)\subseteq\density$, $\tau$ is a valid quantum state. Moreover, because $\omega \in \sB_\ve(\sP)$, there exists $\rho\in \sP$ with $T(\omega,\rho) \le \ve$. Define the measure-and-prepare channel $\cF(\cdot) := \tr[\ket{0}\bra{0} (\cdot)] \tau + \tr[\ket{1}\bra{1}(\cdot)] \omega$. Since $\cF(\ket{1}\bra{1}) = \omega$, we have $T(\cF(\ket{1}\bra{1}),\rho) \le \ve$. Thus the nonequilibrium component is prepared within error $\ve$. 
    
    It remains to verify the Gibbs-preserving condition on the uncertain equilibrium set. For every $M^\prime\in[M,\infty)$, let $\lambda := 1/M^\prime \in (0,1/M]$. Then
    \begin{align}
        \cF(\pi_{M^\prime}) = (1-\lambda)\tau + \lambda \omega \in \sC(\tau,\omega,1/M) \subseteq \sE.
    \end{align}
    Hence every candidate input Gibbs state $\pi_{M^\prime}\in\Pi_M$ is mapped into the target equilibrium set $\sE$. Combining this with the above approximation condition certifies that $(\ket{1}\bra{1}, \Pi_M) \xrightarrow[]{\cF,\;\ve} (\sP, \sE)$ is achievable by $\cF \in \GPO$. This proves the desired equivalence.
\end{proof}

\begin{remark}
    An analogous characterization can be obtained for work cost from a dirty battery under $\GPL$. In this case, the image $X:=\cL(\ket{1}\bra{1})$ need not be a valid quantum state. Nevertheless, all of the above argument remains valid upon replacing the set $\sB_\ve(\sP)$ with the $\ve$-ball of linear operators around $\sP$ defined as $\sB^{\sL}_{\ve}:=\{X\in \sL: T(X, \rho) \leq \ve\text{ for some }\rho\in\sP\}$.
\end{remark}

This result provides an exact characterization of the dirty-battery work cost in terms of the subspace-constrained max-relative entropy. In contrast to the standard smoothed max-relative entropy governing the clean-battery work cost in Theorem~\ref{thm:one-shot clean battery work cost}, this quantity depends explicitly on the geometry of the equilibrium uncertainty set $\sE$. This dependence reflects the additional restriction imposed by the uncertain equilibrium state of the battery during the formation process. 

\begin{figure}[H]
    \centering
    \includegraphics[width=0.55\linewidth]{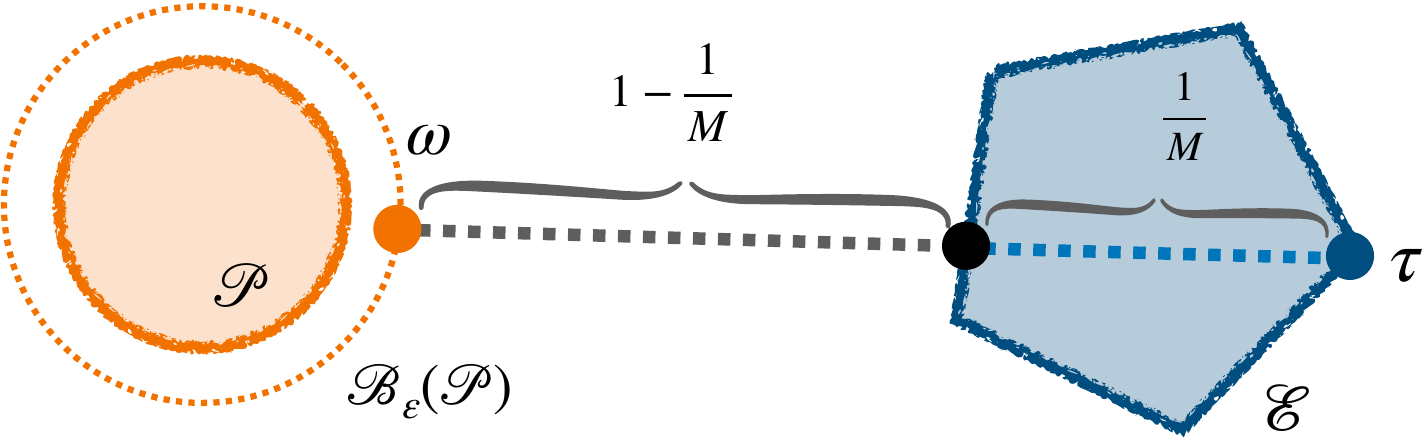}
    \caption{Geometric illustration of the condition for work formation from a dirty battery. The task is to find $\omega \in\sB_{\ve}(\sP)$ and $\tau\in\cl(\sE)$ such that the half-open line segment from $\tau$ (excluded) toward $\omega$ of length $1/M$ is entirely contained in $\sE$.}
    \label{fig:segment max relative entropy}
\end{figure}

Theorem~\ref{thm:one-shot dirty battery work cost} also admits a clear geometric interpretation. As illustrated in Fig.~\ref{fig:segment max relative entropy}, the formation $(\ket{1}\bra{1}, \Pi_M) \xrightarrow[]{\cF,\;\ve} (\sP, \sE)$ is achievable by $\cF\in\GPO$ if and only if there exist $\omega \in \sB_\ve(\sP)$ and $\tau \in \cl(\sE)$ such that the half-open line segment from $\tau$ (excluded) toward $\omega$ of length $1/M$, $\sC(\tau,\omega,1/M)$, lies entirely inside $\sE$. The work cost is then determined by the smallest value of $M$ for which such a segment exists. The same geometric characterization extends to $\GPL$ upon replacing the $\sB_\ve(\sP)$ by $\sB^{\sL}_{\ve}(\sP)$ defined above.

This geometric picture makes clear how the shape of the equilibrium uncertainty set affects the formation cost. Intuitively, the narrower the set $\sE$, the shorter the line segment that can be accommodated inside it, and hence the larger the work cost. In the extreme case where $\sE$ is a singleton lying outside the $\ve$-ball of $\sP$, no admissible pair $(\omega, \tau)$ exists, and the work cost diverges to infinity. We make this intuition more transparent in the following lower bound on work cost, stated in terms of the trace-distance separation $T(\sP, \sE) := \inf_{\rho \in \sP,\; \tau \in \sE} T(\rho, \tau)$ and the diameter $\diam(\sE) := \sup_{\tau, \tau^\prime \in \sE} T(\tau, \tau^\prime)$. \medskip

\begin{proposition}[Geometric lower bound on work cost from a dirty battery]
    \label{prop:lower bound diameter}
    Let $\ve \in [0,1)$, and let $(\sP, \sE)$ be an uncertain athermal state with $\sP,\sE\subseteq\density$. If $\sB_{\ve}(\sP)\cap \sE \neq \emptyset$, then $\beta \widebar{C}_{\TO,\ve}(\sP, \sE) = 0$, and the zero cost is achieved by a trivial thermal operation that discards the input and prepares a suitable Gibbs state in $\sE$. If $T(\sP, \sE) > \ve$, then
    \begin{align}
        \beta \widebar{C}_{\GPL,\ve}(\sP,\sE) \geq \log \frac{T(\sP, \sE) - \ve}{\diam(\sE)},
    \end{align}
    with the conventions that $\diam(\sE) = 0$ when $\sE$ is a singleton and $\log \infty = \infty$.
\end{proposition}
\begin{proof}
    Suppose there exists $\tau \in \sB_{\ve}(\sP)\cap \sE$. Then there is some $\rho\in\sP$ such that $T(\rho, \tau) \leq \ve$. Consider the replacer channel $\cF(\cdot):= \tr[\cdot] \tau$. For any $M>1$, we have
    \begin{align}
        (\ket{1}\bra{1}, \Pi_M) \xrightarrow[]{\cF, 0} (\tau, \tau) \xrightarrow[]{\id, \ve} (\rho, \tau) \in \sP \times \sE.
    \end{align}
    Thus the formation can be achieved with arbitrary $M>1$ through this channel. Taking the infimum over $M$ gives $\beta \widebar{C}_{\TO,\ve}(\sP, \sE) = 0$. \smallskip

    Now suppose that $T(\sP,\sE) > \ve$. It suffices to show that every finite $M$ feasible under GPL satisfies the claimed lower bound. By the proof of Theorem~\ref{thm:one-shot dirty battery work cost}, together with its GPL extension stated in the subsequent remark, feasibility of such an $M$ implies that there exist $\rho \in \sP$, $\tau \in \cl(\sE)$, and $X\in\sL$ such that $T(X,\rho) \leq \ve$ and $(1-\lambda)\tau + \lambda X \in \sE$ for all $\lambda \in (0,1/M]$. Taking $\lambda = 1/M$ gives $\tau_1 := (1-1/M)\tau + (1/M)X \in \sE$, or equivalently, $X = M\tau_1 - (M-1)\tau$. Writing $\rho - X = (\rho - \tau) - M(\tau_1 - \tau)$ and applying the reverse triangle inequality for the trace norm gives
    \begin{align}
        T(X,\rho) & = \frac{1}{2}\bigl\|(\rho - \tau) - M(\tau_1 - \tau)\bigr\|_1 \\
        & \geq \frac{1}{2}\|\rho - \tau\|_1 - \frac{M}{2}\|\tau_1 - \tau\|_1 \\
        & = T(\rho, \tau) - M \cdot T(\tau_1, \tau).
    \end{align}
    Since $T(X,\rho)\leq \ve$, it follows that $MT(\tau_1,\tau)\geq T(\rho,\tau) -\ve$.\smallskip
    
    If $T(\tau_1, \tau) > 0$, then $M \geq \frac{T(\rho, \tau) - \ve}{T(\tau_1, \tau)}$. Becasue $\tau\in\cl(\sE)$ and $\tau_1\in\sE\subseteq\cl(\sE)$, the continuity of trace distance implies
    \begin{align}
        T(\rho, \tau) \geq T(\sP,\cl(\sE)) = T(\sP, \sE)
    \end{align}
    and
    \begin{align}
        T(\tau_1, \tau) \leq \diam(\cl(\sE)) = \diam(\sE).
    \end{align}
    Hence,
    \begin{align}
        M \geq \frac{T(\sP, \sE) - \ve}{\diam(\sE)},
    \end{align}
    which gives the desired bound upon taking the logarithm. \smallskip

    It remains to rule out the case $T(\tau_1,\tau) = 0$. In this case $\tau_1 = \tau$. Since $M<\infty$ and $\tau_1 = (1-1/M)\tau + (1/M)X$, we must have $X=\tau$. Therefore,
    \begin{align}
        T(X,\rho)=T(\tau,\rho) \geq T(\sP,\cl(\sE)) = T(\sP,\sE) > \ve,
    \end{align}
    contradicting the feasibility condition $T(X,\rho) \leq \ve$. Hence no finite feasible $M$ can fall into this case. This completes the proof.
\end{proof}

\begin{remark}
    When $\sP$ and $\sE$ are closed, the overlap condition $\sB_{\ve}(\sP)\cap \sE \neq \emptyset$ in the first statement can be replaced by $T(\sP,\sE) \leq \ve$. Indeed, since $\sP,\sE$ are closed subsets of the finite-dimensional state space $\density$, they are compact. By continuity of the trace distance, the infimum defining $T(\sP,\sE)$ is attained by pair $(\rho, \tau) \in \sP\times\sE$. Hence $T(\sP,\sE) \leq \ve$ ensures the existence of such a pair with $T(\rho, \tau) \leq \ve$, and the replacer channel $\cF(\cdot) := \tr[\cdot]\tau$ then yields zero work cost.
\end{remark}

The above lower bound is determined entirely by the geometry of $\sP$ and $\sE$. The trace-distance separation $T(\sP, \sE)$ measures how far the nonequilibrium states lie from the equilibrium set, while the diameter $\diam(\sE)$ quantifies the size of that uncertainty set. If $T(\sP, \sE) > \ve$, the lower bound increases with the separation distance $T(\sP,\sE)$ and shrinks with the diameter $\diam(\sE)$. This matches the geometric intuition that formation becomes harder when the target states are farther from equilibrium, and also when the equilibrium uncertainty set is too narrow to accommodate the required line segment. In the limiting case where the target $(\sP,\sE)=(\ket{1}\bra{1}, \pi_N)$ is a clean battery, one has $\diam(\sE) = 0$. Hence the work cost diverges whenever the target is nontrivial, i.e., $T(\sP,\sE) = 1-1/N > \ve$, connecting the no-go theorem for battery energy truncation in Corollary~\ref{col:no-go truncation}.

\section{Asymptotic analysis and irreversibility}
\label{sec: asymptotic}

\subsection{Asymptotic analysis}

We now turn to the asymptotic regime, where we consider sequences of uncertain athermal states and investigate the optimal work extraction and formation rates under GPO. 

\begin{definition}[Asymptotic rates]
    \label{def:asymptotic rates}
    Let $\ve\in[0,1)$, and let $\{(\sP_n,\sE_n)\}_{n\in \NN}$ be a sequence of uncertain athermal states with $\sP_n,\sE_n\subseteq\density(\cH^{\ox n})$ for each $n\in\NN$. Whenever the following limits exist, we define the asymptotic work extraction and work cost rates of $\{(\sP_n, \sE_n)\}_{n\in\NN}$ under GPO as
    \begin{align}
        \beta \mathbb{W}_{\GPO, \ve}^\infty\left(\sP, \sE\right) & := \lim_{n\to\infty} \frac{1}{n}\beta \mathbb{W}_{\GPO, \ve}(\sP_n,\sE_n), \\
        \beta \mathbb{C}_{\GPO, \ve}^\infty\left(\sP, \sE\right) & := \lim_{n\to\infty} \frac{1}{n}\beta \mathbb{C}_{\GPO, \ve}(\sP_n,\sE_n),
    \end{align}
    where $\mathbb{W} \in \{W, \widebar{W}\}$ and $\mathbb{C} \in \{C, \widebar{C}\}$ denote the clean- and dirty-battery versions of work extraction and work of formation, respectively.
\end{definition}

For any $\rho\in \density$ and $\tau \in \PSD$, the Umegaki relative entropy is defined as~\cite{umegaki1954conditional}
\begin{align}\label{eq:Umegaki}
    D(\rho\|\tau):= \tr[\rho(\log \rho - \log \tau)],
\end{align}
if $\supp(\rho) \subseteq \supp(\tau)$ and $+\infty$ otherwise. For any sets $\sP, \sE \subseteq \density$, we define the relative entropy between these sets as
\begin{align}
    D(\sP\|\sE) := \inf_{\rho \in \sP, \tau \in \sE} D(\rho\|\tau).
\end{align}

The following theorem is a direct consequence of the one-shot entropic characterizations in Theorems~\ref{thm:one-shot clean battery work cost} and~\ref{thm:one-shot dirty battery work extraction}, together with the generalized asymptotic equipartition properties for max- and min-relative entropies between sets~\cite[Theorem 25]{fang2025generalized}. 

\begin{theorem}[Asymptotic rates]
    \label{thm:asymptotic approximate work extraction and cost}
    Let $\{(\sP_n,\sE_n)\}_{n\in \NN}$ be a sequence of uncertain athermal states with $\sP_n,\sE_n\subseteq\density(\cH^{\ox n})$ for each $n\in\NN$. Assume that the sequence $\{\sP_n\}_{n\in\NN}$ satisfies the following conditions: for each $n\in\NN$, the set $\sP_n$ is convex, compact, and permutation-invariant, and for all $m, k \in \NN$, $\sP_m \otimes \sP_k \subseteq \sP_{m+k}$ and $\polarPSD{(\sP_m)} \otimes \polarPSD{(\sP_k)} \subseteq \polarPSD{(\sP_{m+k})}$, where $\polarPSD{(\cdot)}$ denotes the polar set restricted to the positive semidefinite cone. Assume that the sequence $\{\sE_n\}_{n\in\NN}$ satisfies the same conditions. Suppose further that there exists a constant $c > 0$ such that $D_{\max}(\sP_n\|\sE_n) \leq cn$ for all $n \in \NN$. Then, for any $\ve \in (0, 1)$,
    \begin{align}
        \beta \widebar{W}^\infty_{\GPO, \ve}\left(\sP, \sE\right) & = \lim_{n\to \infty} \frac{1}{n} D(\sP_n \| \sE_n),\\
        \beta C_{\GPO, \ve}^\infty\left(\sP, \sE\right) & = \lim_{n\to \infty} \frac{1}{n} D(\sP_n \| \sE_n).
    \end{align}
\end{theorem}

The above result identifies the ultimate asymptotic rate of work extraction into a dirty battery in terms of the regularized Umegaki relative entropy between sets. A finer characterization concerns the convergence behavior below this rate: given an achievable extraction rate below this threshold, how fast does the optimal error vanish as the system size grows? This is captured by the error exponent for work extraction into a dirty battery defined as follows.

\begin{definition}[Error exponent for work extraction into a dirty battery]
    \label{def:error exponent for work extraction}
    Let $\{(\sP_n,\sE_n)\}_{n\in \NN}$ be a sequence of uncertain athermal states with $\sP_n,\sE_n\subseteq\density(\cH^{\ox n})$ for each $n\in\NN$. For a fixed rate $r$, the optimal error for work extraction into a dirty battery is defined as
    \begin{align}
        \widebar{\zeta}_{\GPO, r}(\sP_n,\sE_n) := \inf_{\cF_n\in\GPO} \left\{\ve \in [0,1]: (\sP_n,\sE_n) \xrightarrow[]{\cF_n,\ve} (\ket{1}\bra{1}, \Pi_{2^{nr}}) \right\}.
    \end{align}
    The corresponding error exponent at rate $r$ is then defined as
    \begin{align}
        \widebar{\zeta}_{\GPO, r}^\infty(\sP\|\sE) := \lim_{n \to \infty} -\frac{1}{n} \log \widebar{\zeta}_{\GPO, r}(\sP_n,\sE_n),
    \end{align}
    provided that the limit exists.
\end{definition}

To determine this exponent, we relate the optimal work-extraction error to the corresponding hypothesis-testing error. For hypothesis testing between $\sP_n\subseteq\density(\cH^{\ox n})$ and $\sE_n\subseteq\density(\cH^{\ox n})$ with type-II error threshold $2^{-nr}$, the optimal type-I error is defined as
\begin{align}
    \alpha_{n, r}(\sP_n\|\sE_n) := \inf_{0 \leq E_n \leq I } \left\{ \sup_{\rho_n \in \sP_n} \tr[(I - E_n)\rho_n]: \sup_{\tau_n \in \sE_n} \tr[E_n \tau_n] \leq 2^{-nr} \right\}.
\end{align}
The following proposition shows that this hypothesis-testing quantity coincides with the optimal error for dirty-battery work extraction at rate $r$. 

\begin{proposition}[Optimal error for work extraction at rate $r$]
    \label{prop:optimal error for work extraction at rate r}
    Let $(\sP_n,\sE_n)$ be an uncertain athermal state with $\sP_n, \sE_n\subseteq \density(\cH^{\ox n})$. For any rate $r > 0$,
    \begin{align}
        \widebar{\zeta}_{\GPO, r}(\sP_n,\sE_n) = \alpha_{n, r}(\sP_n\|\sE_n).
    \end{align}
\end{proposition}
\begin{proof}
    The proof is similar to that of Theorem~\ref{thm:one-shot dirty battery work extraction}. We prove the two inequalities separately. \smallskip

    $(\geq)$: Suppose that $\cF_n\in\GPO$ achieves $(\sP_n,\sE_n) \xrightarrow[]{\cF_n,\;\ve} (\ket{1}\bra{1}, \Pi_{2^{nr}})$ for some error $\ve \in [0,1]$. For every $\rho_n\in\sP_n$, the variational formula for the trace distance gives
    \begin{align}
        T(\cF_n(\rho_n), \ket{1}\bra{1}) = \sup_{0\leq F\leq I} \tr[F(\ket{1}\bra{1} - \cF_n(\rho_n))] \leq \ve.
    \end{align}
    Taking $F=\ket{1}\bra{1}$ gives $1-\tr[\ket{1}\bra{1}\cF_n(\rho_n)] \leq \ve$. Moreover, for every $\tau_n\in\sE_n$, the Gibbs-preserving condition for the target dirty battery implies that there exists some $M_{\tau_n} \geq 2^{nr}$ such that $\cF_n(\tau_n) = \pi_{M_{\tau_n}}$. Hence, $\tr[\ket{1}\bra{1}\cF_n(\tau_n)] = \frac{1}{M_{\tau_n}} \leq 2^{-nr}$. Now define $E_n := \cF_n^\dagger(\ket{1}\bra{1})$. Since $\cF_n$ is CPTP, we have $0\leq E_n \leq I$. For all $\rho_n\in\sP_n$,
    \begin{align}
        \tr[(I-E_n)\rho_n] = 1-\tr[E_n\rho_n] = 1-\tr[\ket{1}\bra{1}\cF_n(\rho_n)] \leq \ve,
    \end{align}
    and for all $\tau_n\in\sE_n$,
    \begin{align}
        \tr[E_n\tau_n] = \tr[\ket{1}\bra{1}\cF_n(\tau_n)] \leq 2^{-nr}.
    \end{align}
    Thus $E_n$ is feasible for $\alpha_{n,r}(\sP_n\|\sE_n)$ and $\alpha_{n,r}(\sP_n\|\sE_n) \leq \ve$. Taking the infimum over all feasible $(\cF_n,\ve)$ gives
    \begin{align}
        \widebar{\zeta}_{\GPO, r}(\sP_n,\sE_n) \geq \alpha_{n,r}(\sP_n\|\sE_n).
    \end{align}

    $(\leq)$: Conversely, let $0\leq E_n \leq I$ be any test operator satisfying $\sup_{\tau_n\in\sE_n}\tr[E_n\tau_n] \leq 2^{-nr}$. Set $\alpha_{E_n} := \sup_{\rho_n\in\sP_n} \tr[(I-E_n)\rho_n]$ and $t_n := 2^{-nr}$. For any $\eta \in(0,1)$, define $E_{n,\eta} := (1-\eta)E_n + \eta t_n I$. Since $n\geq 1$ and $r>0$, we have $t_n \in (0,1)$ and $ 0 \leq E_{n,\eta} \leq I$. Moreover, for every $\tau_n\in\sE_n$,
    \begin{align}
        0 < \eta t_n \leq \tr[E_{n,\eta} \tau_n] = (1-\eta)\tr[E_n\tau_n] + \eta t_n \leq t_n.
    \end{align}
    Now define the measure-and-prepare channel $\cE_{n,\eta}(\cdot) := \tr[(I-E_{n,\eta})(\cdot)]\ket{0}\bra{0} + \tr[E_{n,\eta}(\cdot)]\ket{1}\bra{1}$. For every $\rho_n\in\sP_n$, 
    \begin{align}
        T(\cE_{n,\eta}(\rho_n), \ket{1}\bra{1}) = \tr[(I-E_{n,\eta})\rho_n] & = (1-\eta)\tr[(I-E_n)\rho_n] + \eta(1-t_n) \\
        & \leq (1-\eta)\alpha_{E_n} + \eta(1-t_n).
    \end{align}
    For every $\tau_n\in\sE_n$,
    \begin{align}
        \cE_{n,\eta}(\tau_n) = \tr[(I-E_{n,\eta})\tau_n]\ket{0}\bra{0} + \tr[E_{n,\eta}\tau_n]\ket{1}\bra{1} = \pi_{M_{\tau_n,\eta}},
    \end{align}
    where $M_{\tau_n,\eta} := \frac{1}{\tr[E_{n,\eta}\tau_n]}$. Since $0 < \tr[E_{n,\eta}\tau_n] \leq t_n = 2^{-nr}$, we have $2^{nr} \leq M_{\tau_n,\eta} < \infty$. Thus $\cE_{n,\eta}(\tau_n) \in \Pi_{2^{nr}}$ for all $\tau_n \in \sE_n$. Hence $\cE_{n,\eta}\in\GPO$ achieves the conversion $(\sP_n,\sE_n) \xrightarrow[]{\cE_{n,\eta},\;\ve_{E_{n}, \eta}} (\ket{1}\bra{1}, \Pi_{2^{nr}})$ with error at most $\ve_{E_{n},\eta} := (1-\eta)\alpha_{E_n} + \eta(1-t_n)$. Therefore,
    \begin{align}
        \widebar{\zeta}_{\GPO, r}(\sP_n,\sE_n) \leq (1-\eta)\alpha_{E_n} + \eta(1-t_n).
    \end{align}
    Taking $\eta \to 0^+$ gives
    \begin{align}
        \widebar{\zeta}_{\GPO, r}(\sP_n,\sE_n) \leq \alpha_{E_n} = \sup_{\rho_n\in\sP_n} \tr[(I-E_n)\rho_n].
    \end{align}
    Finally, taking the infimum over all feasible tests $E_n$ yields
    \begin{align}
        \widebar{\zeta}_{\GPO, r}(\sP_n,\sE_n) \leq \alpha_{n,r}(\sP_n\|\sE_n).
    \end{align}
    Combining this with the opposite inequality completes the proof.
\end{proof}

For $n\in\NN$ and $r>0$, the Hoeffding divergence between two states $\rho_n, \tau_n\in\density(\cH^{\ox n})$ is defined as
\begin{align}
    H_{n,r}(\rho_n\|\tau_n) := \sup_{\alpha\in(0,1)} \frac{1}{\alpha}\left((\alpha-1)nr- \log\tr\left[\rho_n^\alpha\tau_n^{1-\alpha}\right]\right).
\end{align}
For two sets of states $\sP_n, \sE_n \subseteq \density(\cH^{\ox n})$, we define the Hoeffding divergence between these two sets as
\begin{align}
    H_{n,r}(\sP_n\|\sE_n) := \inf_{\substack{\rho_n\in\sP_n\\\tau_n\in\sE_n}} H_{n,r}(\rho_n\|\tau_n).
\end{align}
Combining Proposition~\ref{prop:optimal error for work extraction at rate r} with the generalized quantum Hoeffding's theorem~\cite[Theorem 5.1]{fang2025error}, we obtain the following characterization of the error exponent for work extraction into a dirty battery.

\begin{theorem}[Error exponent for work extraction into a dirty battery]
    Let $\{(\sP_n,\sE_n)\}_{n\in \NN}$ be a sequence of uncertain athermal states with $\sP_n,\sE_n\subseteq\density(\cH^{\ox n})$ for each $n\in\NN$. Assume that, for every $n\in\NN$, the sets $\sP_n$ and $\sE_n$ are convex and compact, and that the sequences satisfy $\sP_m\ox\sP_n \subseteq \sP_{m+n}$ and $\sE_m\ox\sE_n \subseteq \sE_{m+n}$ for all $m,n\in\NN$. Assume furthermore that there exists $\rho_1\in \sP_1$ and $\tau_1\in \sE_1$ with $\rho_1 \not\perp \tau_1$. Then for every rate $0 < r < \lim_{n\to \infty} \frac{1}{n} D(\sP_n \| \sE_n)$, the error exponent for work extraction into a dirty battery is given by
    \begin{align}
        \widebar{\zeta}^\infty_{\GPO, r}(\sP,\sE) = \lim_{n\to \infty} \frac{1}{n}H_{n,r}(\sP_n\|\sE_n).
    \end{align}
\end{theorem}

As discussed after Theorem~\ref{thm:one-shot dirty battery work extraction}, when $\sE$ is a singleton, the Gibbs-preserving condition uniquely determines the equilibrium state of the output battery, and the dirty battery effectively reduces to a clean one. In this case, our theorem yields the error exponent for the standard i.i.d.\ setting~\cite{gour2022role} and for the recently introduced black-box setting~\cite{watanabe2024black}, where the source uncertainty satisfies $\sP_m\ox\sP_n \subseteq \sP_{m+n}$ against a fixed i.i.d.\ equilibrium. The general statement further extends these results to the genuinely uncertain regime where neither the nonequilibrium state nor the equilibrium reference is known precisely, providing a unified Hoeffding-type formula.

\subsection{Irreversibility}

In this section, we provide detailed derivations of the asymptotic rates for the irreversibility example presented in the main text.

\begin{example}
Let $\ve < 1/2$ and $\delta > 0$. Consider an uncertain athermal state with precisely known nonequilibrium state $\sP = \{\ket{1}\bra{1}\}$ and uncertain equilibrium state $\sE = \left\{\pi_M: M \in [2, 2+\delta]\right\}$, which describes an excited dirty battery whose capacity is known only up to precision $\delta$. For $n \geq 2$ copies of this battery, the associated sequence of uncertain athermal states $\{(\sP_n, \sE_n)\}_{n\in\NN}$ is given by
\begin{align}
    \sP_n = \{\ket{1}\bra{1}^{\ox n}\},\qquad \sE_n = \left\{\pi_M^{\ox n}: M \in [2, 2+\delta]\right\}.
\end{align}
\begin{itemize}
    \item In the clean-battery setting, the one-shot extractable work and work cost are
    \begin{align}
        \beta W_{\GPO,\ve}(\sP_n, \sE_n) &= - \log(1-\ve), \label{eq:W GPO irre} \\
        \beta C_{\GPO,\ve}(\sP_n, \sE_n)& = n + \log(1-\ve). \label{eq:C GPO irre}
    \end{align}
    Dividing by $n$ and taking the asymptotic limit gives
    \begin{align}
        \beta W_{\GPO,\ve}^\infty(\sP, \sE) & = 0 < 1 = \beta C_{\GPO,\ve}^\infty(\sP, \sE).
    \end{align}
    \item In the dirty-battery setting, the corresponding one-shot quantities are
    \begin{align}
        \beta \widebar{W}_{\GPO,\ve}(\sP_n, \sE_n) &= n - \log(1-\ve), \label{eq:bar W GPO irre} \\
        \beta \widebar{C}_{\GPO,\ve}(\sP_n, \sE_n) & = \infty. \label{eq:bar C GPO irre}
    \end{align}
    Hence the asymptotic rates become
    \begin{align}
        \beta \widebar{W}_{\GPO,\ve}^\infty(\sP, \sE) & = 1 < \infty = \beta \widebar{C}_{\GPO,\ve}^\infty(\sP, \sE).                                
    \end{align}
\end{itemize}
\end{example}             

\begin{proof}[Proof of Eq.~\eqref{eq:W GPO irre}]
We first verify the geometric condition $\conv(\sP_n)\cap\aff(\sE_n)\neq\emptyset$. 
For convenience, write $\rho_p := (1-p)\ket{0}\bra{0} + p\ket{1}\bra{1}$. Because $\pi_M = \rho_{1/M}$ and $M\in[2,2+\delta]$ is equivalent to $p\in [1/(2+\delta),\, 1/2]$, the equilibrium set can be equivalently written as
\begin{align}
    \sE_n = \{\rho_p^{\ox n}: p \in [1/(2+\delta),\, 1/2]\}.
\end{align}

Consider the map $p \mapsto \rho_p^{\ox n}$. Because $\rho_p$ is diagonal in the computational basis, its tensor power $\rho_p^{\ox n}$ is diagonal in the $n$-qubit computational basis. For any bit string $x = x_1 \cdots x_n \in \{0,1\}^n$ with Hamming weight $|x| = \sum_i x_i$, the corresponding diagonal entry is
\begin{align}
    \bra{x}\rho_p^{\ox n}\ket{x} = \prod_{i=1}^n \bra{x_i}\rho_p\ket{x_i} = p^{|x|}(1-p)^{n-|x|},
\end{align}
which is a (scalar) polynomial in $p$ of degree at most $n$. Hence, the map $p \mapsto \rho_p^{\ox n}$ is an \emph{operator-valued polynomial} of degree $n$ and can be written as $\rho_p^{\ox n} = \sum_{k=0}^n A_k\, p^k$, where each $A_k$ is a fixed $2^n \times 2^n$ Hermitian matrix independent of $p$.
                  
Recall that a scalar polynomial of degree $n$ is uniquely determined by its values at $n+1$ distinct points. The same property holds for operator-valued polynomials, since two operators are equal if and only if all of their matrix entries agree. Choosing $n+1$ distinct points $p_0, p_1, \ldots, p_n \in [1/(2+\delta),\, 1/2]$ and applying the Lagrange interpolation formula entry-wise, we obtain
\begin{align}
    \rho_p^{\ox n} = \sum_{i=0}^n \ell_i(p)\, \rho_{p_i}^{\ox n},
\end{align}
where $\ell_i(p) = \prod_{j \neq i} \frac{p - p_j}{p_i - p_j}$ are the Lagrange basis polynomials. This identity holds because both sides are operator-valued polynomials of degree $n$ in $p$ and agree at the $n+1$ points $p_0, \ldots, p_n$. 

Moreover, by construction, we have $\rho_{p_i}^{\ox n} = \pi_{1/p_i}^{\ox n} \in \sE_n$ for each $i=0,\cdots,n$. Since the Lagrange basis polynomials satisfy $\sum_{i=0}^n \ell_i(p) = 1$ for all $p$, the above representation is an affine combination of elements of $\sE_n$. Evaluating at $p = 1$ gives
\begin{align}
    \ket{1}\bra{1}^{\ox n} = \rho_1^{\ox n} = \sum_{i=0}^n \ell_i(1)\, \rho_{p_i}^{\ox n}, \qquad \sum_{i=0}^n \ell_i(1) = 1.
\end{align}
Hence $\ket{1}\bra{1}^{\ox n} \in \aff(\sE_n)$. In particular, since $\conv(\sP_n) \cap \aff(\sE_n) \neq \emptyset$, the asserted result therefore follows from Corollary~\ref{cor:no-go for work extraction to clean battery}.
\end{proof}

\begin{proof}[Proof of Eq.~\eqref{eq:C GPO irre}]
By Theorem~\ref{thm:one-shot clean battery work cost}, we have $\beta C_{\GPO,\ve}(\sP_n, \sE_n) = D_{\max,\ve}(\sP_n \| \sE_n)$. Since $\sP_n = \{\rho_n\}$, where $\rho_n := \ket{1}\bra{1}^{\ox n}$, it follows that
\begin{align}
    D_{\max,\ve}(\sP_n \| \sE_n) = \inf_{M \in [2,2+\delta]} D_{\max,\ve}(\rho_n \| \pi_M^{\ox n}).
\end{align}
Fix $M \in [2,2+\delta]$. For each $x\in\{0,1\}^n$, write $\sigma_x := \bra{x}\pi_M^{\ox n}\ket{x} = (1/M)^{|x|}(1-1/M)^{n-|x|}$. Since both $\rho_n$ and $\pi_M^{\ox n}$ are diagonal in the computational basis, dephasing any candidate smoothing state $\tilde{\rho}$ cannot increase either $T(\tilde{\rho}, \rho_n)$ or $D_{\max}(\tilde{\rho} \| \pi_M^{\ox n})$. Therefore, it suffices to optimize over diagonal states $\tilde{\rho} = \sum_{x} q_x\ket{x}\bra{x}$ satisfying $T(\tilde{\rho}, \rho_n) \leq \ve$. Because $\rho_{n}=\ket{1}\bra{1}^{\ox n}$, this is equivalent to $q_{1^n} \geq 1 - \ve$. Hence, in this setting, 
\begin{align}
    D_{\max,\ve}(\rho_n\|\pi_M^{\ox n}) = \log\inf \Big\{\lambda>1: \exists\, (q_x)_x\text{ with } q_{1^n}\geq 1-\ve \text{ and } q_x \le \lambda \sigma_x \text{ for all } x \Big\}.
\end{align}
We now evaluate this quantity by matching lower and upper bounds. \smallskip

\emph{Lower bound.} Let $(q_x)_x$ be feasible. Then $1-\ve \leq q_{1^n} \leq \lambda\,\sigma_{1^n} = \lambda\, M^{-n}$, and hence $\lambda \geq (1-\ve)\,M^n$. \smallskip

\emph{Upper bound.} Set $\lambda^* := (1-\ve)\,M^n$, and choose $q_{1^n} = 1 - \ve$, $q_{0^n} = \ve$, and $q_x = 0$ for all other $x$. We verify that $(q_x)_x$ is feasible. First, $q_{1^n} = 1 - \ve = \lambda^*\, M^{-n} = \lambda^* \sigma_{1^n}$. 
Next, $q_{0^n} = \ve \leq (1-\ve)(M-1)^n = \lambda^*\,(1-1/M)^n = \lambda^* \sigma_{0^n}$ where we used $\ve < 1/2 < 1-\ve$ and $(M-1)^n \geq 1$. Finally, for all other $x$, $q_x = 0\leq \lambda^* \sigma_x$. so the constraints are satisfied. \smallskip

Combining both bounds gives
\begin{align}
    D_{\max,\ve}(\rho_n\|\pi_M^{\ox n}) = n\log M + \log(1-\ve).
\end{align}
Since this quantity is increasing in $M$, the infimum over $M\in [2, 2+\delta]$ is attained at $M = 2$. Hence,
\begin{align}
    \beta C_{\GPO,\ve}(\sP_n, \sE_n) =D_{\max,\ve}(\sP_n\|\sE_n) = n + \log(1-\ve).
\end{align}
\end{proof}

\begin{proof}[Proof of Eq.~\eqref{eq:bar W GPO irre}]
By~\cite[Lemma~31]{fang2025generalized}, we have the identity
\begin{align}
    D_{\min,\ve}(\sP_n\|\sE_n) = D_{\min,\ve}(\conv(\sP_n) \| \conv(\sE_n)) = \inf_{\substack{\rho_n \in \conv(\sP_n)\\ \tau_n \in \conv(\sE_n)}} D_{\min,\ve}(\rho\|\tau).
\end{align}
Since $\sP_n=\{\ket{1}\bra{1}^{\ox n}\}$, Theorem~\ref{thm:one-shot dirty battery work extraction} gives
\begin{align} 
    \beta \widebar{W}_{\GPO,\ve}(\sP_n, \sE_n) = D_{\min,\ve}(\sP_n \| \sE_n)=\inf_{\substack{\tau_n \in \conv(\sE_n)}} D_{\min,\ve}(\ket{1}\bra{1}^{\ox n}\|\tau_n).
\end{align}
We now evaluate the last quantity for a fixed $\tau_n \in \conv(\sE_n)$ by matching lower and upper bounds. Since every element of $\sE_n$ is diagonal in the computational basis, every $\tau_n \in \conv(\sE_n)$ is also diagonal in the computational basis. \smallskip

\emph{Lower bound.} Let $E$ be any feasible test operator. Then $\tr[(I-E)\ket{1}\bra{1}^{\ox n}] \leq \ve$, which is equivalent to $\bra{1^n}E\ket{1^n}\geq 1-\ve$. Since $\tau_n$ is diagonal in the computational basis and $E\geq 0$, we have
\begin{align}
    \tr[E\tau_n] = \sum_{x\in\{0,1\}^n} \bra{x}\tau_n\ket{x}\bra{x} E\ket{x} \geq \bra{1^n} \tau_n \ket{1^n}\bra{1^n} E \ket{1^n} \geq (1-\ve)\bra{1^n}\tau_n\ket{1^n}.
\end{align}

\emph{Upper bound.} Choose $E = (1-\ve)\ket{1}\bra{1}^{\ox n}$. Then $0\leq E\leq I$ and $\tr[(I-E)\ket{1}\bra{1}^{\ox n}] = \ve$, so $E$ is feasible. Moreover, $\tr[E\tau_n] = (1-\ve)\bra{1^n}\tau_n\ket{1^n}$, which implies that the lower bound is tight. \smallskip

Combining both bounds gives
\begin{align}\label{eq:bar W proof tmp}
    D_{\min,\ve}(\ket{1}\bra{1}^{\ox n} \| \tau_n) = -\log\bigl((1-\ve)\bra{1^n}\tau_n\ket{1^n}\bigr).
\end{align}
Finally, the right-hand side of Eq.~\eqref{eq:bar W proof tmp} is minimized over $\conv(\sE_n)$ at $\tau_n = \pi_2^{\ox n}$, for which $\bra{1^n}\pi_2^{\ox n}\ket{1^n} = 2^{-n}$. Therefore, 
\begin{align}
    \beta \widebar{W}_{\GPO,\ve}(\sP_n, \sE_n) = n - \log(1-\ve).
\end{align}
\end{proof}

\begin{proof}[Proof of Eq.~\eqref{eq:bar C GPO irre}]
The set $\sE_n = \{\pi_M^{\ox n} : M \in [2, 2+\delta]\}$ is the continuous image of a compact interval, and hence is compact. In particular, it is closed, so $\cl(\sE_n) = \sE_n$. By Theorem~\ref{thm:one-shot dirty battery work cost}, any finite work cost $\beta\widebar{C}_{\GPO,\ve}(\sP_n,\sE_n) < \infty$ requires the existence of $\omega_n \in \sB_\ve(\sP_n)$, $\tau_n \in \cl(\sE_n) = \sE_n$, and $M < \infty$ such that the half-open line segment
\begin{align}
    \sC(\tau_n, \omega_n, 1/M) = \{(1-p)\tau_n + p\,\omega_n : p \in (0, 1/M]\} \subseteq \sE_n.
\end{align}
We show that no such triple exists for $n \geq 2$. \smallskip

\emph{Step 1: $\sE_n$ contains no non-degenerate line segment anchored at any point $\tau_n\in\sE_n$ for $n \geq 2$.} Suppose, for contradiction, that there exist a state $\tau_n = \pi_{M_0}^{\ox n} \in \sE_n$, an operator $\omega_n \neq \tau_n$, and a constant $\lambda > 0$ such that $(1-p)\tau_n + p\,\omega_n \in\sE_n$ for all $p \in (0,\lambda]$. Equivalently, for each $p \in (0,\lambda]$, there exists $M(p) \in [2, 2+\delta]$ such that
\begin{align}\label{eq:line_seg_assumption}
    (1-p)\,\pi_{M_0}^{\ox n} + p\,\omega_n = \pi_{M(p)}^{\ox n}.
\end{align}
Since both $\pi_{M_0}^{\ox n}$ and $\pi_{M(p)}^{\ox n}$ are diagonal in the computational basis, Eq.~\eqref{eq:line_seg_assumption} implies that $\omega_n$ is diagonal as well. Let $t(p) := 1/M(p)$ and $t_0 := 1/M_0$. Then $t(p) \in [1/(2+\delta),\, 1/2]$ for $p\in(0,\lambda]$ and $t_0\in[1/(2+\delta),\,1/2]$. For each bit string $x \in \{0,1\}^n$ with Hamming weight $|x|$, we have
\begin{align}
    \bra{x}\pi_{M(p)}^{\ox n}\ket{x} = t(p)^{|x|}\,(1-t(p))^{n-|x|}.
\end{align}
Evaluating Eq.~\eqref{eq:line_seg_assumption} at $x = 1^n$ gives
\begin{align}\label{eq:tn_affine}
    t(p)^n = (1-p)\,t_0^n + p\,\bra{1^n}\omega_n\ket{1^n},
\end{align}
so $t(p)^n$ is affine in $p$ on $(0,\lambda]$. Next, evaluating Eq.~\eqref{eq:line_seg_assumption} at any string $x$ with $|x| = n-1$ yields
\begin{align}\label{eq:tn1_intermediate}
    t(p)^{n-1}(1 - t(p)) = (1-p)\,t_0^{n-1}(1-t_0) + p\,\bra{x}\omega_n\ket{x},
\end{align}
which is also affine in $p$. Since $t(p)^{n-1}(1-t(p)) = t(p)^{n-1} - t(p)^n$ and $t(p)^n$ is affine in $p$ by Eq.~\eqref{eq:tn_affine}, it follows that $t(p)^{n-1}$ is affine in $p$ as well. We may therefore write
\begin{align}\label{eq:ab_cd}
    t(p)^n = a + bp, \qquad t(p)^{n-1} = c + dp,
\end{align} 
where $a = t_0^n$, $c = t_0^{n-1}$ (by continuity at $p\to 0^+$), and $b, d$ are constants determined by the preceding relations. These two affine functions must satisfy $[t(p)^{n-1}]^n = [t(p)^n]^{n-1}$, i.e.,
\begin{align}\label{eq:poly_identity}
    (c + dp)^n = (a + bp)^{n-1}
\end{align}
for all $p \in (0,\lambda]$. Since both sides are polynomials in $p$ and agree on an interval, Eq~\eqref{eq:poly_identity} holds as a polynomial identity. Comparing the coefficient of $p^n$ gives $d^n = 0$, so $d = 0$. Substituting this back into Eq.~\eqref{eq:poly_identity} gives $c^n = (a + bp)^{n-1}$. Comparing the coefficient of $p^{n-1}$ yields $b^{n-1} = 0$, hence $b = 0$. Therefore both $t(p)^n = a$ and $t(p)^{n-1} = c$ are constant. Since $t(p) > 0$, we have $t(p) = a/c = t_0$ for all $p \in (0,\lambda]$. Therefore, $M(p) = 1/t(p) = 1/t_0 = M_0$ for all $p\in(0,\lambda]$, and Eq.~\eqref{eq:line_seg_assumption} reduces to
\begin{align}
    (1-p)\,\pi_{M_0}^{\ox n} + p\,\omega_n = \pi_{M_0}^{\ox n}.
\end{align}
Since $p >0$, this forces $\omega_n = \pi_{M_0}^{\ox n} = \tau_n$, contradicting the assumption that the segment is non-degenerate. \medskip

\emph{Step 2: The degenerate case violates the smoothing constraint.}
By Step~1, if $\sC(\tau_n,\omega_n, 1/M)\subseteq \sE_n$ for some $\tau_n \in \sE_n$, $\omega_n$, and finite $M$, then necessarily $\omega_n = \tau_n = \pi_{M_0}^{\ox n}$ for some $M_0 \in [2, 2+\delta]$. On the other hand, Theorem~\ref{thm:one-shot dirty battery work cost} also requires $\omega_n \in \sB_{\ve}(\sP_n)$, that is, $T(\omega_n, \ket{1}\bra{1}^{\ox n}) \leq \ve$. We now show that this is impossible. By the variational characterization of trace distance,
\begin{align}
    T(\pi_{M_0}^{\ox n}, \ket{1}\bra{1}^{\ox n}) & = \sup_{0\leq F\leq I} \tr[F(\ket{1}\bra{1}^{\ox n} -\pi_{M_0}^{\ox n})].
\end{align}
Taking the test $F = \ket{1}\bra{1}^{\ox n}$ gives
\begin{align}
    T(\pi_{M_0}^{\ox n}, \ket{1}\bra{1}^{\ox n}) \geq 1- \bra{1^n}\pi_{M_0}^{\ox n}\ket{1^n} = 1 - M_0^{-n} \geq 1-2^{-n} \geq \frac{3}{4} > \frac{1}{2} > \ve,
\end{align}
where we used $M_0 \geq 2$, $n\geq 2$ and the assumption $\ve < 1/2$. Hence $\omega_n = \pi_{M_0}^{\ox n} \notin \sB_\ve(\sP_n)$, which contradicts the required smoothing condition. \smallskip

Combining Steps~1 and~2, no triple $(\omega_n,\tau_n, M)$ satisfying the condition of Theorem~\ref{thm:one-shot dirty battery work cost} can exist. Consequently, $\beta \widebar{C}_{\GPO,\ve}(\sP_n, \sE_n) = \infty$ for all $n\geq 2$.
\end{proof}

\end{document}